\documentclass[10pt]{amsart}
\usepackage{amsthm,amsfonts,amssymb}
\usepackage[colorlinks,linkcolor=black,citecolor=blue]{hyperref}

\usepackage{tikz}
\usetikzlibrary{shapes}

\def\a{{\alpha}}
\def\b{{\beta}}
\def\be{{\beta}}

\def\ga{\gamma}

\def\de{\delta}

\def\si{\sigma}

\def\om{\omega}
\def\Om{\Omega}
\def\th{\theta}
\def\ze{\zeta}

\def\varep{\varepsilon}

\def\pr{{\partial}}

\def\les{\lesssim}

\def\rh{{\rho}}
\def\et{{\eta}}

\newcommand{\ab}{{\underline{\alpha}}}
\newcommand{\beb}{{\underline{\beta}}}

\def\MMf{{\mathfrak{M}}}
\def\xdmf{{\dot{\mathfrak{X}}}}

\def\XX{{\mathcal{X}}}
\def\YY{{\mathcal{Y}}}
\def\ZZ{{\mathcal{Z}}}

\def\gdcd{{\dot{\gd_c}}}
\def\Omd{{\dot{\Om}}}

\def\etabd{{\dot{\etab}}}
\def\etad{{\dot{\eta}}}

\def\Kd{{\dot{K}}}
\def\Ldo{{\overset{\circ}{\Ld}}}
\def\phid{{\dot{\phi}}}
\def\omd{{\dot{\om}}}
\def\abd{{\dot{\ab}}}

\def\ombd{{\dot{\omb}}}
\def\ad{{\dot{\alpha}}}
\def\AA{{\mathcal A}}

\def\BB{{\mathcal B}}

\def\CC{{\mathcal C}}

\def\MM{{\mathcal M}}

\def\FF{{\mathcal F}}

\def\HH{{\mathcal H}}

\def\GG{{\mathcal G}}

\def\OO{{\mathcal O}}
\def\SS{{\mathcal S}}

\def\DD{{\mathcal D}}
\def\PP{{\mathcal P}}

\def\QQ{{\mathcal Q}}

\def\HHb{\underline{\mathcal H}}

\def\D{{\bf D}}

\def\g{{\bf g}}

\def\t{{\bf t}}

\def\SSS{{\mathbb{S}}}

\def\xd{{\dot{x}}}
\def\mfc{{\dot{\mathfrak{c}}}}
\def\mfm{\mathfrak{m}}

\def\mfx{\mathfrak{X}}
\def\mfxd{\dot{\mathfrak{X}}}

\DeclareMathOperator{\Div}{\mathrm{div}}
\DeclareMathOperator*{\Curl}{\mathrm{curl}}

\newcommand{\pd}{\pd \mkern-9mu/\ \mkern-7mu}
\newcommand{\Lied}{\mathcal{L} \mkern-9mu/\ \mkern-7mu}
\newcommand{\DDd}{\DD \mkern-10mu /\ \mkern-5mu}
\newcommand{\Db}{\underline{D}}
\newcommand{\RRRic}{\mathrm{Ric}}

\newcommand{\Divd}{\Div \mkern-17mu /\ }
\newcommand{\Divdo}{{\overset{\circ}{\Div \mkern-17mu /\ }}}
\newcommand{\Curld}{\Curl \mkern-17mu /\ }

\newcommand{\Nabs}{\nabla \mkern-13mu /\ }
\newcommand{\Ld}{\triangle \mkern-12mu /\ }
\newcommand{\iin}{\in \mkern-16mu /\ \mkern-5mu}
\newcommand{\gd}{{g \mkern-8mu /\ \mkern-5mu }}

\newcommand{\di}{\mbox{$d \mkern-9.2mu /$\,}}

\newcommand{\LIE}{\mathcal{L} \mkern-11mu /\  \mkern-3mu }

\newcommand{\Dh}{\widehat{D}}
\newcommand{\Dbh}{\widehat{\Db}}

\newcommand{\trchi}{{\tr \chi}}
\newcommand{\trchib}{{\tr \chib}}

\newcommand{\chihd}{{\dot{\chih}}}
\newcommand{\chibhd}{{\dot{\chibh}}}
\newcommand{\omtrchid}{\dot{(\Om \tr \chi)}}
\newcommand{\omtrchibd}{\dot{(\Om \tr \chib)}}

\def\ni{\noindent}

\def\Lb{{\,\underline{L}}}

\def\tr{\mathrm{tr}}

\def\chih{{\widehat \chi}}
\def\chib{{\underline \chi}}
\def\chibh{{\underline{\chih}}}

\def\etab{{\underline \eta}}
\def\omb{{\underline{\om}}}
\def\bb{{\underline{\b}}}
\def\aa{{\underline{\a}}}

\def\th{{\theta}}

\def\SSt{\tilde\SS}
\def\ut{\tilde u}
\def\vt{\tilde v}
\def\tht{\tilde \th}
\def\Lt{\tilde L}
\def\Lbt{\tilde \Lb}
\def\Omt{\tilde\Om}
\def\Omt{\tilde\Om}

\def\Rbf{{\mathbf{R}}}

\def\Gammad{{\Gamma \mkern-11mu /\,}}
\newcommand{\gac}{{\overset{\circ}{\ga}}}

\newtheorem{theorem}{Theorem}[section]
\newtheorem{lemma}[theorem]{Lemma}
\newtheorem{proposition}[theorem]{Proposition}

\newtheorem{definition}[theorem]{Definition}
\newtheorem{remark}[theorem]{Remark}

\numberwithin{equation}{section}

\begin{document}

\title[The $C^3$-null gluing problem]{The $C^3$-null gluing problem:\\ Linear and Nonlinear analysis}
\author{Robert Sansom} 
\address{Queen Mary University of London, Mile End Road, London, E1 4NS}
\email{r.sansom@qmul.ac.uk}

\begin{abstract}
In this paper, we investigate the $C^3$-null gluing problem for the Einstein vacuum equations, that is, we consider the null gluing of up to and including third-order derivatives of the metric. In the regime where the characteristic data is close to Minkowski data, we show that this $C^3$-null gluing problem is solvable up to a $20$-dimensional space of obstructions. The obstructions correspond to $20$ linearly conserved quantities: $10$ of which are already present in the $C^2$-null gluing problem analysed by Aretakis, Czimek and Rodnianski \cite{ACR1,ACR2,ACR3}, and $10$ are novel obstructions inherent to the $C^3$-null gluing problem. The $10$ novel obstructions are linearly conserved charges calculated from third-order derivatives of the metric.
\end{abstract}

\maketitle
\tableofcontents

\section{Introduction}\label{sec:intro}

\ni The Einstein Equations for a $4$-dimensional Lorentzian manifold $(\MM, \g)$ are
\begin{equation}\label{eq:EVE}
\RRRic(\g)=0,
\end{equation}
where $\RRRic$ denotes the Ricci tensor of $\g$. The trivial solution to \eqref{eq:EVE} is the \emph{Minkowski spacetime}, with $\MM \cong \mathbb{R}^4$ and $$g=-dt^2+dx^2+dy^2+dz^2$$ expressed in Cartesian coordinates. The Einstein equations \eqref{eq:EVE} can be cast as a hyperbolic system of equations and hence admit well-posed Cauchy problems for suitable initial data \cite{CB}, \cite{CBG}. The \emph{spacelike} Cauchy problem consists of initial data given by a triple $(\Sigma,g,k)$ where $(\Sigma, g)$ is a $3$-dimensional Riemannian manifold $k$ is a symmetric $2$-tensor, see Figure \ref{fig:IVP}$(a)$. 

On the other hand, the \emph{characteristic} Cauchy problem consists of initial data posed on two transversally intersecting \emph{null hypersurfaces}, an outgoing null hypersurface, $\HH$, and an ingoing null hypersurface, $\HHb$. The defining property of $\HH$ and $\HHb$ is that the metric $\g$ is degenerate along them. The null hypersurfaces intersect at, and are foliated by, compact $2$-dimensional objects. For the purposes of this paper, we will consider them to be spheres denoted by $S$, see Figure \ref{fig:IVP}$(b)$. The well-posedness of the characteristic initial value problem was established in \cite{rendall, Lukexistence}. 

\begin{figure}[h]
 \begin{center}
 \begin{tikzpicture}
\draw[darkgray, thick] [domain=-2.5:2.5,smooth,variable=\t]
plot ({\t},{.25*sin(\t*1.25 r)}) ;
\draw[darkgray, thick] [domain=-2.5:2.5,smooth,variable=\t]
plot ({1.5+\t},{1.5+.25*sin(\t*1.25 r)}) ;
\draw(-2.5,0)--(-1,1.5);
\draw(2.5,0)--(4,1.5);
\draw(0,0) node[node font =\small]{$$};
\draw(.5,-0.3) node[node font =\small]{$(a)$};
\draw(.5,.75) node[node font =\small]{$(\Sigma,g,k)$};
\draw[->](-.75,.5) --(-.75,1.8);
\draw(-.4,1.7) node[node font =\small]{$T$};
\end{tikzpicture}
\qquad
\begin{tikzpicture}
\draw (0,0) ellipse (1.5 and 0.2);
\draw (1.5,0) -- (3,1.5);
\draw (-1.5,0) -- (-3,1.5);
\draw (3,1.5) -- (3.75,2.25);
\draw (-3,1.5) -- (-3.75,2.25);
\draw (1.5,0)--(.5,1);
\draw (-1.5,0)--(-.5,1);
\draw(2.8,.7) node[node font =\small]{$\HH$};
\draw(0.75,.4) node[node font =\small]{$\HHb$};
\draw(0,-.4) node[node font =\small]{$S$};
\draw(0,-1) node[node font =\small]{$(b)$};
\end{tikzpicture}
\caption{$(a)$ Spacelike initial data for \eqref{eq:EVE}. $(b)$ Setup of the characteristic initial value problem with two null hypersurfaces intersecting at spheres.}\label{fig:IVP}
\end{center}
\end{figure}
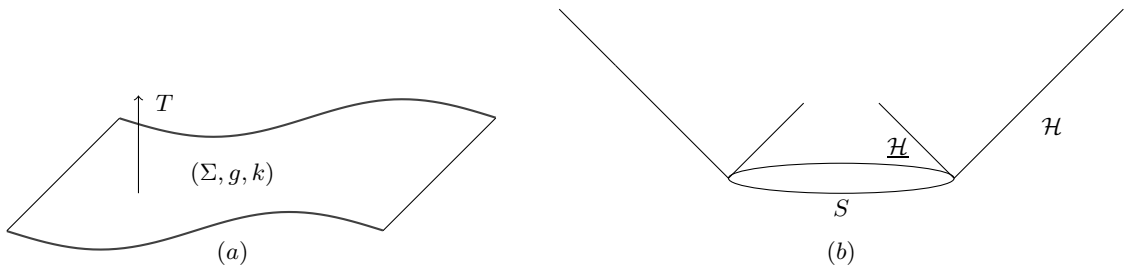

\ni Both spacelike and characteristic initial data cannot be prescribed freely and must satisfy constraints. In the spacelike case, the constraints are elliptic equations. In the characteristic case, the constraints are transport equations along the null generators of $\HH$ and $\HHb$ called \emph{null structure equations}.

Of wide interest is the rigidity and flexibility of solutions to \eqref{eq:EVE}. Investigating the flexibility of solutions to \eqref{eq:EVE} is primarily achieved by studying gluing problems. The general gluing problem of General Relativity asks to construct a solution $(\MM,\g)$ to \eqref{eq:EVE} such that two given solutions $(\MM_\mathcal{A},\g_\mathcal{A})$ and $(\MM_\mathcal{B},\g_\mathcal{B})$ isometrically embed into $(\MM,\g)$ with the aim being to connect the solutions $(\MM_\mathcal{A},\g_\mathcal{A})$ and $(\MM_\mathcal{B},\g_\mathcal{B})$ across a gluing region. Investigating gluing problems at this level is challenging and they become more tractable by considering the gluing of initial data sets of the same type, i.e. both spacelike or both characteristic, to \eqref{eq:EVE}. The gluing problem for characteristic initial data is shown in Figure \ref{fig:IDgluing}. Here,  $(\HH_\AA, \HHb_\AA)$ and $(\HH_\BB, \HHb_\BB)$ are two pairs of transversal null hypersurfaces intersecting at spheres $S_\AA$ and $S_\BB$, respectively. The null gluing problem for characteristic initial data is to construct a solution to the null structure equations whose characteristic data agrees on its boundaries with the given characteristic data on $\HH_\AA$ and $\HH_\BB$.

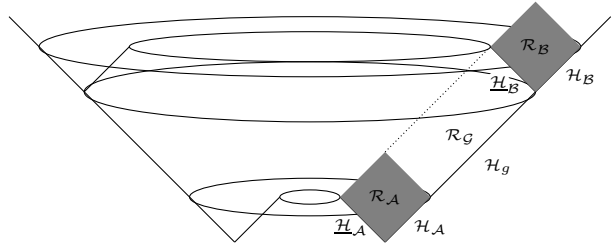
\begin{figure}[h]
\centering
\begin{tikzpicture}
\draw(1,-0.6)--(4,2.4);
\draw(-1,-0.6)--(-4,2.4);
\draw(-1,-.6)--(-.4,0);
\draw(1,-.6)--(.4,0);
\draw (0,0) ellipse (.4 and 0.095);
\draw (0,0) ellipse (1.6 and 0.25);
\draw(0,1.4) ellipse (3 and 0.4);
\draw(3,1.4)--(2.4,2);
\draw(-3,1.4)--(-2.4,2);
\draw (0,2) ellipse (2.4 and 0.2);
\draw (0,2) ellipse (3.6 and 0.4);
\fill[gray,nearly transparent] (.4,0)--(1,.6)--(1.6,0)--(1,-.6);
\draw(1,0) node[node font =\tiny]{$\mathcal{R}_\AA$};
\draw(1.6,-.4) node[node font =\tiny]{$\HH_\AA$};
\draw(.55,-.4) node[node font =\tiny]{$\HHb_\AA$};

\fill[gray,nearly transparent] (2.4,2)--(3,2.6)--(3.6,2)--(3,1.4);
\draw(3,2) node[node font =\tiny]{$\mathcal{R}_\BB$};
\draw(3.6,1.6) node[node font =\tiny]{$\HH_\BB$};
  \node at (2.5,1.51)[circle ,fill=white]{};
\draw(2.6,1.5) node[node font =\tiny]{$\HHb_\BB$};

\draw[densely dotted](1,.6)--(2.4,2);
\draw(2,0.8) node[node font =\tiny]{$\mathcal{R}_\GG$};
\draw(2.5,0.4) node[node font =\tiny]{$\HH_g$};
\end{tikzpicture}
\caption{The gluing problem for characteristic initial data. Given characteristic initial data posed on $(\HH_\AA, \HHb_\AA)$ and $(\HH_\BB, \HHb_\BB)$, the null gluing problem asks if there exists initial data along $\HH_g$ such that there exists a solution to \eqref{eq:EVE} on $\mathcal{R}_\AA\cup\mathcal{R}_\GG\cup\mathcal{R}_\BB$ with some specified differentiability.}
\label{fig:IDgluing}
\end{figure}

Historically, a considerable amount of attention has been given to the \emph{Riemannian gluing problem}. That is, the gluing of spacelike initial data sets to the Cauchy problem. The groundbreaking work of Corvino-Schoen \cite{Cor00, CorSch06} exploited the geometric under-determinedness of the constraint equations to show that asymptotically flat initial data can be glued across a compact region to initial data corresponding to the Kerr spacetime, see also \cite{ChrDel02,ChrDel03,ChrPol08,Cortier13,Hintz21}. Moreover, the work of Carlotto-Schoen \cite{CarSch} showed that initial data can be glued to Minkowski initial data along a non-compact cone. Another type of solution to the Riemannian gluing problem is to consider the connected sum gluing of initial data, see, for example, \cite{CIP04,CIP05,IMP,ChrMaz03,Isenberg2003}.

More recently, the study of \emph{null gluing problems for the Einstein Equations}, \eqref{eq:EVE}, was initiated by Aretakis, Czimek and Rodnianski \cite{ACR1,ACR2,ACR3} and their results are discussed in the next section.

\subsection{The null gluing problem}\label{sec:Cknullgluing}

The null gluing problem formulated in Figure \ref{fig:IDgluing} can straightforwardly be reduced to the problem depicted in Figure \ref{fig:nullgluing2}, where instead of prescribing characteristic data on $\HH_1$ and $\HH_2$, \emph{sphere data} is given on two spheres $S_1$ and $S_2$. Sphere data consists of the prescription of all metric components and their derivatives up to some specified order on a sphere, $S$. The transport nature of the null structure equations implies that ingoing derivatives of components of the sphere data satisfy transport equations along $\HH$ and higher-order components of the sphere data satisfy higher-order transport equations.  Therefore, we consider the \emph{$C^k$-null gluing problem} of gluing transverse derivatives of the components of $\g$ up to order $k$. It follows that $C^k$-sphere data is the prescription of derivatives of the components of $\g$ up to order $k$. The $C^k$-null gluing problem is formulated as follows.\\

\ni{\bf The $C^k$-null gluing problem.} \textit{Let $(\MM_1, {\bf g}_1)$ and $(\MM_2, {\bf g}_2)$ be solutions to \eqref{eq:EVE}.  Moreover, denote by $x_1$ and $x_2$, $C^k$-sphere data on spheres $S_1\subset(\MM_1, {\bf g}_1)$ and $S_2\subset(\MM_2, {\bf g}_2)$, respectively. Does there exist a null hypersurface $\HH_{[1,2]}:=\bigcup_{1\leq v\leq2}S_v$ and a solution $x$ to the null structure equations on $\HH_{[1,2]}$ such that  $$x\vert_{S_1}=x_1, \qquad x\vert_{S_2}=x_2?$$ }

 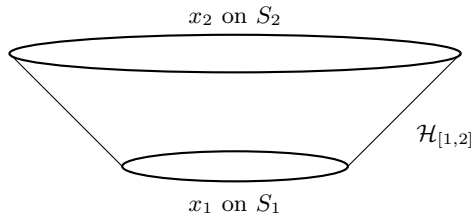
\begin{figure}[h]
 \centering
\begin{tikzpicture}
\draw[thick] (0,0) ellipse (1.5 and 0.2);
\draw[thick] (0,1.5) ellipse (3 and 0.25);
\draw (1.5,0) -- (3,1.5);
\draw (-1.5,0) -- (-3,1.5);
 \draw(0,-.5) node[node font =\small]{$x_1 \text{ on }S_1$};
\draw(0,2) node[node font =\small]{$x_2 \text{ on }S_2$};
\draw(2.8,.4) node[node font =\small]{$\HH_{[1,2]}$};
\end{tikzpicture}
\caption{The null gluing problem formulated with sphere data.}\label{fig:nullgluing2}
\end{figure}

\ni Upon inspection, the null gluing problem formulated in this manner cannot always be solved. An immediate obstruction is the monotonicity of the null expansion  along $\HH_{[1,2]}$ imposed by the Raychaudhuri equation. In double null gauge (see Appendix \ref{sec:doublenullcoords}) the outgoing null expansion is $\tr\chi$ and we immediately observe that the Raychaudhuri equation 
\begin{equation}\label{eq:raychaudhuri1}
L\left(\frac{\phi\tr\chi}{\Om}\right) = -\frac{\phi}{\Om^2}\big\vert\Om\chih\big\vert_\gd^2
\end{equation}
is monotonic, where $L$ denotes a derivation along $\HH$, see \eqref{eq:defL}. A deeper obstruction to the $C^k$-null gluing problem is the existence of an infinite-dimensional family of obstructions to the linearised null gluing problem at Minkowski in the form of linear conservation laws. Therefore, in order to solve the $C^k$-null gluing problem, it is useful to relax its formulation. This is achieved by the introduction of \emph{sphere perturbations}, that is, small changes of the sphere data $x_2$ on $S_2$. The following two types of sphere perturbations are considered.

\begin{itemize}
\item  \emph{Transversal perturbations} perturb the sphere $S_2$ to a nearby sphere $S'_2$ along an ingoing null hypersurface $\HHb_2$ containing both spheres. The $C^2$-sphere data $x'_2$ on $S'_2$ can be expressed in terms of $x_2$, a perturbation function $f$ and the geometry of $\HHb_2$. The precise formulas are given in Section \ref{sec:spherepert}. Moreover, the transversal perturbations correspond to the \emph{linearised pure gauge solutions} of \cite{DHR}.
\item \emph{Angular perturbations} are sphere diffeomorphisms of sphere data and are defined by pulling back the sphere data $x_2$ under a diffeomorphism of the sphere $S_2$.
\end{itemize}

\ni Employing transversal and angular perturbations, a relaxed formulation of the $C^k$-null gluing problem is posed as follows, see Figure \ref{fig:C2ACRnullgluing}.\\

\ni {\bf The perturbative $C^k$-null gluing problem.} \textit{Denote by $x_1$ and $\tilde{x}_2$ given $C^k$-sphere data on spheres $S_1$ and $\tilde{S}_2$ respectively. Consider the ingoing null hypersurface $\tilde{\HHb}_{[-\de,\de],2}$ containing $\tilde{S}_2$.  Does there exist a sphere $S_2\subset\tilde{\HHb}_{[-\de,\de],2}$, arising from a sphere perturbation of $\tilde{S}_2$, with sphere data $x_2$ such that there exists a null hypersurface $\HH_{[1,2]}:=\bigcup_{1\leq v\leq2}S_v$ and a solution $x$ to the null structure equations on $\HH_{[1,2]}$ such that  $$x\vert_{S_1}=x_1, \qquad x\vert_{S_2}=x_2?$$ }

 \begin{figure}[h]
 \centering
\begin{tikzpicture}
\draw (0,0) ellipse (1.5 and 0.2);
\draw (1.5,0) -- (3,1.5);
\draw (-1.5,0) -- (-3,1.5);
\draw[darkgray, thick] [domain=-3:2.9,smooth,variable=\t,smooth cycle]
plot ({2.99*cos(\t r)},{1.5+.25*sin(\t r))+rand*(0.175)*(abs(sin(\t r))}) ;
\draw [domain=-3.141:3,smooth,variable=\t,smooth cycle]
plot ({2.99*cos(\t r)},{1.5+.25*sin(\t r)}) ;
\draw[darkgray,very nearly transparent] (0,2) ellipse (2.5 and 0.25);
\draw[gray,very nearly transparent] (0,1) ellipse (3.5 and 0.25);
\fill[gray,very nearly transparent] (0,2) ellipse (2.5 and 0.25);
\fill[gray,very nearly transparent] (3.5,1) -- (2.5,2) arc (360:180:2.5cm and 0.25cm) -- (-3.5,1) arc (180:360:3.5cm and 0.25cm);
\fill[gray, very nearly transparent] (3.5,1) -- (2.5,2) arc (360:180:2.5cm and 0.25cm) -- (-3.5,1) arc (180:360:3.5cm and -0.25cm);
\draw(2.5,.35) node[node font =\tiny]{$\HH_{[1,2]}$};
\draw(3.4,1.9) node[node font =\tiny]{$\tilde{\HHb}_{[-\de,\de],2}$};
\draw(0,-.5) node[node font =\tiny]{$S_1$};
\draw(-3.3,1.5) node[node font =\tiny]{$S_2$};
\draw(3.2,1.4) node[node font =\tiny]{$\tilde{S}_2$};
\end{tikzpicture}
\caption{The perturbative $C^k$-null gluing problem. The thick line represents the sphere $\tilde{S}_2$ with sphere data $\tilde{x}_2$. }\label{fig:C2ACRnullgluing}
\end{figure}

\ni This formulation of the $C^k$-null gluing problem is solvable in the following sense. Under the sphere perturbations the family of obstructions to the $C^k$-null gluing problem split, at the linear level, into the following:
\begin{itemize}
\item an infinite number of obstructions that depend on the sphere perturbation, i.e. they are \emph{gauge-dependent},
\item a finite number of \emph{gauge-invariant} obstructions that are linearly invariant under the sphere perturbation.
\end{itemize}
 Hence, the aim is to show that the perturbative $C^k$-null gluing problem is solvable close to Minkowski up to the finite number of gauge-invariant obstructions. In \cite{ACR2}, the $C^k$-null gluing problem with $k=2$ is analysed and solved up to a $10$-dimensional space of obstructions. These obstructions will be discussed in the next section.

In this paper, we study the $C^k$-null gluing problem with $k=3$ and investigate its obstructions. As in \cite{ACR2}, we will employ the double null gauge which we outline here with precise definitions and formulas given in Appendix \ref{sec:doublenullcoords}, see also \cite{ChrFormationBlackHoles}. Let $S$ be a spacelike two-dimensional sphere in a spacetime $(\MM,\g)$ with induced metric and connection denoted by $\gd$ and $\Nabs$, respectively. Let $u$ and $v$ be two optical functions of $(\MM,\g)$. Then for real numbers $u_0$ and $v_0$ the level sets $$\HH_{u_0}:=\{u=u_0\},\qquad \HHb_{v_0}:=\{v=v_0\},$$ are outgoing and ingoing null hypersurfaces, respectively. The union of these null hypersurfaces is called a \emph{double null foliation} of $(\MM,\g)$, where in this case, $\HH_{u_0}$ and $\HHb_{v_0}$ intersect at the spheres 
$$S_{u_0,v_0}:=\{u=u_0,v=v_0\}.$$
The coordinates $(u,v,\th^1,\th^2)$ are \emph{double null coordinates} where $(\th^1,\th^2)$ are local angular coordinates defined everywhere by propagating them along the generators of $\HH$ and $\HHb$. Moreover, with respect to the coordinates $(\th^1,\th^2)$, the induced metric $\gd$ can be conformally related to a metric $\gd_c$ by 
\begin{equation}\label{eq:confdecomp1}
\gd=\phi^2\gd_c
\end{equation}
where $\det\gd_c=\det\gac$ and $\gac:=(d\th^1)^2 +(\sin\th^1)(d\th^2)^2$ is the unit round metric on $S$.
Combining the Einstein Equations, \eqref{eq:EVE}, with the double null foliation yield \emph{null structure equations} for the components of the metric and their derivatives. We write the components of the metric as $(\Om, \gd)$, their first derivatives in terms of \emph{Ricci coefficients} \eqref{eq:riccicoeffs} denoted by $$\{\chi,\chib,\eta,\zeta,\om,\omb\},$$ and their second derivatives in terms of \emph{null curvature components} \eqref{eq:nullcurvaturecomp} denoted by 
\begin{equation}\label{eq:nccintro}
\{\a,\ab,\b,\beb,\rh,\si\}.
\end{equation}

\ni Before moving on to discuss the $C^2$-null gluing problem of \cite{ACR2} and the $C^3$-null gluing problem, the general strategy to analyse the $C^k$-gluing problem close to Minkowski is exemplified. Due to the proximity to Minkowski, the implicit function theorem is applied and thus the $C^k$-null gluing problem close to Minkowski can be solved by analysing the \emph{linearised} $C^k$-null gluing problem at Minkowski. The linearisation procedure is denoted by the following expansion of sphere data.
\begin{equation}\label{eq:linintro}
x = x^{Mink.}+\varep \xd +\OO(\varep^2)
\end{equation}
where $x^{Mink.}$ denotes the sphere data associated with Minkowski spacetime. The precise definitions of the $C^2$- and $C^3$-sphere data and the associated Minkowski sphere data are given in Section \ref{sec:C3spheredata}. Following this linearisation procedure, we obtain the \emph{linearised null structure equations} for the components of $\xd$. The full set of null structure equations are presented in Appendix \ref{sec:doublenullcoords}. Their linearisations are given in Appendix \ref{sec:linCC} and the linearisations of the higher-order null structure equations we derive are contained in Lemma \ref{lem:linearisationofconstraints}. As an example, take the following subset of the null structure equations. For the metric $\gd$ with the conformal decomposition \eqref{eq:confdecomp1} and $\eta$, we have the null structure equations
$$D\gd =2\Om\chi, \quad D\phi = \frac{\Om\tr\chi\phi}{2}, \quad D\eta=\Om(\chi\cdot\etab-\b),$$
where $\tr\chi$ and $\chih$ denote the trace and tracefree parts of $\chi$ with respect to $\gd$, respectively. The Raychaudhuri equation \eqref{eq:raychaudhuri1} can be rewritten as
$$D (\Om\tr\chi )+ \frac{(\Om\tr\chi)^2}{2} = - \Om^2 \vert \chih \vert^2_{\gd}.$$
Linearising these equations at Minkowski using \eqref{eq:linintro} yields
\begin{equation}\label{eq:introNSE1}
D\gdcd = \frac{2}{v^2}\chihd,\quad D\left(\frac{\phid}{v}\right)=\frac{\dot{(\Om\tr\chi)}}{2}, \quad D(\dot{\Om\tr\chi})=-\frac{2\dot{(\Om\tr\chi)}}{v},
\end{equation}
and using the Gauss-Codazzi equation \eqref{eq:gausscodazzi}, 
\begin{equation}\label{eq:introNSE2}
D(v^2\etad) +\frac{v^2}{2}\di\left(\dot{(\Om\tr\chi)}-\frac{4}{v}\Omd\right)-2v\di\Omd -\Divdo\chihd =0,
\end{equation}
where $D=\Lied_{\partial_v}$ at Minkowski and $\Divdo:=\Divd_\gac$. Combining the linearised null structure equations \eqref{eq:introNSE1} and \eqref{eq:introNSE2} to obtain the equations 
\begin{align}
\begin{aligned}\label{eq:introNSE3}
D\left(\frac{v}{2}\left(\dot{(\Om\tr\chi)}-\frac{4}{v}\Omd\right)+\frac{\phid}{v}\right)=&0,\\
D\left(v^2\etad +\frac{v^3}{2}\di\left(\dot{(\Om\tr\chi)} -\frac{4}{v}\Omd\right)\right)=&\Divdo\chihd,
\end{aligned}
\end{align}
sheds light on their structure allowing us to make the following observations.\\
\begin{enumerate}
\item \textit{Conserved quantities.}

\ni An inspection of the first formula of \eqref{eq:introNSE3} shows that the quantity 
$$\QQ_\phid :=\frac{v}{2} \left(\omtrchid-\frac{4}{v}\Omd\right) + \frac{\phid}{v}$$
is conserved along $\HH_{[1,2]}$. Projecting the second equation of \eqref{eq:introNSE3} onto spherical harmonics of mode $l=1$ and using that the $l=1$ mode of the divergence of a symmetric tracefree 2-tensor vanishes (see Appendix \ref{sec:sphharm}) implies that the quantity 
$$\QQ_{\etad^{[1]}}:= v^2 \etad^{[1]} + \frac{v^3}{2}\di \left(\omtrchid^{[1]}-\frac{4}{v}\Omd^{[1]}\right)$$
is also conserved along $\HH_{[1,2]}$. The charges $\QQ_\phid$ and $\QQ_{\etad^{[1]}}$ are the first two of a set of conserved charges present in the $C^2$- and $C^3$-null gluing problems. The full set of charges appearing in the $C^2$-null gluing problem are presented in Appendix \ref{sec:linCC} while the novel charges that present themselves in the $C^3$-null gluing problem are derived in Section \ref{sec:conservationlaws}. Furthermore, the charges $\QQ_\phid$ and $\QQ_{\etad^{[1]}}$ exhibit the splitting of conserved charges into those that are gauge-dependent and those that are gauge-invariant. Investigating how they behave under the sphere perturbations yields 
$$\QQ_\phi=\frac{1}{2}\left(\Ldo\dot{f}-\partial_u\dot{f}\right),\qquad \QQ_{\etad^{[1]}} =0 $$
where $\dot{f}$ is the linearised perturbation function. The formulas for sphere perturbations are given in Lemma \ref{lem:pertlin} and Appendix \ref{sec:transversalspherepert}. The independence of $\QQ_{\etad^{[1]}}$ on $\dot{f}$ means that it is a gauge-invariant conserved charge. Moreover, it is related to the linearisations of the energy and linear momentum discussed in the next section. The gauge-invariance of $\QQ_{\etad^{[1]}}$ on $\dot{f}$ means that it is not glued by our methods.\\

\item \textit{Representation formulas and $v$-weights}

\ni The null structure equations have an intricate underlying structure. Sachs \cite{SachsCIVP} showed that they can be rewritten in a  hierarchical order that can be solved sequentially from a \emph{characteristic seed} consisting of a subset of the initial data. More specifically, from a prescription of the characteristic seed the null structure equations can be integrated sequentially where each equation only depends on this characteristic seed and the previous quantities solved for in the hierarchy. More details can be found in Section \ref{sec:charseed}. Using this hierarchical structure, we are able to integrate the linearised null structure equations to obtain representation formulas for the components of the linearised sphere data. Integrating the linearised null structure equations \eqref{eq:introNSE1} and the second of \eqref{eq:introNSE3} in $v$ from $1$ to $2$ yields the following representation formulas:
\begin{align*}
\phid = &2 \int_1^2 \Omd dv' +v\phid(1)+\frac{1}{2}\left(\omtrchid(1)-4\Omd(1)\right),\\
\left[\dot{\gd_c}\right]_1^2 =&2 \int_{1}^v \frac{1}{v'^2} \chihd dv'\\
\left[ v'^2 \etad + \frac{v'^3}{2} \di\left(\omtrchid - \frac{4}{v'} \Omd\right)\right]_1^2
 =& \Divdo\left( \int\limits_1^2 \chihd dv' \right).
\end{align*}
Thus, we can construct a solution to the null structure equations by prescribing and controlling the integrals of $\Omd$ and $\chihd$. The essential property of the representation formulas is that the integrals of $\Omd$ and $\chihd$ contain a different $v$-weight in each representation formula and are linearly independent. Moreover, if the representation formulas for two different quantities contain the same $v$-weight then the representation formulas can be combined to form a conservation law.
\end{enumerate}
Extending this procedure to the entire set of null structure equations gives a solution to the $C^k$-null gluing problem up to the finite number of gauge-invariant obstructions. In the next section, we define the obstructions discovered in \cite{ACR2} and in Section \ref{sec:C3novel}, we define the novel obstructions that appear in the $C^3$-null gluing problem and state the main result of this paper.

\subsection{The $C^2$-null gluing problem of Aretakis-Czimek-Rodnianski}

In the series of papers \cite{ACR1,ACR2,ACR3}, the $C^2$-null gluing problem is solved up to a $10$-dimensional space of obstructions for sphere data close to the Minkowski round spheres. That is, the infinite-dimensional family of obstructions to the linear problem split into infinitely many that depend on the sphere perturbations and a $10$-dimensional family of obstructions that are gauge-invariant. The $10$ gauge-invariant conservation laws are the obstructions to solving the $C^2$-null gluing problem. The charges are defined as follows.

\begin{definition}[$C^2$-charges]\label{def:chargesEPLG} Let $S$ be a $2$-sphere. The $10$ obstructions to the $C^2$-null gluing problem are the charges denoted by $(\mathbf{E},\mathbf{P},\mathbf{L},\mathbf{G})$ and are defined using tensor spherical harmonics on $S$ with $-1\leq m \leq 1$
\begin{align*}
\mathbf{E}:=& -\frac{1}{8\pi}\sqrt{4\pi}\left(\phi^3(\rh+\phi\Divd\b)\right)^{(0)},\\
 \mathbf{P}^m:=&  -\frac{1}{8\pi}\sqrt{4\pi}\left(\phi^3(\rh+\phi\Divd\b)\right)^{(0)},\\
  \quad \mathbf{L}^m := &\frac{1}{16\pi}\sqrt{\frac{8\pi}{3}}\left(\phi^3(\di\tr\chi+\tr\chi(\eta-\di\log\Om))\right)_H^{(1m)},\\
   \mathbf{G}^m :=&\frac{1}{16\pi}\sqrt{\frac{8\pi}{3}}\left(\phi^3(\di\tr\chi+\tr\chi(\eta-\di\log\Om))\right)_E^{(1m)}.
\end{align*}
where for a scalar function $f$ and an $S$-tangential vectorfield $V$, the tensor spherical harmonics with respect to the round unit metric $\gac$ are given by
\begin{align*}
f^{(lm)} := \int_S f \, Y^{(lm)} d\mu_{\gac}, \quad V_E^{(lm)} := \int_S V\cdot E^{(lm)} d\mu_{\gac}, \quad V_H^{(lm)} := \int_S V\cdot H^{(lm)} d\mu_{\gac},
\end{align*}
where $Y^{lm}$ are the standard real-valued spherical harmonics with respect to $\gac$ and $\cdot$ denotes the product with respect to $\gac$. The electric and magnetic tensor spherical harmonics $E^{(lm)}$ and $H^{(lm)}$, respectively are defined by
\begin{align*}
E^{(lm)} := -\frac{1}{\sqrt{l(l+1)}} \di Y^{lm}, \quad H^{(lm)} := \frac{1}{\sqrt{l(l+1)}}{}^* \di Y^{lm}, l\geq1, \, -l\leq m \leq l.
\end{align*}
Definitions and an outline of spherical harmonics can be found in Appendix \ref{sec:sphharm}.
\end{definition}

\begin{remark} We make the following remarks on Definition \ref{def:chargesEPLG}.
\begin{enumerate}
\item The charges $(\mathbf{E},\mathbf{P},\mathbf{L},\mathbf{G})$ have geometric significance. They correspond to the ADM energy, linear momentum, angular momentum and centre-of-mass, respectively.
\item The linearisations of $(\dot{\mathbf{E}},\dot{\mathbf{P}},\dot{\mathbf{L}},\dot{\mathbf{G}})$ at Minkowski correspond to gauge-invariant conserved charges in the linearised $C^2$-null gluing problem along $\HH_{[1,2]}$ which are not glued by the methods here.
\end{enumerate}
\end{remark}

\ni The $10$ gauge-invariant conservation laws $(\mathbf{E},\mathbf{P},\mathbf{L},\mathbf{G})$ will also appear as obstructions to the $C^k$-gluing problem for any integer $k\geq2$. Motivated by the correspondence between obstructions to the $C^2$-null gluing problem and the role of conserved quantities, we study the $C^3$-null gluing problem and investigate the existence of additional obstructions associated with the higher-order null structure equations.

In the literature, since the introduction of the null gluing problem in \cite{ACR1,ACR2,ACR3}, there have been a number of applications and generalisations. The method of null gluing was further developed in \cite{KehUng24b, KehUng24a} in spherical symmetry to disprove the third law of black hole thermodynamics. In \cite{CRobsfree}, a careful analysis of the nonlinear null structure equations with a high frequency ansatz for the characteristic seed yielded obstruction free gluing. In other words, they showed that the charges $(\mathbf{E},\mathbf{P},\mathbf{L},\mathbf{G})$ are glued nonlinearly. The work \cite{CCgluing, CCGgluingb, CCGgluinga} investigates the $C^k$-null gluing problem in Bondi gauge accounting for cosmological constants, higher dimensions and the intersection of null hypersurfaces having higher-genus topology for each $k\geq2$. 

In this paper, we present for the first time the explicit formulas of the third-order charges that act as obstructions to the $C^3$-null gluing problem in the regime close to Minkowski. In the following section, we further motivate this investigation and present our main result.

\subsection{The $C^3$-null gluing problem and novel obstructions}\label{sec:C3novel}

In General Relativity (so far) the only known nonlinearly conserved quantities along null hypersurfaces are the celebrated Newman-Penrose constants which are conserved along future null infinity \cite{NPchargesa, NPchargesb}. These $10$ constants are calculated from the $l=2$ spherical harmonic mode of a third derivative of the metric. The fact that the Newman-Penrose constants arise as a third derivative of a component of the metric motivates a full analysis of the $C^3$-null gluing problem. In this paper, we solve the $C^3$-null gluing problem locally (i.e. in the bulk of a spacetime) near asymptotically flat spacelike infinity. By scaling, this translates into a small data problem. As obstructions to its solution, we find that the $l=2$ modes of the quantity $$\frac{1}{\phi}\Dbh\ab+\frac{1}{\phi^{2}}(1+\phi^2\Nabs\widehat{\otimes}\Divd)\ab$$ are conserved linearly along outgoing null hypersurfaces.

\begin{definition}[Novel $C^3$-charges]\label{def:chargeW} Let $S$ be a $2$-sphere. We define the following charges using tensor spherical harmonics on $S$ with $-2\leq m \leq 2$
\begin{align*}
\mathbf{U}^m:=&\left(\frac{1}{\phi}\Dbh\ab+\frac{1}{\phi^{2}}(1+\phi^2\Nabs\widehat{\otimes}\Divd)\ab\right)^{(2m)}_\psi\\
\mathbf{V}^m:=&\left(\frac{1}{\phi}\Dbh\ab+\frac{1}{\phi^{2}}(1+\phi^2\Nabs\widehat{\otimes}\Divd)\ab\right)^{(2m)}_\phi
\end{align*}
where for an $S$-tangential $2$-tensor $W$, the Fourier coefficients are defined by
\begin{align*}
W_\psi^{(lm)} := \int_S W\cdot \psi^{lm}d\mu_{\gac}, \quad V_\phi^{(lm)} := \int_S W\cdot \phi^{lm}d\mu_\gac, \quad  l\geq2,\, m=-l,\dots,l.
\end{align*}
and the electric and magnetic $2$-tensor spherical harmonics $ \psi^{lm}$ and $ \phi^{lm}$ are defined respectively by
\begin{align*}
\psi^{lm} := \frac{1}{\sqrt{\frac{1}{2} l(l+1)-1}} \DDd_2^* E^{lm}, \quad \phi^{lm} := \frac{1}{\sqrt{\frac{1}{2} l(l+1)-1}} \DDd_2^* H^{lm}, \quad  l\geq2, \, m=-l,\dots,l.
\end{align*}
The operator $\DDd_2^*$ is the Hodge operator (see Appendix \ref{sec:sphharm}) defined by 
\begin{align*}
\left(\DDd_2^*V\right)_{AB}:=-\frac{1}{2}\left(\Nabs_A V_B+\Nabs_B V_A -\Divd V \gd_{AB}\right)
\end{align*}
for an $S$-tangential vectorfield, $V$.
\end{definition}
\begin{remark}
The quantities $\mathbf{U}$ and $\mathbf{V}$ are defined so that their linearisations correspond to the gauge-invariant charges appearing in the linearised $C^3$-null gluing problem. As the smallness regime assumed in our gluing construction corresponds (by scaling) to working in an asymptotically flat region near \emph{spacelike} infinity (see, for example, \cite{ACR2}), our charges are not expected to have the same formula as the Newman-Penrose constants. 
\end{remark}

\ni We solve the $C^3$-null gluing problem up to the $20$-dimensional space of charges $(\mathbf{E},\mathbf{P},\mathbf{L},\mathbf{G},\mathbf{U},\mathbf{V})$. The statement of our main result is the following.
\begin{theorem}[Codimension-$20$ $C^3$-null gluing, version 1]\label{thm:mainv1} Let $\mfx_1$ and $\tilde{\mfx}_2$ denote $C^3$-sphere data on spheres $S_1$ and $\tilde{S}_2$ close to Minkowski sphere data. Moreover, Let  $\tilde{S}_2$ be contained in an ingoing null hypersurface $\HHb_{[-\de,\de],2}$ with real number $\de>0$. Then there exists
\begin{itemize}
\item a solution to the null structure equations  $\mfx$ along an outgoing null hypersurface $\HH_{[1,2]}$,
\item sphere data $\mfx_2$ on a sphere $S_2$ arising from a sphere perturbation of $\tilde{S}_2\in\HHb_{[-\de,\de],2}$ 
\end{itemize} 
such that on $S_1$ we have that 
\begin{equation}
\mfx\big\vert_{S_1} = \mfx_1.
\end{equation}
 Moreover, if the gauge-invariant charges $(\mathbf{E},\mathbf{P},\mathbf{L},\mathbf{G},\mathbf{U},\mathbf{V})$ match on $S_2$, that is,
 $$(\mathbf{E},\mathbf{P},\mathbf{L},\mathbf{G},\mathbf{U},\mathbf{V})(\mfx\vert_{S_2})=(\mathbf{E},\mathbf{P},\mathbf{L},\mathbf{G},\mathbf{U},\mathbf{V})(\mfx_2)$$
 then on $S_2$ it holds that 
 \begin{equation*}
 \mfx\big\vert_{S_2} = \mfx_2.
 \end{equation*}
\end{theorem}

\begin{remark}
The work \cite{CCgluing} characterises the gauge-invariant obstruction space of the $C^3$-null gluing problem in Bondi gauge to be of dimension $20$ and involving $l=2$ spherical harmonic modes. Our analysis in this paper, see specifically the explicit formulas for $\mathbf{U}$ and $\mathbf{V}$ above, shows that in particular these two basic properties are also satisfied in double null gauge.
\end{remark}

\ni The precise statement of Theorem \ref{thm:mainv1} is given in Section \ref{sec:mainthmprecise} and the proof is outlined in the next section.

\subsection{Sketch of the proof of the main theorem}

The strategy to prove Theorem \ref{thm:mainv1} is to apply the standard implicit function theorem, see for example Theorem 2.5.7 in \cite{MarsdenIFT}, which means that we solve the $C^3$-null gluing problem close to Minkowski up to the $20$-dimensional space of obstructions $(\mathbf{E},\mathbf{P},\mathbf{L},\mathbf{G},\mathbf{U},\mathbf{V})$ by solving the \emph{linearised} problem at Minkowski. We pay close attention to the higher-order components of the sphere data. Since, by definition the $C^3$-sphere data contains the $C^2$-sphere data, we employ the analysis of the linearised $C^2$-null gluing problem from \cite{ACR2} and use their result as a black box in order to glue the $C^2$-part of the sphere data. The focus of this paper will be on the novel components of the $C^3$-sphere data, $$ (D^2\om, \Db^2\omb,\Dh\a,\Dbh\ab).$$
\ni We begin by deriving the higher-order null structure equations consisting of nonlinear transport equations that the quantities $\Db^2\omb$ and $\Dbh\ab$ satisfy. The transport equations for $D^2\om$ and $\Dh\a$ are straightforwardly  computed from the definition of $\om$ and the transport equation for $D\chih$, see \eqref{eq:transporteqriccicoeffs}. Using the linearisation procedure outlined in Section  \ref{sec:Cknullgluing} the linearised transport equations are obtained. Integrating the linearised transport equations along $\HH_{[1,2]}$ yields representation formulas for the components of the sphere data $\Db^2\omb$ and $\Dbh\ab$. The associated novel conservation laws are found by analysing the representation formulas. In order to understand how the conservation laws behave under sphere perturbations, we derive the nonlinear sphere perturbation formulas associated with the $C^3$-components of the sphere data, $ (D^2\om, \Db^2\omb,\Dh\a,\Dbh\ab)$ and linearise them, accordingly. 
The gauge-dependence of the novel conservation laws is investigated and we show that the two charges associated with the linearisations of $\mathbf{U}$ and $\mathbf{V}$ are gauge-invariant. Subsequently, we show that the $C^3$-null gluing problem is solved by applying the results of the $C^2$-null gluing problem of \cite{ACR2} and constructing the solution to the null structure equations for the $C^3$-components of the sphere data up to the charges $\mathbf{U}$ and $\mathbf{V}$. 

The rest of this paper is organised as follows.
\begin{itemize}
\item In Section \ref{sec:prelim}, we set up the notation, definitions and relevant formulas for this paper.
\item In Section \ref{sec:mainthmprecise}, we present the precise version of Theorem \ref{thm:mainv1}.
\item In Section \ref{sec:mainthmproof}, we setup the functional analytic framework and prove the main Theorem \ref{thm:main} using the implicit function theorem which relies on the solvability of the linearised $C^3$-null gluing problem.
\item In Section \ref{sec:lingluing}, we solve the linearised $C^3$-null gluing problem at Minkowski.
\item In Appendix \ref{sec:doublenullcoords}, we setup the double null gauge used throughout the paper and present the null structure equations.
\item In Appendix \ref{sec:linsec}, we present the linearised null structure equations and the linearised charges that appear in the $C^2$-null gluing problem.
\item In Appendix \ref{sec:transversalspherepert}, we setup the transversal sphere perturbations and derive the perturbation formulas for the novel quantities in the $C^3$-null gluing problem.
\item In Appendix \ref{sec:sphharm}, we recall the Hodge theory and Fourier analysis on the round unit sphere.
\end{itemize}

\subsection{Acknowledgements} 

This work forms part of my PhD thesis and I express my gratitude to Stefan Czimek for suggesting this problem, his thoughtful guidance and encouragement. Additionally, I thank Juan Valiente-Kroon for his support, useful conversations and comments on an earlier draft. I am grateful for the hospitality of the University of Leipzig and the MPI for Mathematics in the Sciences, Leipzig, where part of this work was completed.

\section{Preliminaries}\label{sec:prelim}

\ni In this section, we introduce definitions and notation used throughout the paper. We work in double null coordinates and an overview of the definitions, notation and relevant null structure equations can be found in Appendix \ref{sec:doublenullcoords}, see also \cite{ChrFormationBlackHoles}. In this work, uppercase Latin indices are over $A,B,C,D=1,2$. The notation $A\lesssim B$ means that $A\leq c B$ where $c>0$ is a universal constant that does not depend on $A$ or $B$. \\

\subsection{$C^3$-sphere data and norms}\label{sec:C3spheredata}

In this section we define sphere data for both the $C^2$- and $C^3$-null gluing problem and associated norms. We first introduce the notion of $C^2$- and $C^3$-sphere data. 
\begin{definition}[$C^2$- and $C^3$-sphere data]\label{def:spheredata}
Let $u$ and $v$ be two real numbers with $v\geq u$ and let $S_{u,v}$ be a $2$-sphere equipped with a round metric. Let \emph{$C^3$-sphere data}, denoted by $\mfx_{u,v}$, be the following tuple of tensors on $S_{u,v}$
\begin{equation}\label{eq:C3spheredata}
\mfx_{u,v}:=x_{u,v}\times (D^2\om, \Db^2\omb,\Dh\a,\Dbh\ab)
\end{equation}
where $x_{u,v}$ denotes \emph{$C^2$-sphere data} and consists of the following tuple of tensors
\begin{align} \label{eq:C2spheredata}
x_{u,v} := \left(\Om,\gd, \Om\trchi, \chih, \Om\trchib, \chibh, \eta, \om, D\om, \omb, \Db\omb, \a, \ab\right),
\end{align}
\end{definition}

\ni\textbf{Notation.} We make use of the following notation.
\begin{itemize}
\item On outgoing null hypersurfaces $\HH_{u_0,[v_1,v_2]}$ with $u_0\leq v_1\leq v_2$, define $C^3$-\emph{outgoing null data} to be the tuple, $\mfx$, on each $S_{u_0,v}\subset\HH_{u_0,[v_1,v_2]}$ such that $\mfx\vert_{S_{u_0,v}}=\mfx_{u_0,v}$. Similarly on ingoing null hypersurfaces $\HHb_{[u_1,u_2],v_0}$ with $u_1\leq u_2 \leq v_0$, define \emph{$C^3$-ingoing null data} to be the tuple, $\underline{\mfx}$, on each $S_{u,v_0}\subset\HHb_{[u_1,u_2],v_0}$ such that $\underline{\mfx}\vert_{S_{u,v_0}}=\mfx_{u,v_0}$. The definition of $C^2$-outgoing and ingoing null data then follows from the restriction to $C^2$-sphere data, $x_{u,v}$.
\item The \emph{reference Minkowski sphere data} is denoted by $\mfm$ with $r=v-u$ and
\begin{equation}
\mfm_{u,v}=\left(1,v^{2}\gac,\frac{2}{v},0,-\frac{2}{v}, 0,0,0,0,0,0,0,0,0,0,0,0\right).
\end{equation}
We abuse notation and let $\mfm_{u,v}$ denote both  reference Minkowski $C^2$- and $C^3$-sphere data.
\end{itemize}

\begin{remark}
We make the following remarks on the definition of $C^3$-sphere data.
\begin{enumerate}
\item Our definition of $C^3$-sphere data agrees with the \emph{higher-order sphere data} defined in \cite{ACR2}.
\item It is not necessary for the definition of sphere data to include the prescription of all metric components, Ricci coefficients and null curvature components due to redundancy in the null structure equations. For example, from the prescription of the components of sphere data in \eqref{eq:C2spheredata} combined with the gauss equations, $$K + \frac{1}{4} \tr \chi \tr \chib - \frac{1}{2} (\chih,\chibh) = - \rh,$$ determines the null curvature component, $\rh$.
\end{enumerate}
\end{remark}

\ni The standard tensor spaces on spheres and null hypersurfaces are defined as follows.
\begin{definition}[Tensor spaces on $2$-spheres] \label{def:spherespaces} Let $u$ and $v$ be two real numbers with $v\geq u$, let $S_{u,v}$ be a $2$-sphere. For a non-negative integer $m$ and $S_{u,v}$-tangential tensors $T$, define
\begin{align*} 
\Vert T \Vert^2_{H^m(S_{u,v})} :=   \sum_{i=0}^m  \left\Vert \Nabs^i T \right\Vert^2_{L^2(S_{u,v})},
\end{align*}
where the covariant derivative $\Nabs$ and the measure in $L^2(S_{u,v})$ are with respect to the metric $\ga=(v-u)^2\gac$. Moreover, let $ H^m(S_{u,v}) := \{ T: \Vert T \Vert_{H^m(S_{u,v})} < \infty \}$.
\end{definition}
\begin{definition}[Tensor spaces on null hypersurfaces] \label{def:nullHHspaces} Let $m$ and $l$ be two non-negative integers.
\begin{enumerate}
\item For real numbers satisfying $u_0 < v_1 < v_2$ and $S_{u_0,v}$-tangential tensors $T$ on $\HH_{u_0,[v_1,v_2]}$, define
\begin{align*} 
\Vert T \Vert^2_{H^m_l(\HH_{u_0,[v_1,v_2]})} :=  \int_{v_1}^{v_2} \sum_{0\leq j\leq l} \left\Vert D^j T \right\Vert^2_{H^m(S_{u_0,v})} \,dv,
\end{align*}
and let $H^m_l(\HH_{u_0,[v_1,v_2]}) := \{ F: \Vert F \Vert_{H^m_l(\HH_{u_0,[v_1,v_2]})} < \infty\}.$
\item For real numbers $u_1 < u_2 < v_0$ and $S_{u,v_0}$-tangential tensors $T$ on $\HHb_{[u_1,u_2],v_0}$, define
\begin{align*} 
\Vert T \Vert^2_{H^m_l(\HHb_{[u_1,u_2],v_0})} :=  \int_{u_1}^{u_2} \sum_{0\leq j\leq l} \left\Vert \Db^j T \right\Vert^2_{H^m(S_{u,v_0})} \,du,
\end{align*}
and let $ H^m_l(\HHb_{[u_1,u_2],v_0}) := \{ F: \Vert F \Vert_{H^m_l(\HHb_{[u_1,u_2],v_0})} < \infty\}.$
\end{enumerate}
\end{definition}

For the $C^2$-sphere data, the following norms are defined to reflect the regularity hierarchy of the null structure equations.

\begin{definition}[$C^2$-sphere data norm]\label{def:C2spheredata} Let $v$ and $u$ be two real numbers with $v>u$ and non-negative integer $k$. For $C^2$-sphere data $x_{u,v}$ on $S_{u,v}$ define
\begin{align} 
\begin{aligned}
\Vert x_{u,v} \Vert_{X^k(S_{u,v})} :=& \Vert\Om\Vert_{H^{k}(S_{u,v})} +\Vert\gd\Vert_{H^{k}(S_{u,v})} + \Vert \Om\trchi \Vert_{H^{k}(S_{u,v})} + \Vert \chih \Vert_{H^{k}(S_{u,v})}\\
& + \Vert \Om\trchib \Vert_{H^{k-2}(S_{u,v})} + \Vert \chibh \Vert_{H^{k-2}(S_{u,v})} + \Vert \eta \Vert_{H^{k-1}(S_{u,v})} \\
&+ \Vert \om \Vert_{H^{k}(S_{u,v})}+ \Vert D\om \Vert_{H^{k}(S_{u,v})}+\Vert \omb \Vert_{H^{k-2}(S_{u,v})} + \Vert \Db\omb \Vert_{H^{k-4}(S_{u,v})} \\
&+ \Vert \a \Vert_{H^{k}(S_{u,v})} +\Vert \ab \Vert_{H^{k-4}(S_{u,v})}.
\end{aligned}\label{eq:C2spherenorm}
\end{align}
Let $X^k(S_{u,v}) := \{ x_{u,v} : \Vert x_{u,v} \Vert_{\XX^k(S_{u,v})} < \infty\}.$
\end{definition}
\begin{remark} In order apply standard calculus estimates and Sobolev embeddings on the norm defined in definition \ref{def:C2spheredata}, we consider $k\geq6$.
\end{remark}
\begin{definition}[$C^2$-null data norm]\label{def:C2nullnorm} 
Let $x$ be outgoing null data on $\HH=\HH_{u_{0},[v_{1},v_{2}]}$ with $u_0\leq v_1\leq v_2$ and $\underline{x}$ be ingoing null data on $\HHb=\HHb_{[u_{1},u_2],v_0}$ with $u_1\leq u_2 \leq v_0$. For integers $k>0$ and $m>0$, define the \emph{$C^2$-outgoing null data norm} by
\begin{align}
\begin{aligned}
\Vert x \Vert_{X^k_m(\HH)} :=& \Vert \Om \Vert_{H^k_m(\HH)} +\Vert \gd \Vert_{H^k_m(\HH)}+ \Vert \Om\trchi \Vert_{H^k_m(\HH)}+\Vert \chih \Vert_{H^k_{m-1}(\HH)}\\
& + \Vert \Om\trchib \Vert_{H^{k-2}_{m+1}(\HH)}+ \Vert \chibh \Vert_{H^{k-2}_m(\HH)}+ \Vert \eta \Vert_{H^{k-1}_m(\HH)}\\
&+\Vert \om \Vert_{H^k_{m-1}(\HH)} +\Vert D\om \Vert_{H^k_{m-2}(\HH)}+ \Vert \omb \Vert_{H^{k-2}_m(\HH)}+ \Vert \Db\omb \Vert_{H^{k-4}_m(\HH)}\\
&+  \Vert \a \Vert_{H^k_{m-2}(\HH)} +\Vert \ab \Vert_{H^{k-4}_m(\HH)},
\end{aligned}\label{eq:C2nullnorm}
\end{align}
and the associated space $X^k_m(\HH) := \{ x : \Vert x \Vert_{X^k_m(\HH)} < \infty \}$. Moreover, define the \emph{$C^2$-ingoing null data norm} by
\begin{align*}
\Vert \underline{x} \Vert_{X^k_m(\HHb)} :=& \Vert \Om \Vert_{H^k_m(\HH)} +\Vert \gd \Vert_{H^k_m(\HH)}+ \Vert \Om\trchib \Vert_{H^k_m(\HH)}+\Vert \chibh \Vert_{H^k_{m-1}(\HH)}\\
& + \Vert \Om\trchi \Vert_{H^{k-2}_{m+1}(\HH)}+ \Vert \chih \Vert_{H^{k-2}_m(\HH)}+ \Vert \eta \Vert_{H^{k-1}_m(\HH)}\\
&+\Vert \omb \Vert_{H^k_{m-1}(\HH)} +\Vert \Db\omb \Vert_{H^k_{m-2}(\HH)}+ \Vert \om \Vert_{H^{k-2}_m(\HH)}+ \Vert D\om \Vert_{H^{k-4}_m(\HH)}\\
&+  \Vert \ab \Vert_{H^k_{m-2}(\HH)} +\Vert \a \Vert_{H^{k-4}_m(\HH)},
\end{align*}
and the associated space $ X^k_m(\HHb) := \{ \underline{x} : \Vert \underline{x} \Vert_{X^k_m(\HH)} < \infty \}$.
\end{definition}
\ni For $C^3$-sphere data $\mfx_{u,v}$, we similarly define the $C^3$-sphere data norm to reflect the regularity hierarchy of the null structure equations.
\begin{definition}[$C^{3}$-sphere data norm]\label{def:C3spherenorm} Let $\mfx_{u,v}$ be $C^3$-sphere data and $x_{u,v}$ be $C^2$-sphere data on $S_{u,v}$. Define
\begin{align*} 
\Vert \mfx_{u,v} \Vert_{\XX(S_{u,v})} :=&\Vert x_{u,v}\Vert_{X^8(S_{u,v})}+\Vert D^{2}\om \Vert_{H^{8}(S_{u,v})} + \Vert \Db^{2}\omb \Vert_{H^{2}(S_{u,v})}  + \Vert \Dh\a \Vert_{H^{8}(S_{u,v})}\\
& +\Vert \Dbh\ab \Vert_{H^{2}(S_{u,v})}.
\end{align*}
The norms defined with respect to the round metric $\gac$ on $S_{u,v}$ and $
\XX(S_{u,v}) := \{ \mfx_{u,v} : \Vert \mfx_{u,v} \Vert_{\XX(S_{u,v})} < \infty\}.$
\end{definition}
\ni For $C^3$-outgoing null data $\mfx$ we have the following norm.
\begin{definition}[$C^{3}$-null data norm (outgoing)]\label{def:C3nullnorm} Let $\mfx$ and $x$ be outgoing $C^3$ and $C^2$ null data, respectively,  on $\HH=\HH_{u_{0},[v_{1},v_{2}]}$ with $u_0\leq v_1\leq v_2$. Define the \emph{$C^3$-outgoing null data norm} by
\begin{align*}
\Vert \mfx \Vert_{\XX(\HH)} :=& \Vert x \Vert_{X^8_4(\HH)} +\Vert D^{2}\om \Vert_{H^8_1(\HH)}+ \Vert \Db^{2}\omb \Vert_{H^2_4(\HH)}+ \Vert \Dh\a \Vert_{H^{8}_1(\HH)}+\Vert \Dbh\ab \Vert_{H^{2}_4(\HH)}
\end{align*}
and the associated space by $X(\HH) := \{ \mfx : \Vert \mfx \Vert_{\XX(\HH)} < \infty \}.$
\end{definition}
\ni We work with the following higher-regularity ingoing null data that reflects the regularity of the null structure equations along the ingoing null hypersurface, $\HHb$.  
\begin{definition}[Higher regularity ingoing null data norms]\label{def:nullnormhighreg} Let $\underline{x}$ be ingoing $C^2$-null data and $\underline{\mfx}$ ingoing $C^3$-null data on $\HHb=\HHb_{[u_{1},u_{2}],v_{0}}$ with $u_1\leq u_2\leq v_0$. Define the $C^2$ higher regularity ingoing null data norm by
\begin{align*}
\Vert \underline{x} \Vert_{X^+(\HHb)}:=\Vert \underline{x}\Vert_{X^7_3(\HHb)}
\end{align*}
with the associated space $X^+(\HHb) := \{ \underline{x} : \Vert \underline{x} \Vert_{X^+(\HHb)} < \infty \}$.
Moreover, define the $C^3$ higher regularity ingoing null data norm by
\begin{align*}
\Vert \underline{\mfx} \Vert_{\XX^+(\HHb)}:=\Vert \underline{x}\Vert_{\XX^9_4(\HHb)}+\Vert \Db^{2}\omb \Vert_{H^9_1(\HH)}+ \Vert D^{2}\om \Vert_{H^3_4(\HH)}+ \Vert \Dbh\ab \Vert_{H^{9}_1(\HH)}+\Vert \Dh\a \Vert_{H^{3}_4(\HH)}.
\end{align*}
with the associated space $ \XX^+(\HHb) := \{ \underline{\mfx} : \Vert \underline{\mfx} \Vert_{\XX^+(\HHb)} < \infty \}.$
\end{definition}

\subsection{Characteristic seed and hierarchy of null structure equations}\label{sec:charseed}

In this section, a characteristic seed is defined for the characteristic Cauchy problem from which a solution to the null structure equations along $\HH_{[1,2]}$ is constructed. This solution is constructed by integrating sequentially a hierarchy of transport equations along $\HH_{[1,2]}$ where each transport equation depends only on the characteristic seed and quantities appearing in the preceding equations in the hierarchy. The transport equations are derived by combining null structure equations. For more details, see for example \cite{SachsCIVP}.

\begin{definition}[$C^2$- and $C^3$-characteristic seed on $\HH_{[1,2]}$]\label{def:charseed} Let $\HH_{[1,2]}=\bigcup_{1\leq v\leq 2}S_v$. A \emph{$C^2$-characteristic seed} along $\HH_{[1,2]}$ consists of the prescription of the following quantities:
\begin{itemize}
\item on $\HH_{[1,2]}$ prescribe a positive scalar function $\Om$ and a conformal class of induced metrics $\mathrm{conf}(\gd)$ on each $S_v$.
\item on $S_1$ prescribe:
\begin{itemize}
\item an induced metric $\gd$ compatible with $\textrm{conf}(\gd)$,
\item scalar functions $\tr\chi, \tr\chib, \omb$ and $\Db\omb$,
\item an $S_1$-tangential vectorfield $\eta$,
\item $S_1$-tangential tracefree symmetric $2$-tensors $\ab$ and $\chibh$.
\end{itemize}
\end{itemize}
A \emph{$C^3$-characteristic seed} additionally includes the prescription on $S_{1}$ of
\begin{itemize}
\item a scalar function $\Db^2\omb$,
\item a tracefree symmetric $2$-tensor $\Dbh\ab$.
\end{itemize}
\end{definition}
\begin{remark}
The prescription of $\textrm{conf}(\gd)$ is equivalent to the prescription of $\gd_c$ defined in equation \eqref{eq:confgc}.
\end{remark}

\ni Beginning with the characteristic seed given in Definition \ref{def:charseed}, the null structure equations can be combined into a hierarchical order from which all components of sphere data $\mfx_{u,v}$ can be determined along $\HH_{[1,2]}$. The hierarchy consists of a sequence of transport equations along the generators of $\HH$ where each equation depends only on quantities determined from preceding equations in the hierarchy. In this section, we present the equations in their hierarchical order. Combining the first variation equation \eqref{eq:firstvariationeq}
and Raychaudhuri \eqref{eq:raychauduri} equation,
we obtain the following transport equation for $\phi$.
\begin{align*} 
\CC_{\phi}:=& D^2\phi -\om \Om\tr\chi \phi + \frac{1}{2} \Om^2\vert \chih \vert^2 \phi =0.
\end{align*}
Integrating this equation and using the characteristic seed, $\phi$ can be calculated along $\HH$ and thus $\gd$ is specified at each sphere. We obtain the following equations by manipulating the null structure equations. We have the equations for $\Om\tr\chi$ and $\chih$
\begin{align*} 
\CC_{\Om\tr\chi} := & 2\phi D\phi + \frac{\phi^2}{2} \tr_{\gd_c} D\gd_c - \Om \tr\chi \phi^2 =0, \\
\CC_{\chih} := &  -2 \Om \chih + \phi^2 \left(D\gd_c -\frac{1}{2} (\tr_{\gd_c} D\gd_c) \gd_c\right) =0.
\end{align*}
These equations can be calculated at each sphere since $\phi$ is known from $C_\phi$. Then we have the transport equations for $\eta$,
\begin{align*} 
\CC_\eta := & D\eta + \Om \tr\chi \eta - \Om \left(\Divd \chih - \frac{1}{2}\di \trchi+\chih \di \log \Om + \frac{3}{2}\trchi \di \log \Om\right) =0,
\end{align*}
and the equations for $\Om\tr\chib$ and $\chibh$,
\begin{align*} 
\CC_{\Om\tr\chib} :=& D(\Om \tr\chib) + \Om\trchi (\Om\tr\chib) +2\Om^2\Divd (\eta-2\di \log\Om) 
- 2 \Om^2 \vert \eta-2\di \log\Om \vert^2 +2 \Om^2 K =0,\\
\CC_{\chibh} :=&D\left(\Om \chibh\right) - (\Om \chih, \Om \chibh) \gd - \frac{1}{2} \Om \tr\chi \Om \chibh \\
&- \Om^2 \left(\Nabs \widehat{\otimes}(2\di \log\Om-\eta) + (2\di \log\Om-\eta) \widehat{\otimes} (2\di \log\Om-\eta) - \frac{1}{2} \tr\chib \chih\right) =0,
\end{align*}
where $K$ is the Gauss curvature. By definition $\om$ and all $D$ derivatives of $\om$ can be calculated directly from the characteristic seed. For $\omb$, we have the following equation
\begin{align*} 
\CC_\omb := D\omb - \Om^2 \left(4(\eta, \di \log \Om)- 3 \vert \eta \vert^2 + K +\frac{1}{4} \tr\chi \tr\chib -\frac{1}{2} (\chih, \chibh)\right)=0.
\end{align*}

From the Bianchi equations \ref{eq:nullBianchiab}, we obtain an equation for $\ab$ only depending on the quantities already derived hierarchically,
\begin{align*} 
\CC_\ab :=& \widehat{D}\ab - \frac{1}{2} \Om \tr\chi \ab + 2 \om \ab \\
&+ \Om \Nabs \widehat{\otimes} \left(\Divd \chibh - \frac{1}{2} \di \tr\chib - \chibh \cdot (\eta-\di\log\Om) + \frac{1}{2} \trchib (\eta-\di \log\Om)\right)\\
&+ \Om \left(9\di \log \Om -5 \eta\right)\widehat{\otimes} \left(\Divd \chibh - \frac{1}{2} \di \tr\chib - \chibh \cdot (\eta-\di\log\Om) + \frac{1}{2} \tr\chib (\eta-\di \log\Om)\right)\\
&-3\Om \chibh \left(K + \frac{1}{4} \tr\chi \tr\chib - \frac{1}{2} (\chih,\chibh)\right) + 3\Om {}^*\chibh \left(\Curld \eta + \frac{1}{2} \chih \wedge \chibh\right)\\
=&0.
\end{align*}
The following equation for $\Db\omb$ is derived in Appendix B of \cite{ACR2}.
\begin{align*} 
\CC_{\Db\omb} :=& D\Db\omb -12 \Om^3 (\di\log\Om-\eta)\omb -2 \Om^2 \omb \left(2(\eta,-\eta+2\di\log\Om)- \vert \eta\vert^2\right)\\
& +\left(K+ \frac{1}{4} \tr\chi\tr\chib-\frac{1}{2} (\chih,\chibh)\right) \left(\frac{3}{2}\Om^3\tr\chib -2\Om^2 \omb\right) \\
&-12 \Om^3 \chib(\eta, \eta-\di\log\Om) - \frac{1}{2} \Om^3 (\chih, \ab)\\
& - \Om^3 \left(\Divd \chibh-\frac{1}{2} \di\tr\chib-\chibh \cdot \left(\eta- \di \log \Om\right)+ \frac{1}{2} \tr\chib \left(\eta-\di\log\Om\right), 7\eta-3\di\log\Om\right)\\
&- \Om^3 \Divd \left(\Divd \chibh-\frac{1}{2} \di\tr\chib-\chibh \cdot \left(\eta- \di \log \Om\right)+ \frac{1}{2} \tr\chib \left(\eta-\di\log\Om\right)\right)  \\
=&0.
\end{align*}
 $\a$ satisfies the constraint equation
\begin{align*} 
\CC_{\ab} :=& \Om \a + D\chih-\Om \vert \chih \vert^2 \gd - \om \chih=0.
\end{align*}
Finally, we have the following higher-order equations for the quantities $\Dbh\ab$ and $\Db^2\omb$. Firstly, from \eqref{eq:Duhab}, we have that

 \begin{align}
\begin{aligned}
\CC_{\Dbh\ab}:=&\Dh\Dbh\ab - \frac{1}{2}\Om\tr\chi\widehat{\Db}\ab+\frac{1}{4}(\Om\tr\chi)(\Om\tr\chib)\ab-\Om \Nabs\widehat{\otimes} \left( \frac{3}{2} \Om \tr \chib \beb + \Om \Divd \aa\right)\\
& +2\om\widehat{\Db}\ab+4\widehat{\LIE}_{\Om^{2}\ze}\ab-\frac{1}{2}\Om\chibh(\Om\chih,\ab)+\frac{1}{2}\Om\chih(\Om\chibh,\ab)\\
&-2\Om^{2}(2(\et,\et-2\di\log\Om)+|\et-2\di\log\Om|^{2}+\rho)\ab\\
&- \frac{1}{2}\Om^{2}\ab(2\Divd\et+2|\et|^{2}-(\chih,\chibh)+2\rh)\\
&+\Om\omb ( \Nabs \widehat{\otimes} \beb - (5\et-9\di\log\Om) \widehat{\otimes} \beb + 3 \chibh \rh -3 {}^* \chibh \si) \\
&-\Om \bigg\{ \Nabs\widehat{\otimes} (\Om \chibh \cdot \beb +\omb \beb +\Om (\et-2\di\log\Om) \cdot\aa)- 2\Om\chibh\Divd\bb\\
&- 2\widehat{\Db\Gammad}\cdot\bb+(5\Om(\chib\cdot\et+\bb)+\di\om)\widehat{\otimes}\bb\\
&+(5\et-9\di\log\Om)\widehat{\otimes} \left(\frac{3}{2} \Om \tr \chib \beb - \Om \chibh \cdot \beb -\omb \beb + \Om \left( \Divd \aa - (\et-2\di\log\Om) \cdot \aa\right)\right) \\
&+2\Om(\bb,5\et-9\di\log\Om)\chibh+3(\om\chibh-\Om\ab)\rh-3(\om{}^{*}\chibh-\Om{}^{*}\ab)\si\\
&+3\chibh\left(- \frac{3}{2} \Om \tr \chib \rh- \Om \left( \Divd \beb + (\et+\di\log\Om, \beb)+ \frac{1}{2} (\chih,\aa) \right)\right)\\
&-3^{*}\chibh\left(- \frac{3}{2} \Om \tr \chib \si - \Om \left( \Curld \beb + (\et+\di\log\Om, {}^* \beb) + \frac{1}{2}\chih \wedge \aa \right)\right)\bigg\}\\
=&0, \raisetag{9\baselineskip}
\end{aligned}
\end{align}
where $\widehat{\Db\Gammad}$ is the trace-free part of $\Db\Gammad$ with
\begin{equation*}
(\Db\Gammad)^C_{AB}=(\gd^{-1})^{CD}(\Nabs_A(\Omega\chib)_{BD}+\Nabs_B(\Omega\chib)_{AD}-\Nabs_D(\Omega\chib)_{AB}).
\end{equation*}
The equation for $\Db^2\omb$ follows from \eqref{eq:Du2ombcoordu}
\begin{align}
 \begin{aligned}
\CC_{\Db^2\omb}:=& D\Db^{2}\omb +4\Om^{3}({\Om\tr\chib}){\Divd\bb}+3\Om^{2}(\Om\tr\chib)^{2}{\rh}+\Om^{3}\Divd(\Om{\Divd\ab})+12\Om^{2}|\di\omb|^{2}\\
 &+7\Om^{4}|\bb|^{2} +16\Om^{2}(\et-\di\log\Om,\di\Db\omb)+40\Om^{2}\omb(\et-\di\log\Om,\di\omb)\\
 &-24\Om^{2}\chib(3\et-2\di\log\Om,\di\omb)-2\Om^{2}(\omb^{2}+\Db\omb)((\et,-3\et+4\di\log\Om)-\rh)\\
 &-\Om^{3}\left(\beb,12\di\omb-\frac{3}{2}\di\Om\tr\chib\right)-72\Om^{3}\omb\chib(\et,\et-\di\log\Om)-6\Om^{3}\omb{\Divd\bb}\\
 &-6\Om^{3}\omb(\beb,7\et-3\di\log\Om)+\Om^{4}\chib(\bb,38\et-15\di\log\Om)-\frac{1}{2} \Om^{3}(\chibh,\Db\ab) \\
 &+12\Om^{4}(\chib\times\chib)(\et,5\et-4\di\log\Om)+\frac{1}{2}\Om^{4}\tr\chib(\bb,31\et-9\di\log\Om)\\
 &-\Om^{2}\left(9\omb\Om\tr\chib-\frac{3}{2}\Om^{2}|\chih|^{2}\right){\rh}+12\Om^{4}\ab(\et,\et-\di\log\Om)\\
 &+\Om^{4}(\Divd\ab-(\et-2\di\log\Om)\ab,7\et-3\di\log\Om)\\
 &+\Om^{3}\Divd(\Om(\chibh\cdot\bb-(\et-2\di\log\Om)\cdot\ab))+2\Om^{4}(\chib,(\chih\times\ab))\\
 &-\frac{1}{2}\Om^{3}\left(4\omb\chih-\Om\tr\chib\chih+\Om(\Nabs\widehat{\otimes}\et+\et\widehat{\otimes}\et)-\frac{\Om}{2}\tr\chi\chibh,\ab\right)\\
 =&0.
  \end{aligned}
 \end{align}

The $12$ null structure equations stated above yield $12$ null transport equations denoted by $\CC_\varphi$ for $$\varphi\in\{\phi,\Om\tr\chi,\chih,\eta,\Om\tr\chib,\chibh,\omb,\ab,\allowbreak\Db\omb,\Db^2\omb,\a,\Dbh\ab\}.$$ We abuse notation by referring to the collection of functions as $\CC_\varphi$. The following lemma shows that $C_\varphi$ is a smooth mapping and defines the space they map into.

\begin{lemma}[Smoothness of constraint functions] \label{lem:codomainconstraints} Consider $C^3$-null data $\mfx$ on $\HH_{0,[1,2]}$. The constraint map $\mathcal{C}$ with
\begin{align*} 
\begin{aligned} 
\mathcal{C}: \mfx \mapsto \CC_\varphi(\mfx),
\end{aligned}
\end{align*}
is a smooth map from an open neighbourhood of $\mathfrak{m}$ in $\mathcal{X}(\HH_{0,[1,2]})$ to $\mathcal{Z}_{\CC}$, where
\begin{align*} 
\begin{aligned} 
\mathcal{Z}_{\CC}:=&  H^8_3 \times  H^8_4 \times H^8_3 \times H^{7}_3 \times  H^{6}_4 \times  H^{6}_3 \times  H^{6}_3\times H^{4}_3 \times  H^{4}_3 \times  H^8_2 \times H^{2}_{3} \times H^{2}_{3},
\end{aligned} 
\end{align*}
Furthermore, we have the Lipschitz estimate
\begin{align*} 
\begin{aligned} 
\Vert \CC(\mfx) \Vert_{\mathcal{Z}_{\CC}} \les \Vert \mfx - \mathfrak{m} \Vert_{\XX(\HH_{0,[1,2]})}.
\end{aligned}
\end{align*}
\end{lemma}

\subsection{Sphere perturbations}\label{sec:spherepert}

In this section, we introduce sphere perturbations and define the perturbation map. We consider two types of sphere perturbations. Firstly, we consider \emph{transversal perturbations} along an ingoing null hypersurface $\HHb_{2}$. Consider a local double null coordinate system $(\ut,\vt,\tht^1,\tht^2)$ about a sphere $\tilde{S}_{u_0,2}\subset\HHb_{2}$ with $\tilde{S}_{u_0,2}=\{\ut=u_0,\vt=2\}$. For a sufficiently small scalar function $f$, the transformation on $\HHb_{u,2}$ given by
\begin{align} 
\begin{aligned} 
\tilde{u}=u+f(u,\th^1,\th^2), \quad \vt=v, \quad \tht^1=\th^1, \quad \tht^2=\th^2,
\end{aligned}
\end{align}
yields a local double null coordinate system $(u,v,\th^1,\th^2)$ on $\HHb_{2}$. We denote by $S_{u_0,2}$ the sphere on $\HHb_{2}$ with $S_{u_0,2}:=\{u=u_0,v=2\}$. The transverse perturbation is not an intrinsic gauge transformation of the sphere $\tilde{S}_{u_0,2}$ but a perturbation to a nearby sphere $S_{u_0,2}$ on $\HHb_{2}$. The sphere data $\mfx_{u_0,2}$ on  $S_{u_0,2}$ can be expressed in terms of the function $f$, called a \emph{perturbation function} and the sphere data $\tilde{\mfx}_{u_0,2}$. The transformations associated with the $C^3$-parts of the sphere data are computed in Appendix \ref{sec:spherepert} and the formulas associated with the $C^2$-sphere data are computed in Appendix A of \cite{ACR2}.

Secondly, we consider \emph{sphere diffeomorphisms} defined using angular perturbation functions $(q_{1},q_{2})$. Consider the $S_{u,v}$-tangential vectorfield $$-(v-u)^{2}\DD_{1}^*(q_{1},q_{2})$$ where the operator $\DDd_1^*$ is defined in \eqref{eq:hodgeops}, and the corresponding $1$-parameter family of diffeomorphisms $$\Phi_{t}(q):S_{u,v}\rightarrow S_{u,v}.$$ The angular perturbation of some sphere data $\mfx_{u,v}$ is the sphere data obtained from the pullback of $\tilde{\mfx}_{u,v}$ under the diffeomorphism $\Phi_{t}(q)$, $$\mfx_{u,v}:=\Phi_{t}^{*}(\tilde{\mfx}_{u,v}).$$

So far, we have shown that sphere data obtained from both transversal and angular perturbations depends on perturbation functions $f,q_1$ and $q_2$ and some sphere data $\tilde{\mfx}_{u,v}$. An analysis of the sphere data transformation formulas, see Appendix \ref{sec:transversalspherepert} and Appendix A of \cite{ACR2}, shows that $f$ only appears in the perturbation formulas as the five scalar functions,
\begin{equation*}
(f(0),\pr_{u} f(0),\pr_{u}^{2} f(0),\pr_{u}^{3} f(0),\pr_{u}^{4} f(0)),
\end{equation*}
suppressing the angular coordinates and setting $u=0$. The \emph{perturbation map} is then defined as the combination of transversal and angular perturbations 
\begin{equation}\label{eq:defpertmap}
\mfx_{0,2}:=\PP_{f,q}(\tilde{\mfx}_{0,2})
\end{equation}
where, we abuse notation to write
\begin{align*}
f:=&(f(0),\pr_{u} f(0),\pr_{u}^{2} f(0),\pr_{u}^{3} f(0),\pr_{u}^{4} f(0)),\\
q:=&(q_{1},q_{2}).
\end{align*}
We provide explicit formulas and estimates in Appendix \ref{sec:spherepert}.  The norms used for the perturbation functions are defined as follows.
\begin{definition}[Norms for perturbation functions] \label{def:pertfuncnorm} On a sphere $S_2$, define
\begin{enumerate}
\item for a transversal perturbation function $f$ with $k\leq4$,
\begin{align}
\begin{aligned}
\Vert f\Vert_{\YY_{f^{(k)}}}:=&\sum_{n=0}^k \Vert \partial_u^n f(0)\Vert_{H^{10-2n}(S_{2})}
\end{aligned}\label{eq:pertnormf}
\end{align}
with respect to the round unit metric and $
\YY_{f^{(k)}}:=\{f : \Vert f \Vert_{\YY_{f^{(k)}}}<\infty\}.
$
\item for an angular perturbation function $q$,
\begin{equation}\label{eq:pertnormq}
\Vert q\Vert_{\YY_{q}}:=\Vert q_{1}\Vert_{H^{10}(S_{2})}+\Vert q_{2}\Vert_{H^{10}(S_{2})}
\end{equation}
with respect to the round unit metric and 
$\YY_{q}:=\{q : \Vert q \Vert_{\YY_{q}}<\infty\}.$
\end{enumerate}
\end{definition}
\ni The following lemma pertains to the smoothness of the perturbation map. Its proof follows from the use of the higher regularity norm from Definition \ref{def:nullnormhighreg} and the formulas of Appendix \ref{sec:transversalspherepert}.
\begin{lemma}[Smoothness of $\PP_{f,q}$] \label{lemma:pertestimate}
Let $\de>0$ be a real number, the map
\begin{align*} 
\PP_{f,q}: \, \XX^+(\tilde{\HHb}_{[-\de,\de],2}) \times \YY_{f^{(4)}} \times \YY_q &\to \XX(S_{0,2}),
\end{align*}
is well-defined and smooth in an open neighbourhood of $(\tilde{\underline{\mfx}},f,q)=(\underline{\mathfrak{m}},0,0)$. Furthermore, we have the bound
\begin{align} 
\Vert \PP_{f,q}(\tilde{\underline{\mfx}}) - \tilde{\underline{\mfx}}\vert_{S_{2}} \Vert_{\XX(S_{0,2})} \les \Vert f \Vert_{\YY_{f^{(4)}}} +\Vert q \Vert_{\YY_q}+ \Vert \tilde{\underline{\mfx}}-\underline{\mathfrak{m}} \Vert_{\XX^+(\tilde{\HHb}_{[-\de,\de],2})},
\label{eq:pertbound}
\end{align}
\end{lemma}

\subsection{Linearised null structure equations and sphere perturbations}\label{sec:linearNSE}

In this section we derive the linearised formulas of null structure equations about the reference Minkowski sphere data $\mathfrak{m}$. We follow the linearisation procedure from \cite{DHR}. Formally, for the $C^{3}$-sphere data, we write
 \begin{align}
 \begin{aligned}\label{eq:formallin}
 \mfx_{u,v} =&(\Om,\gd, \Om\trchi, \chih, \Om\trchib, \chibh, \eta, \om, D\om, D^{2}\om, \omb, \Db\omb, \Db^{2}\omb, \a, \widehat{D}\a, \ab, \widehat{\Db}\ab) \\
 =& \left(1,v^{2}\gac,\frac{2}{v},0,-\frac{2}{v}, 0,0,0,0,0,0,0,0,0,0,0,0\right)\\
 &+ \varep \cdot \left(\dot\Om, \dot{\gd}, \omtrchid, \dot\chih, \omtrchibd, \dot\chibh, \dot\eta, \omd, D\omd, D^{2}\omd,\ombd,\Db\ombd,\Db^{2}\ombd,  D\ad,\Db \abd\right)\\
 &+ \mathcal{O}(\varep^2).
 \end{aligned}
 \end{align}
and differentiate in $\varep$ at $\varep=0$. The proof of the following lemma of the linearisation of the higher-order null structure equations, $\CC_{\Dbh\ab}$ and $\CC_{\Db^2\omb}$ follows from applying this linearisation procedure.

\begin{lemma}[Linearisation of $\CC_{\Dbh\ab}$ and $\CC_{\Db^2\omb}$ at Minkowski] \label{lem:linearisationofconstraints} Denote by $\dot{\CC}_{\Db\abd}$ and  $\dot{\CC}_{\Db^2\ombd}$ the linearisation of the null structure equations $\CC_{\Dbh\ab}$ and $\CC_{\Db^2\omb}$, respectively at Minkowski. We have the following formulas.
\begin{align} 
\begin{aligned} 
\dot{\CC}_{\Db\abd}=&vD\left(\frac{1}{v}\Db\abd\right)- \frac{1}{v^{2}}\abd + 2\DDd_{2}^{*}\left(-\frac{3}{v}\left(\frac{1}{v^{2}}\Divdo\chibhd-\frac{1}{2}\di\omtrchibd-\frac{1}{v}\etad\right)+\frac{1}{v^{2}}\Divdo\abd\right),\\
\dot{\CC}_{\Db^2\ombd}=& D(\Db^{2}\dot{\om})-\frac{12}{v^{2}}\left(\Kd+\frac{1}{2r}\dot{(\Om\tr\chib)}-\frac{1}{2r}\dot{(\Om\tr\chi)} +\frac{2}{v^{2}}\dot{\Om}\right)\\
&- \frac{8}{v^{3}}\Divdo\left(\frac{1}{v^{2}}\Divdo\dot{\chibh}- \frac{1}{2}\di\dot{(\Om\tr\chib)}-\frac{1}{v}\dot{\et}\right)+\frac{1}{v^{4}}\Divdo\Divdo\dot{\ab}.
\end{aligned}\label{eq:newlinearisedconstraints}
\end{align}
\end{lemma}

Using the same linearisation procedure, we linearise the preceding $10$ null structure equations from Section \ref{sec:charseed}. Analogously to \eqref{eq:newlinearisedconstraints}, we define their linearisations by
 $$\dot{\CC}_{\dot{\varphi}} \mbox{ for } \varphi\in\{\dot{\phi},\dot{(\Om\tr\chi)},\chihd,\etad,\dot{(\Om\tr\chib)},\allowbreak\chibhd, \ombd,\Db\ombd,\ad, \abd\}.$$
An explicit presentation of their formulas is given in Appendix \ref{sec:linCC}, see also Section 2.7 of \cite{ACR2}.

Similarly, the linearisation of $\PP_{f,q}(\tilde{\underline \mfx})$ about $f=0$ at Minkowski and $q=0$ at Minkowski is computed in the following lemma.

\begin{lemma}[Linearisation of $\PP_{f,q}$ about $f=0$ and Minkowksi] \label{lem:pertlin} 

Let $\dot{f}$ denote the linearisation of $f$ at $f=0$ so that
\begin{align*} 
\dot{f} = \, (\dot{f}(0), \partial_u \dot{f}(0), \partial_u^2 \dot{f}(0), \partial_u^3 \dot{f}(0), \partial_{u}^{4}\dot{f}(0)).
\end{align*}
and let $\dot q$ be the linearisation of $q$ about $(q_{1},q_{2})=(0,0)$ at Minkowski. Evaluating at $u=0$, the non-trivial components of $\dot{\PP}_{f,q}(\dot{f},\dot{q})$ are
\begin{align*}
 \begin{aligned} 
\Omd =& \frac{1}{2} \partial_{u} (\dot{f}), & \phid=&- \dot{f}+\frac{v}{2} \Ldo \dot{q}_1, & \etad=&  r \di \left(\partial_u \left(\frac{f}{v}\right)\right),\\
\chihd =& - 2 \DDd_2^* \di \dot{f}, & \omtrchibd =& -2\partial_u \left(\frac{f}{v}\right), & \omtrchid=& \frac{2}{v^2} (\Ldo+1) \dot{f},\\
\ombd =&\partial_u(\frac{1}{2} \partial_{u}\dot{f} ),& \Db\ombd =& \partial_u^2 (\frac{1}{2} \partial_{u} (\dot{f}), & \Db^{2}\ombd = & \partial_u^{3} (\frac{1}{2} \partial_{u} (\dot{f})).\\
 \gdcd = &2 \DDd_2^* \DDd_1^* (\dot{q}_1,\dot{q}_2).&&&&
\end{aligned}
\end{align*}
\end{lemma}
The following lemma follows by Definition \ref{def:spheredata} of $C^3$-sphere data,.
\begin{lemma}[Bounds for linearised perturbation functions] \label{lem:estimatespfq} 
The linearised perturbation map $\dot{\PP}_{f,q}(\dot{f},\dot{q})$ satisfies the estimates
\begin{align*} 
\Vert \dot{\PP}_{f,q}(\dot f) \Vert_{\XX(S_{0,2})} \les \Vert \dot f \Vert_{\YY_{f^{(4)}}}+ \Vert \dot{q} \Vert_{\YY_{q}}.
\end{align*}
\end{lemma}

\subsection{$C^3$-matching data}\label{sec:C3matching data}

To solve the $C^3$-null gluing problem up to the $20$-dimensional space of charges $(\mathbf{E},\mathbf{P},\mathbf{L},\mathbf{G},\mathbf{U},\mathbf{V})$, we define $C^3$-matching data to encapsulate all of the \emph{gluable} parts of the $C^3$-sphere data.

\begin{definition}[Matching map $\mathfrak{M}$] \label{def:matchingmap} Let $\mfx_{u,v}$ be $C^2$- and $C^3$-sphere data on the sphere $S_{u,v}$. Define the \emph{matching map} denoted by $\mathfrak{M}_{u,v}:=\mathfrak{M}(x_{u,v})$ by
\begin{align*} 
\begin{aligned} 
\mathfrak{M}(\mfx_{u,v}):= m_{u,v}\times(D^{2}\om,\Db^{2}\omb^{[\geq 2]}, \tilde{\QQ}_{\Db^{2}\ombd^{[1]}},\Dh\a,\Dbh\ab^{[\geq 3]})
\end{aligned} 
\end{align*}

where the quantity $\tilde{\QQ}_{\Db^{2}\ombd} $ is defined by
\begin{align*} 
\tilde{\QQ}_{\Db^{2}\ombd^{[1]}}  :=& \left(\Db^{2}\omb\right)^{[\leq1]} -\frac{1}{2} (\Ldo-3) \left(\frac{1}{v^{2}}\Om\trchib- \frac{2}{v^4}(\Ldo+2)\phi\right)^{[\leq1]} \\
&+ \frac{1}{2v^{2}} \left(\Ldo\Ldo + 2\Ldo -3\right) \left(\Om\trchi-\frac{4}{v}\Om\right)^{[\leq1]} -\frac{1}{v^3} \Divdo \eta^{[1]},
\end{align*}
and $m_{u,v}$ denotes the $C^2$-matching map
\begin{align}
\begin{aligned}
m_{u,v}:=&\left(\Om, \phi, \gd_c, \Om\trchi, \chih, (\Om\trchib)^{[\geq2]}, \chibh, \eta^{[\geq2]}, \om, D\om, \omb^{[\geq 2]}, \Db\omb^{[\geq 2]},\right.\\
&\left. \tilde{\QQ}_{\omb^{[\leq1]}}, \tilde{\QQ}_{\Db\omb^{[\leq1]}}, \a, \ab, \right),
\end{aligned}\label{eq:C2matching}
\end{align}
with $\tilde{\QQ}_{\omb^{[\leq1]}}$ and $\tilde{\QQ}_{\Db\omb^{[\leq1]}}$ given by
\begin{align}
\begin{aligned}
\tilde{\QQ}_{\omb^{[\leq1]}}:=& \omb^{[\leq1]} + \frac{1}{4}\left(\Om\trchib\right)^{[\leq1]}- \frac{1}{6r} \Divdo \eta^{[1]} \\
&- \frac{1}{12v^3} (\Ldo+3) \left(\Om\trchi-\frac{4}{v}\Om\right)^{[\leq1]} - \frac{1}{2v^2}(\Ldo+2)\phi^{[\leq1]}, \\
\tilde{\QQ}_ {\Db\omb^{[\leq1]}}:=& \left(\Db\omb\right)^{[\leq1]} -\frac{1}{6} (\Ldo-3) \left(\frac{1}{v}\Om\trchib- \frac{2}{v^3}(\Ldo+2)\phi\right)^{[\leq1]} \\
&+ \frac{1}{6r} \left(\Ldo\Ldo + \Ldo -3\right)\left(\Om\trchi-\frac{4}{v}\Om\right)^{[\leq1]} -\frac{2}{3v^2} \Divdo \eta^{[1]}.
\end{aligned}\label{eq:matchingQ}
\end{align}
The operators $\Divdo$ and $\Ldo$ are the divergence and Laplace-Beltrami operator with respect to $\gac$ on $S_{u,v}$.
\end{definition}
The proof that the $C^3$-matching data contains the $C^3$-sphere data orthogonal to the charges $(\mathbf{E},\mathbf{P},\mathbf{L},\mathbf{G},\mathbf{U},\mathbf{V})$ can be found in Section \ref{sec:matchingdata}. The following lemma shows that the matching map is a linear Lipschitz map.
\begin{lemma}[Smoothness of $\MMf$] Let $v\geq u$ be two real numbers. The matching map $\MMf$ is a smooth, well-defined Lipschitz map from an open neighbourhood of sphere data $\mfx_{u,v}\in\XX(S_{u,v})$ to $\mathcal{Z}_{\mathfrak{M}}(S_{u,v})$ defined by
\begin{align*} 
\begin{aligned} 
\mathcal{Z}_{\mathfrak{M}}(S_{u,v}) =& H^8 \times H^8 \times H^8 \times H^8 \times H^8 \times H^6 \times H^6 \times H^7 \times H^8 \times H^8 \times H^8\\
& \times H^6 \times H^4 \times H^2 \times H^6 \times H^4 \times H^{2} \times H^{8} \times H^{8} \times H^{4} \times H^{2}.
\end{aligned}
\end{align*}
Furthermore, the bound
\begin{align*} 
\begin{aligned} 
\Vert \mathfrak{M}(\mfx_{u,v}) - \mathfrak{M}(\mathfrak{m}) \Vert_{\mathcal{Z}_{\mathfrak{M}}(S_{u,v})} \leq C_{u,v} \Vert \mfx_{u,v} -\mathfrak{m}\Vert_{\XX(S_{u,v})},
\end{aligned} 
\end{align*}
holds where $C_{u,v}>0$ is a constant depending on $u$ and $v$.
\end{lemma}

\section{Precise statement of the main theorem}\label{sec:mainthmprecise}

\ni We are now in a position to formulate the precise statement of our main theorem.
\begin{theorem}[Codimension-$20$ perturbative $C^3$-null gluing] \label{thm:main}
Consider $C^3$-sphere data $\mfx_{1}$ on $S_{1}$ and $\tilde{\mfx}_{2}$ on $\tilde{S}_{2}$ where $\tilde{S}_{2}$ is contained within ingoing null data $\tilde{\underline{\mfx}}$ on $\HHb_{[-\delta,\delta],2}$ for real number $\delta>0$. Assume that for some real number $\varep>0$,
\begin{equation}\label{eq:minkassump}
\Vert \mfx_{1}-\mathfrak{m}\Vert_{\XX(S_{1})}+\Vert \tilde{\underline{\mfx}}-\underline{\mathfrak{m}}\Vert_{\XX^{+}(\tilde{\HHb}_{[-\delta,\delta],2})} <\varep.
\end{equation}
Then for $\varep>0$ sufficiently small there exists
\begin{enumerate}
\item a solution $\mfx$ to the null structure equations on $\HH_{0,[1,2]}$,
\item a perturbed sphere $S_{2}$ with induced sphere data $\mfx_{2}$ arising from $\tilde{\mfx}_{2}$ on $\tilde{S}_{2}$ contained within $\HHb_{[-\delta,\delta],2}$. That is,
$$\mfx_{2}=\PP_{f,q}(\tilde{\underline{\mfx}}).$$
\end{enumerate}
Moreover, the solution $\mfx$ is such that on $S_{1}$ and $S_2$,
\begin{align}\label{eq:gluingS1}
\mfx\vert_{S_{1}}=\mfx_{1},\qquad \mathfrak{M}({\mfx\vert_{S_2}})=\mathfrak{M}({\mfx_2}),
\end{align}
respectively. That is,  if the charges $(\mathbf{E},\mathbf{P},\mathbf{L},\mathbf{G},\mathbf{U},\mathbf{V})$, see Definitions \ref{def:chargesEPLG} and \ref{def:chargeW}, satisfy
$$(\mathbf{E},\mathbf{P},\mathbf{L},\mathbf{G},\mathbf{U},\mathbf{V})(\mfx\vert_{S_{2}})=(\mathbf{E},\mathbf{P},\mathbf{L},\mathbf{G},\mathbf{U},\mathbf{V})(\mfx_{2}),$$
then on $S_2$,
\begin{equation}\label{eq:gluingS2}
\mfx\vert_{S_{2}}=\mfx_{2}.
\end{equation}

Furthermore, we have the bounds
\begin{align}
\begin{aligned}
\Vert \mfx_{1}-\mathfrak{m}\Vert_{\XX(\HH_{[1,2]})}+\Vert\mfx_{2}-\tilde{\mfx}_{2}\Vert_{\XX(S_{2})} \les&\,\varep,\\
\Vert f\Vert_{\YY_{f^{(4)}}}+\Vert q \Vert_{\YY_q}\les&\,\varep
\end{aligned}\label{eq:boundsmainthm}
\end{align}
where the norms are defined in Section \ref{sec:C3spheredata} and \ref{sec:spherepert}.
\end{theorem}

\section{Proof of Theorem \ref{thm:main}}\label{sec:mainthmproof}

\ni In this section, we setup the functional analytic framework and apply the implicit function theorem to prove Theorem \ref{thm:main}. We first define the following functional $\FF$ on which we will apply the implicit function theorem.

\begin{definition}[The functional $\FF$]\label{def:Fmap} Let 
\begin{itemize}
\item $\mfx_{1} \in \mathcal{X}(S_{1})$ be $C^3$-sphere data, 
\item $\tilde{\underline{\mfx}} \in \XX^+(\tilde{\HHb}_{[-\de,\de],2})$ be ingoing null data on $\tilde{\HHb}_{[-\de,\de],2}$,
\item $\mfx \in \XX(\HH_{01,2]})$ be null data on $\HH_{[1,2]}$, 
\item $f \in \YY_{f^{(4)}}$ and $q \in \YY_{q}$ be perturbation functions,
\end{itemize}
where the function spaces are defined in Sections \ref{sec:C3spheredata} and \ref{sec:spherepert}. Define the functional $\FF$ by the following map.
\begin{align*} 
\FF: \mathcal{X}(S_{1}) \times \mathcal{X}^+(\tilde{\HHb}_{[-\de,\de],2}) \times\XX(\HH_{[1,2]}) \times \YY_{f^{(4)}}\times \mathcal{Y}_{q} \to \mathcal{X}(S_{1}) \times \ZZ_\MMf(S_{2}) \times \mathcal{Z}_\CC
\end{align*}
with
\begin{align} 
\FF(\mfx_{1},\tilde{\underline{\mfx}},\mfx,f,q) := \left(\mfx\vert_{S_1} - \mfx_1, \MMf\left(\mfx\vert_{S_2}\right)-\MMf\left(\PP_{f,q}\left(\underline{\tilde{\mfx}}_{2}\right)\right),\CC(\mfx)\right), \label{eq:deff}
\end{align}
where 
\begin{itemize}
\item $\mathfrak{M}$ is the matching map defined in Definition \ref{def:matchingmap},
\item $\CC$ is the collection of constraint functions defined in Lemma \ref{lem:codomainconstraints},
\item $\PP_{f,q}$ are the perturbations of sphere data defined by equation \eqref{eq:defpertmap}.
\end{itemize}
\end{definition}

\ni It is clear that the Minkowski data lies in the kernel of $\FF$, that is,
$$\FF(\mfm_{1},\tilde{\underline{\mfm}},\mfm,0,0) = (0,0,0).$$
Moreover, by definition, elements of the kernel of $\FF$ are solutions to the $C^3$-null gluing problem. That is, if it holds that 
$$\FF\left(\mfx_1,\underline{\mfx}_{2},\mfx,f,q\right) = (0,0,0),$$
then an inspection of the right-hand side of \eqref{eq:deff} yields
$$\mfx\vert_{S_1}=\mfx_1,\qquad \MMf(\mfx\vert_{S_2})=\MMf\left(\PP_{f,q}\left(\underline{\tilde{\mfx}}_{2}\right)\right),\qquad \CC(\mfx)=0.$$
This is in complete correspondence with a solution to the $C^3$-null gluing problem given in Theorem \ref{thm:main}.

In order to apply the implicit function theorem to $\FF$, we need to show that the linearisation of $\FF$. denoted by $\dot{\FF}$ at Minkowski is a bijection. The key observation is that proving the surjectivity of $\dot\FF$ is equivalent to solving the linearised $C^3$-null gluing problem at Minkowski and in Section \ref{sec:lingluing}, we will prove the following proposition.
\begin{proposition}[Surjectivity at Minkowski]\label{prop:surjection} Let $\FF$ be the map given by \eqref{eq:deff}. The linearisation of $\FF$ at Minkowski given by the map
$$\dot{\FF}:(\dot{\mfx},\dot{f},\dot{q})\mapsto\left(\dot{\mfx}\vert_{S_1},\MMf(\mfxd\vert_{S_2} - \dot{\PP}_{f,q}(\dot f,\dot q)),\dot{\CC}(\dot{\mfx})\right)$$
is a surjection. That is, the linearised $C^3$-null gluing problem at Minkowski is solvable.
\end{proposition}
With this result in hand, it remains to show that $\dot{\FF}$ is injective. In general $\dot{\FF}$ is not injective and we must consider the following suitable restriction of the domain. The obstruction to injectivity is the kernel of $\dot{\FF}$. Since $\dot{\FF}$ is a bounded linear map its kernel is a closed subspace of the domain and we can split the domain of $\dot{\FF}$ into 
\begin{align*} 
\begin{aligned} 
\XX(\HH_{[1,2]}) \times \YY_{f^{(4)}}\times \mathcal{Y}_{q} = \mathrm{ker}(\dot{\FF}) \oplus \left(\mathrm{ker}(\dot{\FF})\right)^\perp.
\end{aligned}
\end{align*}
Hence, we consider the following restriction of $\FF$ to
\begin{align*} 
\begin{aligned} 
\overline{\FF}: \mathcal{X}(S_{1}) \times \mathcal{X}^+(\tilde{\HHb}_{[-\de,\de],2}) \times \left(\mathrm{ker}(\dot{\FF})\right)^\perp \to \mathcal{X}(S_{1}) \times \ZZ_\MMf(S_{2}) \times \mathcal{Z}_\CC.
\end{aligned} 
\end{align*}
By construction, the linearisation of $\overline{\FF}$ is injective and thus $\dot{\overline{\FF}}$ is a bijection between Hilbert spaces and we can apply the implicit function theorem to it about $(\mathfrak{m},\underline{\mathfrak{m}},\mathfrak{m},0,0)$. Consider the functional $\GG$ defined by
\begin{align*} 
\GG: B\left((\mathfrak{m},\underline{\mathfrak{m}},0),r_0\right) \to \left(\mathrm{ker}(\dot{\FF})\right)^\perp,
\end{align*}
where $B\left((\mathfrak{m},\underline{\mathfrak{m}},0),r_0\right)$ denotes the open ball of radius $r_0>0$ centered at $(\mathfrak{m},\underline{\mathfrak{m}},0)$ in the space
\begin{align*} 
 \mathcal{X}(S_{1}) \times \mathcal{X}^+(\HHb_{[-\de,\de],2})\times \ZZ_\CC.
\end{align*}
Then for a given $(\mfx_{1},\tilde{\underline{\mfx}}) \in B\left((\mathfrak{m},\underline{\mathfrak{m}},0),r_0\right)$, define 
\begin{align*} 
(\mfx,f,q):= \GG(\mfx_{1},\tilde{\underline{\mfx}}).
\end{align*} 
Applying the implicit function yields that for all $(\mfx_{1},\tilde{\underline{\mfx}}) \in B\left((\mathfrak{m},\underline{\mathfrak{m}},0),r_0\right)$, 
\begin{align*} 
\overline{\FF}\left(\mfx_{1},\tilde{\underline{\mfx}}, \GG(\mfx_{1},\tilde{\underline{\mfx}},0)\right)= (0,0,0).
\end{align*}
Thus, by the definition of $\overline{\FF}$, we have that
\begin{align*} 
&\mfx \vert_{S_{1}} = {\mfx_{1}}, \qquad \MMf(\mfx\vert_{S_{2}}) = \MMf\left(\PP_{f,q}(\tilde{\underline{\mfx}})\right),\qquad \CC(\mfx) = 0 
\end{align*}

\ni This proves the gluing parts \eqref{eq:gluingS1} and \eqref{eq:gluingS2} of Theorem \ref{thm:main}.
Moreover, by standard calculus estimates on Hilbert spaces, we have the bound
\begin{align} 
\begin{aligned} 
&\Vert \mfx - \mathfrak{m} \Vert_{\XX(\HH_{0,[1,2]})} + \Vert f \Vert_{\YY_{f^{(4)}}} + \Vert q \Vert_{\mathcal{Y}_{q}} \\
\les& \Vert \mfx_{1}-\mathfrak{m} \Vert_{\XX(S_{1})} + \Vert \tilde{\underline{\mfx}}-\underline{\mathfrak{m}} \Vert_{\tilde{\mathcal{X}}^+(\HHb_{[-\de,\de],2}) }.
\end{aligned} \label{eq:finalbound}
\end{align} 
Applying the estimate \eqref{eq:pertbound}, it follows that 
\begin{equation*}
\Vert \PP_{f,q}(\tilde{\underline{\mfx}}) -\tilde{\underline{\mfx}}_{2} \Vert_{\XX(S_{2})} \les \Vert \mfx_{1}-\mathfrak{m}_{1} \Vert_{\XX(S_{1})} + \Vert \tilde{\underline{\mfx}}-\underline{\mathfrak{m}} \Vert_{\tilde{\mathcal{X}}^+(\HHb_{[-\de,\de],2}) }.
\end{equation*}
\ni Combining these estimates with the assumption \eqref{eq:minkassump} obtains the bound \eqref{eq:boundsmainthm} and concludes the proof of Theorem \ref{thm:main}. It remains to prove Proposition \ref{prop:surjection}, which we do in the next section.

\section{Surjectivity at Minkowski}\label{sec:lingluing}

\ni In this section, we prove Proposition \ref{prop:surjection} by proving the solvability of the linearised $C^3$-null gluing problem which we formulate as follows.

\begin{proposition}[Linearised $C^3$-null gluing]\label{prop:lingluing} Given
\begin{itemize}
\item linearised $C^3$-sphere data on $S_{1}$, $\dot{\mfx}_{1}\in\XX(S_{1})$,
\item linearised $C^3$-matching data on $\tilde{S}_{2}\subset\HHb_{[-\delta,\delta],2}$, $\tilde{\mathfrak{M}}_{2}\in\ZZ_\mathfrak{M}(\tilde{S}_{2})$,
\item linearised source functions $\mfc_{\dot{\varphi}}\in\ZZ_{\CC}$ on $\HH_{[1,2]}$ for $$\dot{\varphi}\in\{\dot{\phi},\dot{(\Om\tr\chi)},\chihd,\etad,\dot{(\Om\tr\chib)},\allowbreak\chibhd, \ombd,\Db\ombd,\Db^2\ombd, \ad, \abd,\Db\abd\},$$
\end{itemize}
there exists linearised null data $\dot\mfx\in\XX(\HH_{[1,2]})$ and perturbation functions $\dot{f}$ and $\dot{q}$
such that 
\begin{align}\label{eq:spheredatamatching}
\dot\mfx\big\vert_{S_{1}} = \dot\mfx_{1}, \qquad
{\mathfrak{M}}(\dot\mfx\big\vert_{S_{2}} + \PP_{f,q}(\dot{f},\dot{q})) = {\tilde{\mathfrak{M}}}_{2},
\end{align}
and $\dot\mfx$ solves the linearised null structure equations for the source functions
\begin{align}\label{eq:inhomconstraints}
\dot{\CC}_{\dot{\varphi}}(\dot\mfx) = \mfc_{\dot\varphi}(\dot\mfx).
\end{align}
Furthermore, $\dot\mfx$ and the perturbations are bounded as follows,
\begin{align}
\begin{aligned}
&\Vert\dot\mfx\Vert_{\XX(\HH_{[1,2]})} + \Vert \dot{f}\Vert_{\YY_{f^{(4)}}}+\Vert\dot{q}\Vert_{\YY_q} + \Vert \PP_{f,q}(\dot{f},\dot{q})\Vert_{\XX(S_{2})},\\
 \les& \Vert\dot\mfx_{1}\Vert_{\XX(S_{1})} + \Vert{\tilde{\mathfrak{M}}}_{2}\Vert_{\ZZ_{\mathfrak{M}}(\tilde{S}_{2})} + \Vert \mfc_{\dot\varphi}\Vert_{\ZZ_{\CC}}.
\end{aligned}\label{eq:estimatelingluing}
\end{align}
\end{proposition}

\begin{remark} We make the following remarks on Proposition \ref{prop:lingluing}.
\begin{enumerate}
\item In the linearised setting, the characteristic seed along $\HH_{[1,2]}$ can be taken to be the pair $(\Omd,\chihd)$ as a result of the linearised first variation equation, see the third equation of \eqref{eq:linearisedconstraints}. Since we will construct the characteristic seed  along $\HH_{[1,2]}$ and $(\Omd,\chihd)$ are given on $S_1$ and $S_2$, the gluing of the $D$ derivative of these quantities is trivial. That is the quantities 
$$(\omd,D\omd,D^2\omd,\chihd,D\chihd)$$ are glued at $S_2$ given that the characteristic seed $(\Omd,\chihd)$ matches the data given on $S_2$.

\item The proof of Proposition \ref{prop:lingluing} employs the analysis of the linearised $C^2$-null gluing problem of Theorem 4.1 of \cite{ACR2} as well as an analysis of the quantities novel in the $C^3$-sphere data. Therefore, we analyse the novel components of the $C^3$-sphere data $$(D^2\omd,D\ad.\Db^2\ombd, \Db\ab),$$ and employ the linearised $C^2$-null gluing result of Theorem 4.1 in \cite{ACR2} to glue the linearised $C^2$-sphere data.
\end{enumerate}
\end{remark}

We will prove Proposition \ref{prop:lingluing} by splitting the given data into its constituent $C^2$- and $C^3$-parts. The $C^2$-part has already been solved for in \cite{ACR2} and we collect the necessary formulas we need in Section \ref{sec:C2lingluing}. The gluing of the $C^3$-part is contained in the following proposition.

\begin{proposition}\label{prop:C3lin2}
Denote by $\check{x}_1$ the $C^3$-part of the linearised sphere data $$\check{x}_1:= (D^2\omd, \Db^2\ombd,\Dh\ad,\Dbh\abd)$$ on $S_1$. Denote by $\tilde{\check{m}}_2$the $C^3$-part of the matching data $$\tilde{\check{m}}_2:=(D^{2}\om,\Db^{2}\omb^{[\geq 2]}, \tilde{\QQ}_{\Db^{2}\ombd^{[1]}},\Dh\a,\Dbh\ab^{[\geq 3]})$$ on $\tilde{S}_2$.
 Let $\check{x}_1$, $\tilde{\check{m}}_2$ and the higher-order source functions $\mfc_{\Db^2\ombd}$ and $\mfc_{\Db\abd}$ on $\HH_{[1,2]}$ be given. Then there exists linearised null data $\dot\mfx\in\XX(\HH_{[1,2]})$ and perturbation functions $\dot{f}$ and $\dot{q}$ such that the restriction of $\dot\mfx$ to its $C^3$-part, $\check{x}$, satisfies
\begin{align}\label{eq:spheredatamatchingcx}
\check{x}\vert_{S_{1}} = \check{x}_1, \qquad
\check{m}(\dot{\mfx}\vert_{S_{2}} + \PP_{f,q}(\dot{f},\dot{q})) = \tilde{\check{m}}_{2},
\end{align}
and $\dot\mfx$ solves the linearised null structure equations for the source functions
\begin{align}\label{eq:inhomconstraintsxc}
\dot{\CC}_{\Db^2\ombd}(\dot\mfx) = \mfc_{\Db^2\ombd}(\dot\mfx)\qquad \dot{\CC}_{\Db\abd}(\dot\mfx) = \mfc_{\Db\abd}(\dot\mfx).
\end{align}
\end{proposition}
This section is organised as follows.
\begin{itemize}
\item In Section \ref{sec:C2lingluing}, we present the necessary formulas and estimates from the linearised $C^2$-null gluing problem of \cite{ACR2} which we will need to prove Proposition \ref{prop:lingluing}
\item In Section \ref{sec:C3repformulas} we derive transport equations and representation formulas for the novel $C^3$ components of the sphere data, $\check{x}$.
\item In Section \ref{sec:conservationlaws}, we analyse the novel conservation laws that appear in the representation formulas derived in Section \ref{sec:C3repformulas}.
\item In Section \ref{sec:Qgaugedependence}, we investigate the gauge dependence of the conserved quantities under the sphere perturbations of Section \ref{sec:spherepert}.
\item Finally, in Section \ref{sec:lingluingproof}, we apply the results of the preceding sections in order to construct a characteristic seed $(\Omd,\chihd)$ and associated solution to the null structure equations along $\HH_{[1,2]}$. The solution is constructed so that it satisfies the conditions \eqref{eq:spheredatamatching} and \eqref{eq:inhomconstraints}. Moreover, we prove the bound \eqref{eq:estimatelingluing} and conclude the proof of Proposition \ref{prop:lingluing}.
\end{itemize}

\subsection{The solution of the linearised $C^2$-null gluing problem}\label{sec:C2lingluing}

In this section, we collect the necessary results and estimates in order to apply the linearised $C^2$-null gluing problem of \cite{ACR2} as a black box. We may write the main result of their linear analysis (Theorem 4.1 of \cite{ACR2}) as follows.

\begin{proposition}\label{prop:C2lingluing}
Given 
\begin{itemize}
\item linearised $C^2$-sphere data $\xd_1\in X^8(S_1)$,
\item linearised $C^2$-matching data $\tilde{m}_2\in\ZZ_\MM(\tilde{S}_2)$,
\item linearised source functions $\mfc_{\dot{\varphi}}\in\ZZ_{\CC}$ on $\HH_{[1,2]}$ for $$\dot{\varphi}\in\{\dot{\phi},\dot{(\Om\tr\chi)},\chihd,\etad,\dot{(\Om\tr\chib)},\allowbreak\chibhd, \ombd,\Db\ombd, \ad, \abd\},$$
\end{itemize}
there exists linearised null data $\xd\in\XX(\HH_{[1,2]})$ and perturbation functions $\dot{f}$ and $\dot{q}$
such that 
\begin{align}\label{eq:C2spheredatamatching}
\xd\vert_{S_{1}} = \xd_{1}, \qquad
{m}(\xd\vert_{S_{2}} + \PP_{f,q}(\dot{f},\dot{q})) = \tilde{m}_{2},
\end{align}
and $\xd$ solves the linearised null structure equations for the source functions
\begin{align}\label{eq:C2inhomconstraints}
\dot{\CC}_{\dot{\varphi}}(\dot\mfx) = \mfc_{\dot\varphi}(\dot\mfx),
\end{align}
with $\dot{\varphi}\in\{\dot{\phi},\dot{(\Om\tr\chi)},\chihd,\etad,\dot{(\Om\tr\chib)},\allowbreak\chibhd, \ombd,\Db\ombd, \ad, \abd\}$. Furthermore, $\xd$, $\dot f$ and $\dot q$ are bounded as follows,
\begin{align}
\begin{aligned}
&\Vert\xd\Vert_{X^8_3(\HH_{[1,2]})} + \Vert \dot{f}\Vert_{\YY_{f^{(3)}}}+\Vert\dot{q}\Vert_{\YY_q} + \Vert \PP_{f,q}(\dot{f},\dot{q})\Vert_{X^8(S_{2})},\\
 \les& \Vert\xd_{1}\Vert_{X^8(S_{1})} + \Vert\tilde{m}_2\Vert_{\ZZ_{\mathfrak{M}}(\tilde{S}_{2})} + \Vert \mfc_{\dot\varphi}\Vert_{\ZZ_{\CC}}.
\end{aligned}\label{eq:C2estimatelingluing}
\end{align}
\end{proposition}

The proof of Proposition \ref{prop:C2lingluing} can be found in Section 4 of \cite{ACR2}. We will make use of the following representation formula for $\Db\ombd$ from equation (4.40) in \cite{ACR2}.
\begin{align} 
\begin{aligned} 
&\left[ \Db\ombd -\frac{1}{6v^3} \left(\Ldo-3\right)\QQ_{\omtrchibd}+ \frac{1}{2v^2} \Divdo \Divdo \QQ_\chibhd +\frac{1}{v^2} \Divdo \Divdo \DDd_2^* \di \QQ_\phid\right]_1^v\\
&-\left[ \frac{1}{4v^2} \Divdo \left(\di \Divdo -2 + \Divdo\DD_2^*\right) \left(\eta+ \frac{v}{2}\di\left(\omtrchid-\frac{4}{v}\Omd\right)\right) \right]_1^v\\
&+ \left[\frac{1}{8v^2} \Divdo \di \Divdo \Divdo \gdcd+ \frac{1}{2v^3}\left(\frac{1}{12}\Ldo\Ldo-\frac{1}{6}\Ldo+\frac{1}{4}\Divdo \Divdo \DDd_2^* \di - 1\right)\mfc_2\right]_1^v\\
=& \frac{1}{4} \Divdo \left(2-\Divdo \DDd_2^*\right) \Divdo \left(\int_1^v \frac{1}{v'^4}\chihd dv'\right)+\int_1^v h_{\Db\ombd} dv'.
\end{aligned}\label{eq:repformDuomb}
\end{align}
where $h_{\Db\ombd}$ is a function of the source terms $\mfc_{\varphi}$.

\subsection{Representation formulas for higher-order sphere data}\label{sec:C3repformulas}

In this section, we derive representation formulas for  $\Db\abd$ and $\Db^2\ombd$. In order to derive the following representation formulas, consider some given sphere data $\mfx_{0,1}$ on $S_{0,1}$ and linearised characteristic seed $(\dot{\Om},\dot{\chih})$ on $\HH_{0,[1,2]}$. The representation formulas construct the null data along  $\HH_{0,[1,2]}$ used to solve the linearised characteristic gluing problem. We will also make use of the linearised null structure equations \eqref{eq:newlinearisedconstraints} as source terms in the null structure equations. We therefore set $\dot{\CC}_{\dot{\varphi}}=\mfc_{\dot{\varphi}}$. For ease of notation, we evaluate at $u=0$ and set $\HH=\HH_{0,[1,2]}$.

\subsubsection{Analysis of $\protect\Db\abd$.} Recall from Lemma \ref{lem:linearisationofconstraints} that the linearised null constraint equation for $\Db\abd$ is 
\begin{align}
vD\left(\frac{1}{v}\Db\abd\right)+\frac{1}{v^{2}}(2\DDd_{2}^{*}\Divdo-1)\abd -\frac{6}{v}\DDd_{2}^{*}\left(\frac{1}{v^{2}}\Divdo\chibhd-\frac{1}{2}\di\omtrchibd-\frac{1}{v}\etad\right)=\mfc_{\Db\abd}.\label{eq:linNSEDuab}
\end{align}
Using the linearised null structure equations, we obtain the following lemma.

\begin{lemma}[Representation formula of $\Db\abd$]\label{lem:repformDuab}
The null structure equation \eqref{eq:linNSEDuab} implies the following null transport equation for $\Db\abd$
\begin{align}
\begin{aligned}
&D\left(\frac{1}{v}\Db\abd +\frac{1}{v^{2}}(1-2\DDd_{2}^{*}\Divdo)\abd\right)\\
&+D\left(\frac{1}{v^{2}}\DD_{2}^{*}(2-\Divdo\DDd_{2}^{*})\left\{2\Divdo\QQ_\chibhd-\frac{2}{3v}\di\QQ_{\dot{(\Om\tr\chib)}} -2\di(\Ldo+2)\QQ_\phid\right. \right.\\
&\phantom{+D\Biggl(} -\left(\Divdo \DDd_2^* + 1 + \di\Divdo\right) \left(\etad+ \frac{v}{2}\di\left(\omtrchid-\frac{4}{v}\Omd\right)\right)\\
&\phantom{+D\Biggl(}\left.\left.+ \left(\Divdo \DDd_2^* + 1 + \di\Divdo\right)\Divdo\gdcd - \frac{1}{2v}\left(\Divdo \DDd_2^* + 1 + \di\Divdo\right)\di\mfc_{\dot{(\Om\tr\chi)}}\right\}\right)\\
&=\frac{1}{v^{4}}\DDd_{2}^{*}(2-\Divdo\DDd_{2}^{*}) \left(\Divdo \DDd_2^* + 1 + \di\Divdo\right)\Divdo\chihd + h_{\Db\abd},\label{eq:repformulaforDbab}
\end{aligned}
\end{align}
where $h_{\Db\abd}$ is the source term
\begin{align*}
h_{\Db\abd} =& \frac{1}{v}\mfc_{\Db\abd}+\frac{1}{v^{2}}(1-2\DDd_{2}^{*}\Divdo)\mfc_{\abd}\\
&+\frac{1}{v^{2}}(2-\Divdo\DDd_{2}^{*})\left(2\Divdo(D\QQ_\chibhd) -\frac{2}{3v}\di(D\QQ_{\dot{(\Om\tr\chib)}})-2\di(\Ldo+2)(D\QQ_\phid)\right)\\
&+\frac{1}{v^{2}}(2-\Divdo\DDd_{2}^{*})\left(\Divdo \DDd_2^* + 1 + \di\Divdo\right)\left(\frac{1}{v^{2}}\Divdo\mfc_{\chihd} + \frac{3}{2v^{2}}\di\mfc_{\dot{(\Om\tr\chi)}}-\mfc_{\etad}-\di\mfc_{\phid}\right)
\end{align*}
and the quantities $\QQ_\phid$, $\QQ_{\dot{(\Om\tr\chib)}}$ and $\QQ_\chibhd$ are defined in Appendix \ref{sec:C2QQ}.
Integrating the transport equation yields the representation formula
\begin{align}
\begin{aligned}
&\left[ \frac{1}{v}\Db\abd +\frac{1}{v^{2}}(1-2\DDd_{2}^{*}\Divdo)\abd\right]_{1}^{v} \\
&+\left[ \frac{1}{v^{2}}\DD_{2}^{*}(2-\Divdo\DDd_{2}^{*})\left\{2\Divdo\QQ_\chibhd-\frac{2}{3v}\di\QQ_{\dot{(\Om\tr\chib)}} -2\di(\Ldo+2)\QQ_\phid\right. \right.\\
&\phantom{+D\Biggl(} -\left(\Divdo \DDd_2^* + 1 + \di\Divdo\right) \left(\etad+ \frac{v}{2}\di\left(\omtrchid-\frac{4}{v}\Omd\right)\right)\\
&\phantom{+D\Biggl(}\left.\left.+ \left(\Divdo \DDd_2^* + 1 + \di\Divdo\right)\Divdo\gdcd - \frac{1}{2v}\left(\Divdo \DDd_2^* + 1 + \di\Divdo\right)\di\mfc_{\dot{(\Om\tr\chi)}}\right\} \right]^v_{1}\\
&=\DDd_{2}^{*}(2-\Divdo\DDd_{2}^{*}) \left(\Divdo \DDd_2^* + 1 + \di\Divdo\right)\Divdo\left(\int_{1}^{v}\frac{1}{v'^{4}}\chihd dv'\right)+ \int_{1}^{v} h_{\Db\abd} dv'.
\end{aligned}\label{eq:repformDuab}
\end{align}
The linearised null structure equation further implies the estimate 
\begin{equation*}
\Vert \Db\abd \Vert_{\HH^{2}_{4}}\les  \Vert \chihd \Vert_{{H^8_3(\HH)}}+\Vert \Omd \Vert_{{H^8_4(\HH)}}+\Vert \dot{\mfx}_1 \Vert_{\XX^8(S_1)}+  \Vert (\mfc_i)_{\varphi} \Vert_{\mathcal{Z}_{\check{\CC}}}
\end{equation*}
\end{lemma}
\begin{remark}
The quantities $\QQ_\phid$, $\QQ_{\dot{(\Om\tr\chib)}}$ and $\QQ_\chibhd$ are themselves conserved charges appearing in the $C^2$-linearised null gluing problem. In the notation of \cite{ACR2}, they correspond to $\QQ_1$, $\QQ_2$ and $\QQ_3$, respectively.
\end{remark}

\subsubsection{Analysis of $\protect\Db^{2}\protect\omb$} Recall from Lemma \ref{lem:linearisationofconstraints} that the linearised null constraint equation for $\Db^{2}\ombd$ is
\begin{align}
\begin{aligned}
& D(\Db^{2}\dot{\om})-\frac{12}{v^{2}}\left(\Kd+\frac{1}{2v}\dot{(\Om\tr\chib)}-\frac{1}{2v}\dot{(\Om\tr\chi)} +\frac{2}{v^{2}}\dot{\Om}\right)\\
&= \frac{8}{v^{3}}\Divdo\left(\frac{1}{v^{2}}\Divdo\dot{\chibh}- \frac{1}{2}\di\dot{(\Om\tr\chib)}-\frac{1}{v}\dot{\et}\right)-\frac{1}{v^{4}}\Divdo\Divdo\dot{\ab} +\mfc_{\Db^2\ombd}
\end{aligned}\label{eq:linNSEDu2omb}
\end{align}

The null structure equations imply the following lemma. 
\begin{lemma}[Representation formula for $\Db^{2}\omb$]\label{lem:repformDu2omb}
For the linearised null structure equation \eqref{eq:linNSEDu2omb} for $\Db^{2}\ombd$, we have the following transport equation
\begin{align}
\begin{aligned}
&D\left(\Db^{2}\ombd - \frac{1}{2v^{3}}\Divdo\Divdo\abd + \frac{1}{3v^{3}}\Divdo(8-\Divdo\DDd_{2}^{*})\Divdo \QQ_\chibhd \right)\\
-&D\left(\frac{1}{4v^{4}}\left(2(\Ldo-3) +\frac{1}{4}\Ldo(\Ldo+2) \right)\QQ_{\dot{(\Om\tr\chib)}}\right)\\
-&D\left(\frac{1}{3v^{3}}\Divdo\left(8-\Divdo\DDd_{2}^{*}\right)\di(\Ldo+2)\QQ_\phid\right)\\
-&D\left(\frac{1}{5v^{5}}\Divdo\left(8\Divdo\DDd_{2}^{*}-4+8\di\Divdo+\frac{5}{3}\Divdo\DDd_{2}^{*}\left(\Divdo \DDd_2^* + 1 + \di\Divdo\right)\right)\dot{\mathfrak{B}}\right)\\
+&D\left(\frac{1}{6v^{3}}\Divdo(8-\Divdo\DDd_{2}^{*})\left(\Divdo \DDd_2^* + 1 + \di\Divdo\right)\Divdo\gdcd\right)\\
+&D\left(\frac{1}{5v^{4}}\Divdo\left(4\Divdo\DDd_{2}^{*}-2+4\di\Divdo+\frac{5}{6}\Divdo\DDd_{2}^{*}\left(\Divdo \DDd_2^* + 1 + \di\Divdo\right)\right)\di\mfc_{\dot{(\Om\tr\chi)}}\right)\\
=&\frac{1}{15}\Divdo\left(36+2(8-\Divdo\DDd_{2}^{*})\left(\Divdo \DDd_2^* + 1 + \di\Divdo\right)\right)\Divdo\left(\frac{1}{v^{5}}\chih\right) + h_{\Db^{2}\ombd}
\end{aligned}\label{eq:repDu2ombD}
\end{align}
where $\dot{\mathfrak{B}}$ denotes the quantity
\begin{equation*}
\dot{\mathfrak{B}}=v^2\etad+ \frac{v^3}{2}\di\left(\omtrchid-\frac{4}{v}\Omd\right)
\end{equation*}
and the source term $h_{\Db^{2}\ombd}$ is
\begin{align*}
&h_{\Db^{2}\ombd}\\
=&\mfc_{\Db^2\ombd}-\frac{1}{2v^{3}}\Divdo\Divdo\mfc_{\abd}+\frac{1}{3v^{3}}\Divdo(8-\Divdo\DDd_{2}^{*})\Divdo(D\QQ_\chibhd)\\
&-\frac{1}{4v^{4}}\left(2(\Ldo-3) +\frac{1}{4}\Ldo(\Ldo+2) \right)(D\QQ_{\dot{(\Om\tr\chib)}})-\frac{1}{3v^{3}}\Divdo\left(8-\Divdo\DDd_{2}^{*}\right)\di(\Ldo+2)(D\QQ_\phid)\\
&-\frac{1}{5v^{3}}\Divdo\left(8\Divdo\DDd_{2}^{*}-4+8\di\Divdo+\frac{5}{3}\Divdo\DDd_{2}^{*}\left(\Divdo \DDd_2^* + 1 + \di\Divdo\right)\right)\left(\mfc_{\etad}+\di\mfc_{\phid}+\frac{1}{2v^{2}}\di\mfc_{\dot{(\Om\tr\chi)}}\right)\\
&+\frac{1}{6v^{5}}\Divdo(8-\Divdo\DDd_{2}^{*})\left(\Divdo \DDd_2^* + 1 + \di\Divdo\right)\Divdo\mfc_{\chihd}
\end{align*}
Integrating \eqref{eq:repDu2ombD} yields the desired representation formula. The null structure equation \eqref{eq:linNSEDu2omb} implies that the following estimate holds.
\begin{equation*}
\Vert \Db^{2}\ombd \Vert_{\HH^{2}_{4}}\les  \Vert \chihd \Vert_{{H^8_3(\HH)}}+\Vert \Omd \Vert_{{H^8_4(\HH)}}+\Vert \dot{\mfx}_1 \Vert_{\XX^8(S_1)}+  \Vert (\mfc_i)_{\varphi} \Vert_{\mathcal{Z}_{\check{\CC}}}.
\end{equation*}\\
\end{lemma}

\subsection{Analysis of novel conservation laws}\label{sec:conservationlaws}

By analysing the kernels of the differential operators on the right-hand sides of the representation formulae \eqref{eq:repformulaforDbab} and \eqref{eq:repDu2ombD}, we can derive conservation laws along $\HH_{0,[1,2]}$. Moreover, quantities for which the $v'$-weight appearing in the $\chihd$ integrals of the right-hand sides of representation formula can be combined into conservation laws. Beginning with the right-hand side of \eqref{eq:repformulaforDbab}, performing a spherical harmonic decomposition on the operator $$2-\Divdo\DDd_2^*,$$ see \ref{sec:sphharmanalysis}, we obtain that the kernel consists of the set of vectorfields 
\begin{equation*}
\{X:X=X^{[2]}\}.
\end{equation*}
Projecting the representation formula \eqref{eq:repformulaforDbab} onto the $l=2$ spherical harmonic modes yields the following conservation law.
\begin{lemma}[Conservation law for $\Db\abd$]\label{lem:consvlawsQDuab}
The charge $\QQ_{\Db\abd^{[2]}}$ defined by
\begin{equation}
\QQ_{\Db\abd^{[2]}} := \frac{1}{v}\Db\abd^{[2]} +\frac{1}{v^{2}}(1-2\DDd_{2}^{*}\Divdo)\abd^{[2]},
\end{equation}
is conserved. It satisfies the equation
\begin{equation*}
D\QQ_{\Db\abd^{[2]}}  = \frac{1}{v}\mfc_{\Db\abd}^{[2]}+\frac{1}{v^{2}}(1-2\DDd_{2}^{*}\Divdo)\mfc_{\abd}^{[2]}.
\end{equation*}
\end{lemma}

\begin{remark}
The charge $\QQ_{\Db\abd^{[2]}}$ is precisely the linearisation the combination of linearisations of $\mathbf{U}$ and $\mathbf{V}$, i.e. 
\begin{equation*}
\QQ_{\Db\abd^{[2]}} = \sum_{-2\leq m \leq 2} \mathbf{U}^{m}\psi^{2m}+\mathbf{V}^{m}\phi^{2m}
\end{equation*}
\end{remark}

 Comparing the representation formulas for $\Db\ombd$ from \eqref{eq:repformDuomb} and for $\Db\abd$ from \eqref{eq:repformulaforDbab}, we see that the $v'$-weight appearing in both representation formulas is $1/v'^4$. Thus, these representation formulas can be combined into a conservation law. The operator $2-\Divd\DDd_2^*$ on the right-hand side has a kernel, see Appendix \ref{sec:sphharmanalysis} and thus for the $l=2$ modes there exists the conserved charge\footnote{This is precisely the definition of $\QQ_7$ in \cite{ACR2}}

\begin{align}
\begin{aligned}
\QQ_{\Db\ombd^{[2]}} :=& \Db\ombd^{[2]} +\frac{3}{2v^3} \QQ_{\omtrchibd}^{[2]}+ \frac{1}{2v^2} \Divdo \Divdo \QQ_\chibhd^{[2]} -\frac{12}{v^2} \QQ_\phid^{[2]} \\
&+ \frac{3}{2v^2} \Divdo \left(\eta+ \frac{v}{2}\di\left(\omtrchid-\frac{4}{v}\Omd\right)\right)^{[2]} - {\frac{3}{4v^2} \Divdo \Divdo \gdcd }^{[2]}. 
\end{aligned}
\end{align}

Combining the representation formula \eqref{eq:repformulaforDbab} and \eqref{eq:repformDuomb} for modes $l\geq 3$ yields the following lemma.
\begin{lemma}[Conservation law for $\Db\ombd$ and $\Db\abd$]\label{lem:consvlawQDuombDuab}
The charge $\QQ_{\Db\abd^{[\geq3]}_\psi}$ defined by
\begin{align}
\begin{aligned}
\QQ_{\Db\abd^{[\geq3]}_\psi}:=&\frac{1}{v}\Db\abd_{\psi}^{[\geq 3]}+\frac{1}{v^{2}}(1-2\DDd_{2}^{*}\Divdo)\abd_{\psi}^{[\geq 3]}\\ &+2\DDd_{2}^{*}\left(\frac{3}{v}\left(\frac{1}{v^{2}}\Divdo\chibhd-\frac{1}{2}\di\dot{(\Om\tr\chib)}-\frac{1}{v}\dot{\et}\right)+\DDd_{1}^{*}(\Db\ombd,0)   \right)_{\psi}^{[\geq 3]}
\end{aligned}
\end{align}
is conserved. It satisfies the equations 
\begin{align*}
D\QQ_{\Db\abd^{[\geq3]}_\psi}  =& \frac{1}{v}(\mfc_{\Db\abd})_{\psi}^{[\geq 3]}+\di(\mfc_{\Db\ombd})_\psi^{[\geq 3]}+\frac{1}{v^{2}}(1-2\DDd_{2}^{*}\Divdo)(\mfc_{\abd})_{\psi}^{[\geq 3]}+\frac{3}{v^{3}}\Divdo(\mfc_{\chibhd})_{\psi}^{[\geq 3]}\\
& -\frac{3}{2v}\di(\mfc_{\dot{(\Om\tr\chib)}})_{\psi}^{[\geq 3]}-\frac{3}{v^{2}}(\mfc_{\etad})_{\psi}^{[\geq 3]}
\end{align*}
\end{lemma}
\begin{proof}
The conservation law for $\QQ_{\Db\abd^{[\geq3]}_\psi} $ is obtained by applying the null structure equations. Differentiating $\QQ_{\Db\abd^{[\geq3]}_\psi} $, we obtain (suppressing the notation $[\geq 3]$) 
\begin{align*}
\begin{aligned}
D\QQ_{\Db\abd^{[\geq3]}_\psi}  &= D\left(\frac{1}{v}\Db\abd_{\psi}+\frac{1}{v^{2}}(1-2\DDd_{2}^{*}\Divdo)\abd_{\psi}\right)\\
&+D\left(2\DDd_{2}^{*}\left(\frac{3}{v}\left(\frac{1}{v^{2}}\Divdo\chibhd-\frac{1}{2}\di\dot{(\Om\tr\chib)}-\frac{1}{v}\dot{\et}\right)-\di\Db\ombd\right)_{\psi}\right).
\end{aligned}
\end{align*}
Applying the null structure equations, we have that
\begin{align*}
&D\left(\frac{1}{v}\Db\abd_{\psi}+\frac{1}{v^{2}}(1-2\DDd_{2}^{*}\Divdo)\abd_{\psi}\right)\\
=&\frac{1}{v}(\mfc_{\Db\abd})_{\psi}+\frac{1}{v^{2}}(1-2\DDd_{2}^{*}\Divdo)(\mfc_{\abd})_{\psi}+\frac{4}{v^{2}}\DDd_{2}^{*}(2-\Divdo\DDd_{2}^{*})\left(\frac{1}{v^{2}}\Divdo\chibhd-\frac{1}{2}\di\dot{(\Om\tr\chib)}-\frac{1}{v}\dot{\et}\right),
\end{align*}
while
\begin{align*}
&D\left(2\DDd_{2}^{*}\left(\frac{3}{v}\left(\frac{1}{v^{2}}\Divdo\chibhd-\frac{1}{2v^{2}}\di v^{2}(\dot{(\Om\tr\chib)})-\frac{1}{v^{3}}(v^{2}\dot{\et})\right)-\di\Db\ombd\right)_{\psi}\right)\\
=&2\DDd_{2}^{*}\left\{\left(-\frac{3}{v^{2}}\left(\frac{1}{v^{2}}\Divdo\chibhd-\frac{1}{2}\di\dot{(\Om\tr\chib)}-\frac{1}{v}\dot{\et}\right)\right)_{\psi}\right.\\
&+\frac{3}{v}\left(\frac{3}{v^{2}}\dot{\et}+\frac{1}{v}\di\dot{(\Om\tr\chib)} -\frac{1}{v^3}\Divdo\chibhd \right)_{\psi}+\frac{3}{v^{2}}\Divdo\left(\frac{2}{v} \DDd_2^* \left(\etad -2 \di\Omd\right)+\frac{1}{v}\mfc_6\right)_{\psi}\\
&\left.-\frac{3}{2v^{3}}\di\left( - 2 \Divdo \left(\etad-2\di\Omd\right)+ v^2 \mfc_5\right)_{\psi}-\frac{3}{v^{4}}\left(\frac{v^2}{2} \di \left(\frac{4}{v}\Omd\right) +v\di\Omd + v^2 \mfc_4\right)_{\psi}\right\}\\
&-\di\left(\frac{3}{2v^{2}}\left(\omtrchibd\right)
+ \frac{1}{v^2} \Divdo \left(\frac{1}{v^2} \Divdo \chibhd - \frac{1}{v}\etad - \frac{1}{2} \di \omtrchibd\right)+ \mfc_9\right)_\psi.
\end{align*}
Collecting terms in this formula and using that 
\begin{equation*}
(\Divdo\DDd_{2}^{*}+1+\frac{1}{2}\di\Divdo)_{\psi}V = 0
\end{equation*}
for vectorfields $V$, yields the stated formula.
\end{proof}
Finally, the operator on the right-hand side of \eqref{eq:repDu2ombD} does not have a kernel for modes $l\geq2$, the related charge is obtained by projecting the representation formula onto the $l\leq 1$ spherical harmonic modes.
\begin{lemma}[Conservation law for $\Db^2\ombd$] \label{lem:consvlawsDu2ombd} Let $\QQ_{\Db^{2}\ombd^{[\leq1]}} $ denote the charge
\begin{align}
\QQ_{\Db^{2}\ombd^{[\leq1]}} =\Db^{2}\ombd^{[\leq 1]}-\frac{1}{2v^{4}}(\Ldo-3) \QQ_{2}^{[\leq 1]}+\frac{4}{v^{5}}\Divdo \QQ_{0}
\end{align}
Then, the charge satisfies the equation
\begin{align*}
D(\QQ_{\Db^{2}\ombd^{[\leq1]}}  +\frac{4}{v^{4}}\mfc_{\dot{(\Om\tr\chi)}})=h_{\Db^{2}\ombd}^{[\leq 1]}
\end{align*}
\end{lemma}
The following lemma shows that the  novel charges are bounded.
\begin{lemma}[Charge estimate]\label{lem:chargebdQ} Let $\dot{\mfx}_{v}$ denote linearised $C^3$-sphere data on a sphere $S_{v}$ with $1\leq v \leq 2$. Then
\begin{align*}
 \Vert \QQ_{\Db\abd^{[2]}}  \Vert_{H^2(S_v)}+ \Vert \QQ_{\Db\abd^{[\geq3]}_\psi}  \Vert_{H^2(S_v)}+\Vert \QQ_{\Db^{2}\ombd^{[\leq1]}}  \Vert_{H^2(S_v)} \les \Vert \dot{\mfx}_{v}\Vert_{\XX^8(S_{v})}.\\
\end{align*}
\end{lemma}

\subsection{Gauge dependence of conserved charges}\label{sec:Qgaugedependence}

In this section, we investigate the gauge-dependence of the novel charges of Lemmas \ref{lem:consvlawsQDuab}, \ref{lem:consvlawQDuombDuab} and \ref{lem:consvlawsDu2ombd}. That is, we prove the following proposition.

\begin{proposition}\label{prop:gaugedepnewQ}
The charge $\QQ_{\Db\abd^{[2]}}$ is gauge-invariant, i.e. for all linearised perturbation functions $\dot{f}$ and $\dot{q}$,
\begin{equation*}
\QQ_{\Db\abd^{[2]}}(\dot{\PP}_{f,q}(\dot{f},\dot{q}))=0.
\end{equation*}
Now, let  $(\QQ_{\Db^{2}\ombd^{[\leq1]}} )_{0}$ be a scalar field and $\left(\QQ_{\Db\abd^{[\geq3]}_\psi}\right)_{0}$ a symmetric tracefree $2$-tensor satisfying
\begin{align}\label{eq:modematching}
(\QQ_{\Db^{2}\ombd^{[\leq1]}} )_{0} = (\QQ_{\Db^{2}\ombd^{[\leq1]}} )_{0}^{[2]}, \qquad \left(\QQ_{\Db\abd^{[\geq3]}_\psi}\right)_{0}=\left(\left(\QQ_{\Db\abd^{[\geq3]}_\psi}\right)_{0}\right)_\psi^{[\geq3]}.
\end{align}
Then there exists linearised perturbation functions $\dot f$ and $\dot{q}$ at $S_2$ such that
\begin{align*} 
\QQ_{\Db^{2}\ombd^{[\leq1]}}\left(\dot{\PP}_{f,q}(\dot{f},\dot{q})\right)=(\QQ_{\Db^{2}\ombd^{[\leq1]}})_0 \qquad \left(\QQ_{\Db\abd^{[\geq3]}_\psi}\right)\left(\dot{\PP}_{f,q}(\dot{f},\dot{q}\right)=\left(\QQ_{\Db\abd^{[\geq3]}_\psi}\right)_0.
 \end{align*}
 Furthermore, under the assumption that  
 $$ \Vert (\QQ_{\Db\abd^{[\geq3]}_\psi} )_0 \Vert_{H^2(S_2)}
+\Vert (Q_{\Db^2\ombd})_{0} \Vert_{H^{2}(S_{2})}<\infty,$$
the following bound holds for $\pr_u^3\dot{f}$ and $\pr_u^4\dot{f}$
\begin{align}
\begin{aligned}
\Vert \pr_{u}^{3}\dot{f}^{[\geq3]} \Vert_{H^{4}(S_{2})}+\Vert \pr_{u}^{4}\dot{f}\Vert_{H^{2}(S_{2})}  \les& \Vert (\QQ_{\Db\abd^{[\geq3]}_\psi} )_0 \Vert_{H^2(S_2)}+\Vert (Q_{\Db^2\ombd})_{0} \Vert_{H^{2}(S_{2})}.
\end{aligned}\label{eq:higherfestimate}
\end{align}
\end{proposition}
\begin{proof}
Direct substitution of the linearised perturbation map $\dot{\PP}_{f,q}(\dot{f},\dot{q})$ from Lemma \ref{lem:pertlin} into the charges $\QQ_{\Db\abd^{[2]}}$, $\QQ_{\Db\abd^{[\geq3]}_\psi}$ and $\QQ_{\Db^{2}\ombd^{[\leq1]}}$ yields
\begin{align}
\begin{aligned}
\QQ_{\Db\abd^{[2]}}=&0,\\
\QQ_{\Db\abd^{[\geq3]}_\psi}=&\DDd_{2}^{*}\DDd_{1}^{*}(\pr_{u}^{3}\dot{f},0)^{[\geq 3]}_\psi,\\
\QQ_{\Db^{2}\ombd^{[\leq1]}}=&\frac{1}{2}\pr_{u}^{4}\dot{f}^{[\leq 1]}.
\end{aligned}\label{eq:newQperturbations}
\end{align}
We see immediately that $\QQ_{\Db\abd^{[2]}}$ is gauge-invariant. Furthermore, by definition $\QQ_{\Db\abd^{[2]}}$ is an $l=2$ spherical harmonic mode consisting of an electric and magnetic part and therefore it is $10$-dimensional. Thus, this gives the $10$-dimensional space of novel gauge-invariant charges of the $C^3$ linearised gluing problem.
\\

It remains to show that for $\QQ_{\Db\abd^{[\geq3]}_\psi}$ and $\QQ_{\Db^{2}\ombd^{[\leq1]}}$ a suitable choice of $\dot{f}$ can be made. We proceed by constructing directly the functions $\pr_u^3\dot{f}^{[\geq3]}$ and $\pr_u^4\dot{f}$. For $\pr_{u}^{3}\dot{f}^{[\geq3]}$, we have the equations from \eqref{eq:newQperturbations}
\begin{align*}
 \DDd_{2}^{*}\DDd_{1}^{*}\left(\pr_u^3 \dot{f},0\right)^{[\geq 3]}_\psi=&\left(\QQ_{\Db\abd^{[\geq3]}_\psi}\right)_{0}.
\end{align*}
By \eqref{eq:modematching}, $\pr_{u}^{3}\dot{f}$ is well defined and performing a spherical harmonic decomposition yields
\begin{align*}
\left(\pr_u^3 \dot{f}\right)^{(lm)}=-\frac{1}{\sqrt{\frac{1}{2} l(l+1)-1}\sqrt{l(l+1)}}\left(\left(\QQ_{\Db\abd^{[\geq3]}_\psi}\right)_{0}\right)^{(lm)}.
\end{align*}
Putting this together we obtain the bound
\begin{align*}
\Vert \pr_{u}^{3}\dot{f} ^{[\geq3]}\Vert_{H^{4}(S_{2})}\les \Vert (\QQ_{\Db\abd^{[\geq3]}_\psi} )_0 \Vert_{H^2(S_2)}
\end{align*}
For $\pr_{u}^{4}\dot{f}$, we have the formula
\begin{align*}
\begin{aligned}
\pr_{u}^{4}\dot{f}^{[\leq 1]}=&2(Q_{\Db^2\ombd})_{0}.
\end{aligned}
\end{align*}
By letting $\pr_{u}^{4}\dot{f}^{[\geq 2]}=0$, $\pr_{u}^{4}\dot{f}$ is well defined and bounded by
\begin{equation*}
\Vert \pr_{u}^{4}\dot{f}\Vert_{H^{2}(S_{2})}\les\Vert (Q_{\Db^2\ombd})_{0} \Vert_{H^{2}(S_{2})}.
\end{equation*}
Thus, putting the two estimates together, we obtain the estimate \eqref{eq:higherfestimate}.
\end{proof}

\begin{remark}
Proposition \ref{prop:gaugedepnewQ} shows that the linearised perturbation function $\dot{f}$ is well defined for the novel $C^3$ charges. However, we also need to show that $\dot{f}$ is well-defined for the charges arising from the $C^2$-null gluing problem. The charges appearing in the $C^2$-null gluing can be found in Appendix \ref{sec:C2QQ}. In Section 4 of \cite{ACR2} it is found, for example, that
\begin{align}
\begin{aligned}
\QQ_{\Db\ombd^{[\leq1]}}=& \frac{1}{2} \pr_u^3 \dot{f}^{[\leq 1]} , \\
 \QQ_{\Db\ombd^{[2]}}=&\frac{1}{2} \pr_u^3 \dot{f}^{[2]}-3 \dot{f}^{[2]}.\\
\end{aligned}\label{eq:QDuombgauge}
\end{align}
Proposition 4.4 in \cite{ACR2} shows that $\dot{f}$ is well defined for the charges appearing for the $C^2$-sphere data. Since the $l\geq3$ modes of $\pr_u^3\dot{f}$ are not considered, this result can be combined with Proposition \ref{prop:gaugedepnewQ} and Lemma \ref{lem:estimatespfq} to obtain the estimate
 \begin{align}
\begin{aligned}
&\Vert \dot{f} \Vert_{\YY_{f^{(4)}}} + \Vert\dot{q}\Vert_{\YY_{q}}+ \Vert \dot{\PP}_{f,0}(\dot f) \Vert_{\mathcal{X}(S_2)} +\Vert \dot{\PP}_{0,q}(\dot{q}) \Vert_{\mathcal{X}(S_2)}\\
\les& \Vert (\QQ_\phid)_0 \Vert_{H^8(S_2)}+\Vert (\QQ_{\dot{(\Om\tr\chib)}})_0 \Vert_{H^6(S_2)}+\Vert (\QQ_\chibhd)_0 \Vert_{H^6(S_2)}+\Vert (\QQ_{\ab_\psi})_0\Vert_{H^4(S_2)}\\
&+\Vert (\QQ_{\ombd^{[\leq1]}})_0 \Vert_{H^6(S_2)}+\Vert (\QQ_{\Db\ombd^{[\leq1]}})_0 \Vert_{H^4(S_2)}+\Vert (\QQ_{\Db\ombd^{[2]}})_0\Vert_{H^4(S_2)}\\
&+\Vert (\QQ_{\Db\abd^{(\geq3)}_\psi} )_0 \Vert_{H^2(S_2)}+\Vert (\QQ_{\Db^{2}\ombd^{[\leq1]}} )_0 \Vert_{H^2(S_2)}.
\end{aligned}\label{eq:estimateforpertfunc}
\end{align}
where the charges on the right-hand side are defined analogously to \eqref{eq:modematching}.
\end{remark}

\subsection{Proof of Proposition \ref{prop:lingluing}}\label{sec:lingluingproof}

In this section, we use the results of sections \ref{sec:C3repformulas}, \ref{sec:conservationlaws} and \ref{sec:Qgaugedependence} to prove Proposition \ref{prop:lingluing} and hence prove Proposition \ref{prop:surjection}. Let 
\begin{align*}
\dot{\mfx}_{1}\in\XX(S_{1}),\qquad {\tilde{\mathfrak{M}}}_{2}\in\ZZ_\mathfrak{M}(\tilde{S}_{2}), \qquad \mfc_{\dot{\varphi}}\in\ZZ_{\CC},
\end{align*}
be given linearised $C^3$-sphere data on $S_{1}$, linearised $C^3$-matching data on $\tilde{S}_{2}\subset\HHb_{[-\delta,\delta],2}$ and linearised source functions on $\HH_{0,[1,2]}$, respectively. 
We prove Proposition \ref{prop:lingluing} by showing the following.
\begin{itemize}
\item In Section \ref{sec:matchingdata}, we show that the matching map $\mathfrak{M}_{u,v}$ contains the complement of $({\mathbf{E}},{\mathbf{P}},{\mathbf{L}},{\mathbf{G}},{\mathbf{U}},{\mathbf{V}})$ in the sphere data $\mathfrak{X}_{u,v}$.
\item In Section \ref{sec:Qgluing}, we apply Proposition \ref{prop:gaugedepnewQ} to match the gauge-dependent charges of Lemmas \ref{lem:consvlawQDuombDuab} and \ref{lem:consvlawsDu2ombd} at $S_2$.
\item In Section \ref{sec:lingluingsoln}, we construct conditions on the linearised characteristic seed $(\Omd,\chihd)$ such that the novel part of the $C^3$-null data $\mfx_{u,v}$ given by the representation formulae of Section \ref{sec:C3repformulas} satisfies \eqref{eq:spheredatamatching}.

\item Finally, in Section \ref{sec:prooflingluingest}, we prove the estimate \eqref{eq:estimatelingluing} and conclude the proof of Proposition \ref{prop:lingluing}.
\end{itemize}

\subsubsection{$\mathfrak{M}$ complements $({\mathbf{E}},{\mathbf{P}},{\mathbf{L}},{\mathbf{G}},{\mathbf{U}},{\mathbf{V}})$}\label{sec:matchingdata}

This section is concerned with the proof of the following lemma.

\begin{lemma}[Matching all gluable quantities] \label{lem:matchinggluable} Let $\mfx_{u,v}$ and $\mfx'_{u,v}$ be $C^3$-sphere data on $S_{u,v}$ such that for a real number $\varep>0$,
\begin{align*} 
\Vert \mfx_{u,v} -\mathfrak{m} \Vert_{\XX(S_{u,v})} + \Vert \mfx'_{u,v} -\mathfrak{m} \Vert_{\XX(S_{u,v})} < \varep,
\end{align*}
and
\begin{align} 
\MMf(\mfx_{u,v})=\MMf(\mfx'_{u,v}).
\label{eq:lemmatching1}
\end{align}
For $\varep>0$ sufficiently small, if the charges coincide, that is,
\begin{align} \label{eq:chargematchinglemma}
(\mathbf{E},\mathbf{P}, \mathbf{L}, \mathbf{G},\mathbf{U},\mathbf{V})(\mfx_{u,v}) = (\mathbf{E},\mathbf{P}, \mathbf{L}, \mathbf{G},\mathbf{U},\mathbf{V})(\mfx'_{u,v}),
\end{align}
then the following holds,
\begin{align*} 
\mfx_{u,v}= \mfx'_{u,v}.
\end{align*}
\end{lemma}
\begin{remark} Before stating the proof, we make the following remarks.
\begin{enumerate}
\item The quantities $\tilde{\QQ}_{\omb}$, $\tilde{\QQ}_{\Db\omb}$ and $\tilde{\QQ}_{\Db^2\omb}$ appearing in the matching map, see Definition \ref{def:matchingmap}, linearise to the conserved charges $\QQ_{\ombd}$, $\QQ_{\Db\ombd}$ and $\QQ_{\Db^{2}\ombd^{[\leq1]}}$.
\item The proof for the $C^2$ part of the matching data $m_{u,v}$ is contained in Lemma 2.12 of \cite{ACR2} and therefore, we only need to prove Lemma \ref{lem:matchinggluable} for the novel $C^3$ part, $\check{m}$, consisting of the quantities 
$$(D^{2}\om,\Db^{2}\omb^{[\geq 2]}, \tilde{\QQ}_{\Db^{2}\ombd^{[1]}},\Dh\a,\Dbh\ab^{[\geq 3]}).$$
\end{enumerate}
\end{remark}
\begin{proof}[Proof of Lemma \ref{lem:matchinggluable}]
The matching of $\mathbf{U}$ and $\mathbf{V}$ in \eqref{eq:chargematchinglemma} can be combined to write
\begin{equation*}
\frac{1}{\phi}\Dbh\ab^{[2]}(\mfx_{u,v})+\frac{1}{\phi^{2}}(1+\Nabs\widehat{\otimes}\Divd)\ab^{[2]}(\mfx_{u,v})=\frac{1}{\phi}\Dbh\ab^{[2]}(\mfx'_{u,v})+\frac{1}{\phi^{2}}(1+\Nabs\widehat{\otimes}\Divd)\ab^{[2]}(\mfx'_{u,v}).
\end{equation*}
By the matching of $\ab$ and $\phi$ in \eqref{eq:lemmatching1}, we directly obtain that
\begin{equation*}
\Dbh\ab^{[2]}(\mfx_{u,v})=\Dbh\ab^{[2]}(\mfx'_{u,v}).
\end{equation*}
Consider now the quantity $\Db^{2}\omb$. On the one hand $\Db^{2}\omb^{[\geq 2]}$ is matched by \eqref{eq:lemmatching1}. On the other hand, the quantity 
\begin{align*}
\begin{aligned}
\tilde{\QQ}_{\Db^{2}\ombd}  :=& \left(\Db^{2}\omb\right)^{[\leq1]} -\frac{1}{2} (\Ldo-3) \left(\frac{1}{v^{2}}\Om\trchib- \frac{2}{v^4}(\Ldo+2)\phi\right)^{[\leq1]} \\
&+ \frac{1}{2v^{2}} \left(\Ldo\Ldo + 2\Ldo -3\right) \left(\Om\trchi-\frac{4}{v}\Om\right)^{[\leq1]} -\frac{1}{v^3} \Divdo \eta^{[1]}
\end{aligned}
\end{align*}
is matched by \eqref{eq:lemmatching1}. In turn, this implies that for modes $l\leq1$,
$$ \Db^{2}\omb^{[\leq 1]}(\mfx_{u,v})=  \Db^{2}\omb^{[\leq 1]}(\mfx'_{u,v}). $$
\end{proof}

\subsubsection{Gluing of gauge-dependent charges}\label{sec:Qgluing}

We proceed by showing that linearised perturbations of sphere data can be added to the charges 
$$(\QQ_{\Db\abd^{[\geq3]}_\psi}, \QQ_{\Db^{2}\ombd^{[\leq1]}})$$
in order to match the charges at $S_{2}$ with the charges coming from $S_{1}$. Firstly, let $^{(1)}\QQ_{\Db\abd^{[\geq3]}_\psi}$ and $^{(1)} \QQ_{\Db^{2}\ombd^{[\leq1]}}$ be solutions to transport equations from Lemmas \ref{lem:consvlawQDuombDuab} and \ref{lem:consvlawsDu2ombd} with initial data on $S_{1}$ calculated from the given sphere data $\mfx_{1}$. Moreover, the following estimate follows from Lemma \ref{lem:chargebdQ}
\begin{equation}\label{eq:QS1estimate}
 \Vert ^{(1)}\QQ_{\Db\abd^{[\geq3]}_\psi}  \Vert_{H^2(S_v)}+\Vert ^{(1)}\QQ_{\Db^{2}\ombd^{[\leq1]}}  \Vert_{H^2(S_v)} \les \Vert \dot{\mfx}_{1}\Vert_{\XX(S_{1})}+\Vert \mfc_{\dot{\varphi}}\Vert_{\ZZ_\CC}
\end{equation}
Secondly, from the given linearised matching data ${\tilde{\MMf}}_{2}\in\ZZ_{\MMf}(\tilde{S}_{2})$ define $^{(2)}\QQ_{\Db\abd^{[\geq3]}_\psi}$ and $^{(2)} \QQ_{\Db^{2}\ombd^{[\leq1]}}$ to be the charges calculated from ${\tilde{\MMf}}_{2}$. Analogously, from Lemma \ref{lem:chargebdQ}, we have the estimate 
\begin{equation}\label{eq:QS2estimate}
 \Vert ^{(2)}\QQ_{\Db\abd^{[\geq3]}_\psi}  \Vert_{H^2(\tilde{S}_{2})}+\Vert ^{(2)}\QQ_{\Db^{2}\ombd^{[\leq1]}}  \Vert_{H^2(\tilde{S}_{2})} \les \Vert \tilde{\MMf}_{2}\Vert_{\XX(\tilde{S}_{2})}.
 \end{equation}
Now defining 
$$(\QQ_{\Db^{2}\ombd^{[\leq1]}} )_{0} :=^{(1)}\QQ_{\Db\abd^{[\geq3]}_\psi} -^{(2)}\QQ_{\Db\abd^{[\geq3]}_\psi}, \qquad \left(\QQ_{\Db\abd^{[\geq3]}_\psi}\right)_{0}:=^{(1)}\QQ_{\Db^{2}\ombd^{[\leq1]}} -^{(2)} \QQ_{\Db^{2}\ombd^{[\leq1]}} $$ 
we have that for linearised perturbation functions $\dot{f}$ and $\dot{q}$
\begin{align}
\begin{aligned}
\QQ_{\Db\abd^{[\geq3]}_\psi} (\dot{\tilde{\mfx}}_{2}+\dot{\PP}_{f,q}(\dot{f},\dot{q}))=&^{(2)}\QQ_{\Db\abd^{[\geq3]}_\psi} +\QQ_{\Db\abd^{[\geq3]}_\psi} (\dot{\PP}_{f,q}(\dot{f},\dot{q}))\\
=&^{(2)}\QQ_{\Db\abd^{[\geq3]}_\psi} +\left(\QQ_{\Db\abd^{[\geq3]}_\psi}\right)_0\\
=&^{(1)}\QQ_{\Db\abd^{[\geq3]}_\psi}.
\end{aligned}
\end{align}
where $\dot{\tilde{\mfx}}_{2}$ is the sphere data that gives rise to ${\tilde{\MMf}}_{2}$. A similar calculation shows that 
\begin{equation}
\QQ_{\Db^{2}\ombd^{[\leq1]}}(\dot{\tilde{\mfx}}_{2}+\dot{\PP}_{f,q}(\dot{f},\dot{q}))=^{(1)}\QQ_{\Db^{2}\ombd^{[\leq1]}}.
\end{equation}
Moreover, combining the estimates \eqref{eq:QS1estimate} and \eqref{eq:QS2estimate} yields the estimate 
\begin{equation}\label{eq:novelQestimate}
 \Vert (\QQ_{\Db\abd^{[\geq3]}_\psi})_0  \Vert_{H^2(\tilde{S}_{0,2})}+\Vert (\QQ_{\Db^{2}\ombd^{[\leq1]}})_0  \Vert_{H^2(\tilde{S}_{0,2})} \les \Vert \dot{\mfx}_{1}\Vert_{\XX^(S_{1})}+ \Vert \tilde{\MMf}_{2}\Vert_{\XX(\tilde{S}_{2})}+\Vert \mfc_{\dot{\varphi}}\Vert_{\ZZ_\CC}.
\end{equation}

\subsubsection{Constructing $\Omd$ and $\chihd$}\label{sec:lingluingsoln}

In this section, we utilise the matching of perturbation dependent charges from the previous section, the representation formulas from Section \ref{sec:C3repformulas} and the freedom to prescribe $\Omd$ and $\chihd$ along $\HH_{0,[1,2]}$ to construct a solution $\dot{\mfx}$ to the linearised constraint equations, $C_{\dot{\varphi}}(\dot{x}) =\mfc_{\dot{\varphi}}$, satisfying
\begin{align} 
\begin{aligned} 
\dot{\mfx}\vert_{S_1} =& \mfx_{1}, \\
{\MMf}\left(\dot{\mfx}\vert_{S_2}- \dot{\PP}_{f,q}(\dot f, \dot q)\right) =& {\tilde{\MMf}}_{2} \text{ on } S_2,
\end{aligned} \label{eq:bdrymatchingdata}
\end{align}
We prove that we can choose the characteristic seed ($\Omd,\chihd)$ to satisfy the integral formulas of Section \ref{sec:C3repformulas}. We do this first by constructing such a characteristic seed for $\check{x}$ in Proposition \ref{prop:C3lin2} and apply the results for the $C^2$-sphere data from Proposition \ref{prop:C2lingluing}.\\

\ni\textit{(1) Gluing of $D\ad$ and $D^2\omd$.}\\

\ni Recall the linearised null structure equation for $\ad$
\begin{align*} 
\ad + D\chihd=\mfc_{\ad}.
\end{align*}
Differentiating this equation, we obtain the linearised null structure equation for $D\ad$ 
\begin{align*} 
D\ad + D^{2}\chihd=D\mfc_{\ad}.
\end{align*}
Therefore, we glue $D\ad$ at $S_{2}$ by prescribing
\begin{align} 
D^{2}\chihd(1)= D\mfc_{10}(1) - D\ad(1), \qquad D^{2}\chihd(2) = D\mfc_{10}(2) - D\ad(2).
\label{eq:prescribechihdD2}
\end{align}

\ni For $D^2\omd$, recall that $$\omd = D\Omd.$$ Then to glue $D^{2}\omb$ at $S_{2}$ we write
\begin{align} 
\begin{aligned} 
D^{3}\Omd(1)=& D^{2}\omd(1), & D^{3}\Omd(2)=& D^{2}\omd(2).
\end{aligned} \label{eq:prescribeOmD1}
\end{align}
Thus, we have the estimate 
\begin{align} 
\begin{aligned} 
&\Vert D^{2}\chihd(1) \Vert_{H^8(S_1)}+\Vert D^{2}\chihd(2) \Vert_{H^8(S_2)}+\Vert D^3\Omd(1) \Vert_{H^8(S_1)}+\Vert D^3\Omd(2) \Vert_{H^8(S_2)} \\
 \les& \Vert \mfc_{\ad} \Vert_{H^8(S_1)}+\Vert \mfc_{\ad} \Vert_{H^8(S_2)} + \Vert \xdmf_{1}\Vert_{\XX(S_1)} + \Vert \MMf_{2} \Vert_{\ZZ_\MMf(S_2)}.\\
\end{aligned} \label{eq:estimateaDaD2om}
\end{align}

\ni \textit{(2) Gluing of $\Db\abd$ and $\Db^2\ombd$} \\

\ni Recall the representation formula for $\Db\abd$, \eqref{eq:repformulaforDbab},
\begin{align}
\begin{aligned}
&\left[ \frac{1}{v}\Db\abd +\frac{1}{v^{2}}(1-2\DDd_{2}^{*}\Divdo)\abd\right]_{1}^{v} \\
&+\left[ \frac{1}{v^{2}}\DD_{2}^{*}(2-\Divdo\DDd_{2}^{*})\left\{2\Divdo\QQ_{3}-\frac{2}{3v}\di\QQ_{2} -2\di(\Ldo+2)\QQ_{1}\right. \right.\\
&\phantom{+D\Biggl(} -\left(\Divdo \DDd_2^* + 1 + \di\Divdo\right) \left(\etad+ \frac{v}{2}\di\left(\omtrchid-\frac{4}{v}\Omd\right)\right)\\
&\phantom{+D\Biggl(}\left.\left.+ \left(\Divdo \DDd_2^* + 1 + \di\Divdo\right)\Divdo\gdcd - \frac{1}{2v}\left(\Divdo \DDd_2^* + 1 + \di\Divdo\right)\di\mfc_{\dot{(\Om\tr\chib}}\right\} \right]^v_{1}\\
&=\DDd_{2}^{*}(2-\Divdo\DDd_{2}^{*}) \left(\Divdo \DDd_2^* + 1 + \di\Divdo\right)\Divdo\left(\int_{1}^{v}\frac{1}{v'^{4}}\chihd dv'\right)+ \int_{1}^{v} h_{\Db\abd} dv'.
\end{aligned}\label{eq:prescribechih4}
\end{align}
For modes $l\geq3$, the operator $$\DDd_{2}^{*}(2-\Divdo\DDd_{2}^{*}) \left(\Divdo \DDd_2^* + 1 + \di\Divdo\right)\Divdo$$ has trivial kernel and is elliptic. More specifically, prescribing that the $\chihd$-integral vanishes on the $l=2$ mode, we get the estimate
\begin{align} 
\begin{aligned} 
\left\Vert \int_1^2 \frac{1}{v'^4}\chihd dv' \right\Vert_{H^8(S_1)} \les& \Vert \xdmf_1 \Vert_{\mathcal{X}(S_1)} + \Vert \MMf_2 \Vert_{\ZZ_\MMf(S_2)} + \Vert (\mfc_{\dot{\varphi}}) \Vert_{\mathcal{Z}_\CC}. 
\end{aligned} \label{eq:chihestimate4}
\end{align}

Similarly, integrating the representation formula for$\Db^{2}\ombd$, \eqref{eq:repDu2ombD} yields
\begin{align}
\begin{aligned}
&\left[\Db^{2}\ombd - \frac{1}{2v^{3}}\Divdo\Divdo\ab + \frac{1}{3v^{3}}\Divdo(8-\Divdo\DDd_{2}^{*})\Divdo \QQ_{3} \right]_{1}^{v}\\
-&\left[\frac{1}{4v^{4}}\left(2(\Ldo-3) +\frac{1}{4}\Ldo(\Ldo+2) \right)\QQ_{2}+\frac{1}{3v^{3}}\Divdo\left(8-\Divdo\DDd_{2}^{*}\right)\di(\Ldo+2)\QQ_{1}\right]_{1}^{v}\\
-&\left[\frac{1}{5v^{5}}\Divdo\left(8\Divdo\DDd_{2}^{*}-4+8\di\Divdo+\frac{5}{3}\Divdo\DDd_{2}^{*}\left(\Divdo \DDd_2^* + 1 + \di\Divdo\right)\right)\dot{\mathfrak{B}}\right]_{1}^{v}\\
+&\left[\frac{1}{6v^{3}}\Divdo(8-\Divdo\DDd_{2}^{*})\left(\Divdo \DDd_2^* + 1 + \di\Divdo\right)\Divdo\gdcd\right]_{1}^{v}\\
+&\left[\frac{1}{5v^{4}}\Divdo\left(4\Divdo\DDd_{2}^{*}-2+4\di\Divdo+\frac{5}{6}\Divdo\DDd_{2}^{*}\left(\Divdo \DDd_2^* + 1 + \di\Divdo\right)\right)\di\mfc_{\dot{(\Om\tr\chib)}}\right]_{1}^{v}\\
=&\Divdo\left(36+2(8-\Divdo\DDd_{2}^{*})\left(\Divdo \DDd_2^* + 1 + \di\Divdo\right)\right)\Divdo\left(\int_{1}^{v}\frac{1}{15v'^{5}}\chih dv'\right)\\
& + \int_{1}^{v} h_{\Db^{2}\ombd} dv'
\end{aligned}\label{eq:prescribechih5}
\end{align}
We have that $\Db^{2}\ombd$ is glued by prescribing the integral over $\frac{1}{v^5} \chihd$. On modes $l\geq2$, the operator 
$$\Divdo\left(36+2(8-\Divdo\DDd_{2}^{*})\left(\Divdo \DDd_2^* + 1 + \di\Divdo\right)\right)\Divdo$$  has trivial kernel and is elliptic, therefore we have the following estimate
\begin{align} 
\begin{aligned} 
\left\Vert \int_1^2 \frac{1}{v'^5}\chihd dv' \right\Vert_{H^8(S_1)} \les& \Vert \xdmf_1 \Vert_{\mathcal{X}(S_1)} + \Vert \MMf_2 \Vert_{\ZZ_\MMf(S_2)} + \Vert (\mfc_{\dot{\varphi}})\Vert_{\mathcal{Z}_\CC}. 
\end{aligned} \label{eq:chihestimate5}
\end{align} 

\ni\newline\ni This finishes the construction of novel components of the $C^3$-null data such that at $S_{1}$ and $S_{2}$ the matching \eqref{eq:spheredatamatchingcx} holds. Along with the gluing \eqref{eq:C2spheredatamatching} for the $C^2$-null data, we can glue up to the $20$-dimensional space of charges all components of the $C^3$-sphere data.  Next, we need to ensure that we can prescribe $\Omd$ and $\chihd$ such that the integral conditions \eqref{eq:prescribeOmD1}, \eqref{eq:prescribechih4} and \eqref{eq:prescribechih5} are satisfied.

\begin{remark}
In fact, due to the existence of the conservation law $\QQ_{\Db\abd^{[\geq3]}_\psi}$ innate to the $C^3$-null gluing problem, the gluing of $\Db\omb^{[\geq3]}$ in the $C^3$ setting is done through the matching of $\QQ_{\Db\abd^{[\geq3]}_\psi}$ rather than through the representation formula, \eqref{eq:repformDuomb} as it is in the $C^2$-null gluing problem.
\end{remark}

\ni The proof of the next lemma follows by an orthogonality argument and combines the integrals appearing in the representation formulas of the $C^2$-part of the linearised sphere data $\xd$.

\begin{lemma}[Existence of $\Omd$ and $\chihd$.] \label{lem:existenceofOmdchihd} Consider a scalar function $h$ and symmetric tracefree $2$-tensors $I_k$ for $0\leq k\leq 5$ on $\HH_{0,[1,2]}$ defined by 
\begin{align*}
h=\int_1^2\Omd dv', \qquad I_k=\int_1^2\frac{1}{v'^k}\chihd dv', \qquad 0\leq k \leq 5.
\end{align*}
The integrals $h$ and $I_k$ are orthogonal for all $0\leq k\leq 5$. Moreover, the integrals are independent of the prescription of the quantities 
\begin{align*}
&D^i\Omd(1), \qquad D^i\Omd(2),  \qquad 0\leq i\leq3\\
&D^j\chihd(1), \qquad D^j\chihd(2),  \qquad 0\leq j \leq2.
\end{align*}
\end{lemma} 

\ni This finishes the proof of Proposition \ref{prop:C3lin2}. Letting the $C^3$-null data be the combination $$\dot{\mfx}=\xd \times \check{x}$$ and applying the gluing \eqref{eq:C2spheredatamatching} proves the gluing \eqref{eq:spheredatamatching}.

\subsubsection{Proof of estimate \eqref{eq:estimatelingluing}}\label{sec:prooflingluingest}

Finally, in this section we prove the estimate \eqref{eq:estimatelingluing},
\begin{align*}
&\Vert\dot\mfx\Vert_{\XX(\HH_{[1,2]})} + \Vert \dot{f}\Vert_{\YY_{f^{(4)}}}+\Vert\dot{q}\Vert_{\YY_q} + \Vert \PP_{f,q}(\dot{f},\dot{q})\Vert_{\XX(S_{2})},\\
 \les& \Vert\dot\mfx_{1}\Vert_{\XX(S_{1})} + \Vert{\tilde{\mathfrak{M}}}_{2}\Vert_{\ZZ_{\mathfrak{M}}(\tilde{S}_{2})} + \Vert \mfc_{\dot\varphi}\Vert_{\ZZ_{\CC}}.
\end{align*}

\ni Combining the estimate \eqref{eq:estimateforpertfunc} and \eqref{eq:novelQestimate} with the estimate \eqref{eq:C2estimatelingluing}, we obtain
 \begin{align}
\begin{aligned}\label{eq:finalest1}
&\Vert \dot{f} \Vert_{\YY_{f^{(4)}}} + \Vert\dot{q}\Vert_{\YY_{q}}+ \Vert \dot{\PP}_{f,q}(\dot f,\dot q) \Vert_{\mathcal{X}(S_2)}\\
\les &\Vert \dot{\mfx}_{1}\Vert_{\XX(S_{1})}+ \Vert {{\MMf}}_{2}\Vert_{\XX({S}_{2})}+\Vert \mfc_{\dot{\varphi}}\Vert_{\ZZ_\CC}.
\end{aligned}
\end{align}

\ni From Lemmas \ref{lem:repformDuab} and \ref{lem:repformDu2omb} along with the estimate  \eqref{eq:C2estimatelingluing} we have the estimate
\begin{align}\label{eq:finalest2}
\Vert \dot{\mfx} \Vert_{\XX(\HH_{0,[1,2]})}\les \Vert \xdmf_1 \Vert_{\mathcal{X}(S_1)} + \Vert \MMf_2 \Vert_{\ZZ_\MMf(S_2)} + \Vert \mfc_{\dot{\varphi}}\Vert_{\ZZ_{\CC}}.
\end{align}
Combining \eqref{eq:finalest1} and \eqref{eq:finalest2} yields 
 \begin{align*}
&\Vert \dot{\mfx} \Vert_{\XX(\HH_{0,[1,2]})}+\Vert \dot{f} \Vert_{\YY_{f^{(4)}}} + \Vert\dot{q}\Vert_{\YY_{q}}+ \Vert \dot{\PP}_{f,0}(\dot f) \Vert_{\mathcal{X}(S_2)} +\Vert \dot{\PP}_{0,q}(\dot{q}) \Vert_{\mathcal{X}(S_2)}\\
\les &\Vert \dot{\mfx}_{0,1}\Vert_{\XX^(S_{1})}+ \Vert \dot{{\MMf}}_{2}\Vert_{\XX({S}_{2})}+\Vert \mfc_{\dot{\varphi}}\Vert_{\ZZ_\CC}
\end{align*}
as required. This completes the proof of Proposition \ref{prop:lingluing} and solves the linearised $C^3$-null gluing problem.

\appendix
\addtocontents{toc}{\protect\setcounter{tocdepth}{1}}

\section{Double null coordinates}\label{sec:doublenullcoords}

\ni In this section, we outline the double null coordinate formalism used throughout the paper, see \cite{ChrFormationBlackHoles} for more details.

\subsection{Double null foliation}

We begin by constructing the double null coordinate system $(u,v,\th^1,\th^2)$ on a spacetime $(\MM,\bf g)$. Let $\D$ denote the covariant derivative and $\Rbf$ the Riemann curvature tensor of $(\MM,\g)$ and consider two optical functions defined to satisfy the \emph{eikonal equations}, $$|\D u|^2_{\g}=0,\qquad |\D v|^2_{\g}=0.$$ Locally,  $(\MM,\g)$ can be foliated by the level sets of $u$ and $v$. That is, the level sets of $u$ are outgoing null hypersurfaces denoted by $\HH$ and the level sets of $v$ are ingoing null hypersurfaces $\HHb$. The null hypersurfaces $\HH$ and $\HHb$ intersect at $2$-dimensional spacelike spheres denoted by $S$. For two real numbers $v_0>u_0$, define 
$$S_{u_0,v_0}:=\{u=u_0,v=v_0\}.$$
Similarly, $\HH_{u_0}:=\{u=u_0\}$ and $\HHb_{v_0}:=\{v=v_0\}$. The angular coordinates $(\th^1,\th^2)$ are first defined as angular coordinates on $S_{u_0,v_0}$ and extended to the entire foliation by propagating them along the generators of $\HH_{u_0}$ and  $\HHb_v$. This completes the construction of the double null coordinate system $(u,v,\th^1,\th^2)$. In this coordinate system the metric $\bf g$ can be written as 

$${\bf g} = -4 \Om^2 du dv + \gd_{AB} \left(d\th^A + b^A dv\right)\left(d\th^B + b^B dv\right)$$
where the scalar function $\Om$ is the \emph{null lapse} and the $S_{u,v}$-tangential vectorfield $b=b^A \pr_A$ is the \emph{shift vector} and by construction satisfies
$$b=0 \text{ on } \HH_{u_0}.$$
Each $S_{u,v}$ is a $2$-dimensional manifold with induced Riemannian metric $\gd$ and covariant derivative $\Nabs$. Moreover, we let $K$ denote the Gauss curvature of each $(S_{0,v},\gd)$. Using the coordinates $(\th^1,\th^2)$, define the round unit metric by 
\begin{equation}\label{eq:roundmetric}
\gac:= \left(d\th^1\right)^2 + \sin^2 \th^1 \left(d\th^2\right)^2.
\end{equation}
Moreover, we consider the conformal decomposition of $\gd$,
\begin{equation}\label{eq:confgc}
\gd=\phi^2\gd_c
\end{equation}
where the conformal factor $\phi$ and $\gd_c$ are determined by $\mathrm{det}\gd_c=\mathrm{det}\gac$.
We now decompose the connection and Riemann tensor $\Rbf$ of $(\MM,\g)$ into the \emph{Ricci coefficients} and \emph{null curvature components} as in \cite{ChrFormationBlackHoles}. First, we define the vectorfields 
\begin{equation}\label{eq:defL}
(L,\Lb)=(\Om\widehat{L},\Om\widehat{\Lb}):=(\pr_v+b^A\pr_{\th^A},\pr_u)
\end{equation}
where $(\widehat{L},\widehat{\Lb})$ are \emph{normalised} in the sense that they satisfy  $\gd\left(\widehat{L},\widehat{\Lb}\right)=-2$. Furthermore, for a general $S_{u,v}$-tangent tensorfield $W$ define $D$ and $\Db$ to be the Lie derivatives
\begin{align}
DW:= \Lied_L W, \qquad \Db W:= \Lied_\Lb W, \label{eq:defDDu}
\end{align}
where $\Lied$ denotes the projection of the Lie derivative on $(\MM,\g)$ onto the tangent space of $S_{u,v}$. For $S_{u,v}$-tangent vectorfields $X$ and $Y$, define the \emph{Ricci coefficients} by
\begin{align}
 \begin{aligned}
\chi(X,Y) :=& \g(\D_X \widehat{L},Y), & \chib(X,Y) :=& \g(\D_X \widehat{\Lb},Y), &
\zeta(X) :=& \frac{1}{2} \g(\D_X \widehat{L}, \widehat{\Lb}), \\
\eta :=& \zeta + \di \log \Om ,&\om :=& D \log \Om, & \omb :=& \Db \log \Om,
\end{aligned}\label{eq:riccicoeffs}
 \end{align}
where $\di$ denotes the extrinsic derivative of $S_{u,v}$. It holds that 
\begin{align} 
\begin{aligned} 
\underline{\zeta}=-\zeta, \,\, \etab=-\eta+2\di\log\Om.
\end{aligned} \label{eq:riccirelationetazeta} 
\end{align}

For $S_{u,v}$-tangent vectorfields $X$ and $Y$, define the \emph{null curvature components} by
\begin{align} 
\begin{aligned}
\alpha(X,Y) :=& \Rbf(X,\widehat{L}, Y, \widehat{L}), &\!\! \beta(X) :=& \frac{1}{2} \Rbf(X, \widehat{L},\widehat{\Lb},\widehat{L}), 
 & \!\! \rh :=& \frac{1}{4} \Rbf(\widehat{\Lb}, \widehat{L}, \widehat{\Lb}, \widehat{L}), \hspace{2 cm} \\
  \ab(X,Y) :=& \Rbf(X,\widehat{\Lb}, Y, \widehat{\Lb}), &\!\!  \beb(X) :=& \frac{1}{2} \Rbf(X, \widehat{\Lb},\widehat{\Lb},\widehat{L})
 , & \!\! \sigma \iin(X,Y) :=& \frac{1}{2} \Rbf(X,Y,\widehat{\Lb}, \widehat{L}), \raisetag{2\baselineskip}
 \end{aligned}\label{eq:nullcurvaturecomp}
\end{align}
where $ \iin$ denotes the area form of $(S_{u,v},\gd)$. Finally, consider the splitting of $\chi$ and $\chib$ into a trace part (with respect to $\gd$) and tracefree part denoted by $\chih$ and $\chibh$,
$$\chi=\chih+\frac{1}{2}\tr\chi\gd, \quad  \chib=\chibh+\frac{1}{2}\tr\chib\gd.$$

\subsection{Null structure equations}

Writing the vacuum Einstein field equations \eqref{eq:EVE} in terms of the Ricci coefficients and null curvature components imply \emph{null structure equations}. We first introduce notation as in chapter 1 of \cite{ChrFormationBlackHoles}.
 
For two $S_{u,v}$-tangential $1$-forms $X$, and $Y$, we define
 \begin{align}
\begin{aligned}
 \Divd X :=& \Nabs^A X_A, &  (X,Y):=& \gd(X,Y), \\
 \Curld X :=& \in^{AB}\Nabs_A X_B, & (X \widehat{\otimes} Y)_{AB} :=& X_A Y_B + X_B Y_A - (X \cdot Y)\gd_{AB},\\
 ({}^* X)_A :=& \in_{AB}X^B,&  (\Nabs \widehat{\otimes} Y)_{AB} :=& \Nabs_A Y_B + \Nabs_B Y_A - (\Divd Y)\gd_{AB}
\end{aligned}
\end{align}
 Similarly, for a symmetric $S_{u,v}$-tangential $2$-tensor $V$ let the divergence, trace and the tracefree part of $V$ be
\begin{align*}
\begin{aligned}
\left(\Divd V\right)_A := \Nabs^B V_{BA,}
\quad \tr V := \gd^{AB} V_{AB},
\quad \widehat{V} := V - \frac{1}{2} \tr V \gd,
\end{aligned}
\end{align*}
and for a symmetric $S_{u,v}$-tangential $2$-tensor $V$ and a $1$-form $X$,
\begin{align*} 
\begin{aligned} 
(V \cdot X)_A := V_{AB}X^B.
\end{aligned} 
\end{align*}
For a symmetric $S_{u,v}$-tangential tensor $W$, let $\widehat{D}W$ and $\widehat{\Db}W$ denote the tracefree parts of $DW$ and $\Db W$ with respect to $\gd$, respectively.

We now state the subset of the null structure equations associated with the transport equations along an outgoing null hypersurface $\HH$.
We have the \emph{first variation equation},
\begin{align}  \label{eq:firstvariationeq}
D \gd =& 2 \Om \chi,
\end{align}
which, by the decomposition \eqref{eq:confgc}, implies that
\begin{align} 
D \phi = \frac{\Om \tr \chi \phi}{2}, 
\label{eq:relationtoDphi}
\end{align}
the \emph{Raychauduri equations}, 
\begin{align}
D \trchi +\frac{\Om}{2} (\trchi)^2 -\om\tr\chi =& - \Om \vert \chih \vert^2_{\gd},
\label{eq:raychauduri}\end{align}
and the following null transport equations for the Ricci coefficients
\begin{align} \begin{aligned}
D\chih =& \Om \vert \chih \vert^2 \gd + \om \chih - \Om \a, \\
D \eta =& \Om (\chi \cdot \etab - \beta),  \\
D \omb =& \Om^2(2 (\eta, \etab) - \vert \eta \vert^2 -\rh), \\
D\etab =& - \Om (\chi \cdot \etab -\be) + 2 \di \om, 
\end{aligned} \label{eq:transporteqriccicoeffs}\end{align}
as well as the transport equation for $\Db\omb$
\begin{align} 
\begin{aligned} 
D\Db\omb =& -12\Om^{2}(\et-\di\log\Om,\di\omb)+2\Om^{2}\om((\et,-3\et+4\di\log\Om)-\rh)\\
&+12\Om^{3}\chib(\et,\et-\di\log\Om)+\Om^{3}(\beb,7\et-3\di\log\Om)\\
&+\frac{3}{2}\Om^{3}\tr\chib\rh+\Om^{3}\Divd\bb+\frac{\Om^{3}}{2}(\chih,\ab)
\end{aligned} \label{eq:DDuom}
\end{align}
Furthermore, we have the \emph{Gauss equation}, 
\begin{align} \label{eq:gauss}
K + \frac{1}{4} \tr \chi \tr \chib - \frac{1}{2} (\chih,\chibh) = - \rh,
\end{align}
where $K$ denotes the Gauss curvature of $S_{u,v}$ and the \emph{Gauss-Codazzi equations}
\begin{align} 
\begin{aligned}
\Divd \chih -\frac{1}{2} \di \tr \chi + \chih \cdot \zeta - \frac{1}{2} \trchi \zeta =& - \beta,\\
\Divd \chibh - \frac{1}{2} \di \trchib -\chibh \cdot \zeta +\frac{1}{2} \trchib \zeta =& \beb.
\end{aligned}\label{eq:gausscodazzi}
\end{align}

We have the following higher-order null structure equation for $\Dbh\ab$ which follows from equation (12.214) in \cite{ChrFormationBlackHoles}.
 \begin{align}
\begin{aligned}
\Dh\Dbh\ab =& \frac{1}{2}\Om\tr\chi\widehat{\Db}\ab-\frac{1}{4}(\Om\tr\chi)(\Om\tr\chib)\ab+\Om \Nabs\widehat{\otimes} \left( \frac{3}{2} \Om \tr \chib \beb + \Om \Divd \aa\right)\\
& -2\om\widehat{\Db}\ab-4\widehat{\LIE}_{\Om^{2}\ze}\ab+\frac{1}{2}\Om\chibh(\Om\chih,\ab)-\frac{1}{2}\Om\chih(\Om\chibh,\ab)\\
&+2\Om^{2}(2(\et,\et-2\di\log\Om)+|\et-2\di\log\Om|^{2}+\rho)\ab\\
&+ \frac{1}{2}\Om^{2}\ab(2\Divd\et+2|\et|^{2}-(\chih,\chibh)+2\rh)\\
&-\Om\omb ( \Nabs \widehat{\otimes} \beb - (5\et-9\di\log\Om) \widehat{\otimes} \beb + 3 \chibh \rh -3 {}^* \chibh \si) \\
&+\Om \bigg\{ \Nabs\widehat{\otimes} (\Om \chibh \cdot \beb +\omb \beb +\Om (\et-2\di\log\Om) \cdot\aa)- 2\Om\chibh\Divd\bb\\
&- 2\widehat{\Db\Gammad}\cdot\bb+(5\Om(\chib\cdot\et+\bb)+\di\om)\widehat{\otimes}\bb\\
&+(5\et-9\di\log\Om)\widehat{\otimes} \left(\frac{3}{2} \Om \tr \chib \beb - \Om \chibh \cdot \beb -\omb \beb + \Om \left( \Divd \aa - (\et-2\di\log\Om) \cdot \aa\right)\right) \\
&+2\Om(\bb,5\et-9\di\log\Om)\chibh+3(\om\chibh-\Om\ab)\rh-3(\om{}^{*}\chibh-\Om{}^{*}\ab)\si\\
&+3\chibh\left(- \frac{3}{2} \Om \tr \chib \rh- \Om \left( \Divd \beb + (\et+\di\log\Om, \beb)+ \frac{1}{2} (\chih,\aa) \right)\right)\\
&-3^{*}\chibh\left(- \frac{3}{2} \Om \tr \chib \si - \Om \left( \Curld \beb + (\et+\di\log\Om, {}^* \beb) + \frac{1}{2}\chih \wedge \aa \right)\right)\bigg\}, \raisetag{9\baselineskip}
\end{aligned}\label{eq:Duhab}
\end{align}
where $\widehat{\Db\Gammad}$ is the trace-free part of $\Db\Gammad$ with
\begin{equation*}
(\Db\Gammad)^C_{AB}=(\gd^{-1})^{CD}(\Nabs_A(\Omega\chib)_{BD}+\Nabs_B(\Omega\chib)_{AD}-\Nabs_D(\Omega\chib)_{AB}).
\end{equation*}
We also have the following higher-order null structure equation for $\Db^2\omb$ which follows from the transport equations \eqref{eq:transporteqriccicoeffs} and \eqref{eq:DDuom}.
\begin{align}
 \begin{aligned}
 D\Db^{2}\omb=& -4\Om^{3}({\Om\tr\chib}){\Divd\bb}-3\Om^{2}(\Om\tr\chib)^{2}{\rh}-\Om^{3}\Divd(\Om{\Divd\ab})-12\Om^{2}|\di\omb|^{2}\\
 &-7\Om^{4}|\bb|^{2} -16\Om^{2}(\et-\di\log\Om,\di\Db\omb)-40\Om^{2}\omb(\et-\di\log\Om,\di\omb)\\
 &+24\Om^{2}\chib(3\et-2\di\log\Om,\di\omb)+2\Om^{2}(\omb^{2}+\Db\omb)((\et,-3\et+4\di\log\Om)-\rh)\\
 &+\Om^{3}\left(\beb,12\di\omb-\frac{3}{2}\di\Om\tr\chib\right)+72\Om^{3}\omb\chib(\et,\et-\di\log\Om)+6\Om^{3}\omb{\Divd\bb}\\
 &+6\Om^{3}\omb(\beb,7\et-3\di\log\Om)-\Om^{4}\chib(\bb,38\et-15\di\log\Om)+\frac{1}{2} \Om^{3}(\chibh,\Db\ab) \\
 &-12\Om^{4}(\chib\times\chib)(\et,5\et-4\di\log\Om)-\frac{1}{2}\Om^{4}\tr\chib(\bb,31\et-9\di\log\Om)\\
 &+\Om^{2}\left(9\omb\Om\tr\chib-\frac{3}{2}\Om^{2}|\chih|^{2}\right){\rh}-12\Om^{4}\ab(\et,\et-\di\log\Om)\\
 &-\Om^{4}(\Divd\ab-(\et-2\di\log\Om)\ab,7\et-3\di\log\Om)\\
 &-\Om^{3}\Divd(\Om(\chibh\cdot\bb-(\et-2\di\log\Om)\cdot\ab))-2\Om^{4}(\chib,(\chih\times\ab))\\
 &+\frac{1}{2}\Om^{3}\left(4\omb\chih-\Om\tr\chib\chih+\Om(\Nabs\widehat{\otimes}\et+\et\widehat{\otimes}\et)-\frac{\Om}{2}\tr\chi\chibh,\ab\right).
  \end{aligned}\label{eq:Du2omb}
 \end{align}

Along with the null structure equation there exists the null Bianchi equations, see Proposition 1.2 of \cite{ChrFormationBlackHoles}, the only relevant null Bianchi equation to this paper is 
\begin{equation}\label{eq:nullBianchiab}
\widehat{D} \aa - \frac{1}{2} \Om \tr \chi \aa + 2 \om \aa + \Om \left( \Nabs \widehat{\otimes} \beb + (4 \etab - \zeta) \widehat{\otimes} \beb + 3 \chibh \rh -3 {}^* \chibh \si \right) =0.
\end{equation}

\section{The linearisation of null structure equations and conservation laws}\label{sec:linsec}

\subsection{The linearised hierarchy of null structure equations}\label{sec:linCC}

In this section, we present the explicit formulas for the linearised hierarchy of transport equations from Section \ref{sec:charseed} first given in \cite{ACR2}. Using the linearisation procedure, \eqref{eq:formallin}, we have the following linearisations.
\begin{align*} 
\begin{aligned} 
\dot{\gd} = 2r \phid \gac + v^2 \gdcd, \,\, \omd = D\Omd, \,\, \ombd = \Db\Omd, \,\, \etabd =& - \etad + 2 \di \Omd.
\end{aligned}
\end{align*}
By equation (242) in \cite{DHR} the linearisation of the Gauss curvature is
\begin{align*}
\Kd =& \frac{1}{2v^2} \Divdo \Divdo \gdcd - \frac{1}{v^3} (\Ldo +2) \phid.
\end{align*}

The linearised hierarchy of null structure equations is as follows.

\begin{align} 
\begin{aligned} 
\dot{\CC}_\phid=& D^2 \phid - 2 \omd = D(D\phid-2\Omd), \\
\dot{\CC}_{\omtrchid}=& v^2 \left(2D\left(\frac{\phid}{v}\right)-\omtrchid\right) ,\\
\dot{\CC}_\chihd=& v^2 D\gdcd - 2 \chihd, \\
\dot{\CC}_\etad=& \frac{1}{v^2} D\left(v^2 \etad\right) - \frac{4}{v} \di \Omd - \frac{1}{v^2} \Divdo \chihd + \frac{1}{2} \di \omtrchid,\\
\dot{\CC}_{\omtrchibd}=& \frac{1}{v^2} D\left(v^2 \omtrchibd\right) - \frac{2}{v} \omtrchid + \frac{2}{v^2} \Divdo \left(\etad-2\di \Omd\right) +2 \Kd + \frac{4}{v^2} \Omd,\\
\dot{\CC}_\chibhd=& v D\left(\frac{\chibhd}{v}\right) - 2 \DDd_2^* \left(\etad-2\di\Omd\right) -\frac{1}{v} \chihd ,\\
\dot{\CC}_\ombd=& D\ombd - \Kd -\frac{1}{2v}\omtrchibd + \frac{1}{2v} \omtrchid - \frac{2}{v^2} \Omd, \\
\dot{\CC}_\abd=& v D\left(\frac{\abd}{v}\right) - 2 \DDd_2^* \left(\frac{1}{v^2} \Divdo \chibhd - \frac{1}{2} \di \omtrchibd - \frac{1}{v} \etad\right), \\
\dot{\CC}_{\Db\ombd}=& D \left(\Db \ombd\right) -\frac{3}{v} \left(\Kd +\frac{1}{2v}\omtrchibd - \frac{1}{2v} \omtrchid+ \frac{2}{v^2} \Omd\right)\\
&- \frac{1}{v^2} \Divdo \left(\frac{1}{v^2} \Divdo \chibhd - \frac{1}{2} \di \omtrchibd - \frac{1}{v} \etad\right),\\
\dot{\CC}_{\a}=& \ad + D\chihd.\\
\end{aligned}\label{eq:linearisedconstraints}
\end{align}

\subsection{Conserved charges from the linearised $C^2$-null gluing problem}\label{sec:C2QQ}

In this section, we present the conserved charges appearing in the analysis of the linearised $C^2$-null gluing problem. These charges are conserved along $\HH_{[1,2]}$ and are useful in the derivation of representation formulas for higher-order quantities in Section \ref{sec:C3repformulas}. See Section 4 \cite{ACR2} for their derivation and analysis.

\begin{align} 
\begin{aligned} 
\QQ_{\etad^{[1]}}:=& v^2 \etad^{[1]} + \frac{v^3}{2}\di \left(\omtrchid^{[1]}-\frac{4}{v}\Omd^{[1]}\right), \\
\QQ_\phid :=&\frac{v}{2} \left(\omtrchid-\frac{4}{v}\Omd\right) + \frac{\phid}{v}, \\
\QQ_{\omtrchibd} :=& v^2 \omtrchibd -\frac{2}{v}\Divdo \left(v^2\etad+\frac{v^3}{2}\di\left(\omtrchid-\frac{4}{v}\Omd\right)\right) \\
&-v^2 \left(\omtrchid-\frac{4}{v}\Omd\right) +2v^3 \Kd, \\
\QQ_\chibhd :=& \frac{\chibhd}{v} -\frac{1}{2} \left( \DDd_2^* \Divdo +1\right) \gdcd + \DDd_2^* \left( \etad + \frac{v}{2}\di \left(\omtrchid-\frac{4}{v}\Omd\right)\right) \\
&- v \DDd_2^* \di \left(\omtrchid-\frac{4}{v}\Omd\right), \\
\QQ_{\abd_\psi} :=& \frac{\abd_\psi}{v} + 2 \DDd_2^* \left(\frac{1}{v^2} \Divdo \chibhd - \frac{1}{v} \etad- \frac{1}{2} \di \omtrchid + \DDd_1^* \left(\ombd,0\right)\right)_{\psi},\\
\QQ_{\ombd^{[\leq1]}}:=& \ombd^{[\leq1]} +\frac{1}{4v^2}\QQ_{\omtrchibd}^{[\leq1]} +\frac{1}{3v^3} \Divdo\QQ_{\etad^{[1]}}, \\
\QQ_{\Db\ombd^{[\leq1]}} :=& \Db\ombd^{[\leq1]} - \frac{1}{6v^3} (\Ldo-3) \QQ_{\omtrchibd}^{[\leq1]} +\frac{1}{v^4} \Divdo \QQ_{\etad^{[1]}}, \\
\QQ_{\Db\ombd^{[2]}} :=& \Db\ombd^{[2]} +\frac{3}{2v^3} \QQ_{\omtrchibd}^{[2]}+ \frac{1}{2v^2} \Divdo \Divdo \QQ_\chibhd^{[2]} -\frac{12}{v^2} \QQ_\phid^{[2]} \\
&+ \frac{3}{2v^2} \Divdo \left(\eta+ \frac{v}{2}\di\left(\omtrchid-\frac{4}{v}\Omd\right) \right)^{[2]} - {\frac{3}{4v^2} \Divdo \Divdo \gdcd }^{[2]}.
\end{aligned} \label{eq:C2charges}
\end{align}

\section{Setup of transversal sphere perturbations}\label{sec:transversalspherepert}

We briefly outline the setup of the transversal sphere perturbations here, see also Appendix A in \cite{ACR2}. Let $\SSt$ be a two dimensional sphere in $(\MM, \g)$ with sphere data $\tilde x_{0,2}$. Let $(\ut,\vt,\tht^{1},\tht^{2})$ be a local double null coordinate system about $\SSt$ with $\ut$ and $\vt$ such that $\SSt=\{\ut=0,\vt=2\}$. The metric takes the form
\begin{align*} 
\g = - 4 \tilde{\Om}^2 d\tilde{u} d\tilde{v} + \tilde{\gd}_{CD} (d\tilde{\th}^C - \tilde{b}^C d\tilde{v})(d\tilde{\th}^D - \tilde{b}^D d\tilde{v}).
\end{align*}
  The null geometry then follows along the same lines as Appendix \ref{sec:doublenullcoords}. We define the following vectorfields in the $(\ut,\vt,\tht^{1},\tht^{2})$ coordinates.
\begin{itemize}
\item The \emph{normalised null vectorfields} are defined by $\widehat{\Lbt}:=-2\Omt\D\vt$ and $\widehat{\Lt}:=-2\Omt\D\ut$,
\item The \emph{equivariant null vectorfields} are defined by $\Lbt:=-2\Omt^{2}\D\vt$ and $\Lt:=-2\Omt^{2}\D\ut$.
\end{itemize}
The equivariant null vectorfields can be expressed in the coordinate basis as
\begin{align*}\begin{aligned}
\Lt=\pr_{\vt}+\tilde{b}^{A}\pr_{\tht^{A}}, \, \, \Lbt=\pr_{\ut}.
\end{aligned}\end{align*}
The Ricci coefficients and null curvature components are then defined in the usual way
\begin{align*} 
\begin{aligned} 
\tilde{\chi}_{AB} :=& \g(\D_{\tilde A} \widehat{\tilde L},\pr_{\tilde{B}}), & \tilde{\chib}_{AB} :=& \g(\D_{\tilde{A}} \widehat{\tilde \Lb},\pr_{\tilde{B}}), & \tilde{\zeta}_A :=& \frac{1}{2} \g(\D_{\tilde{A}} \widehat{\tilde L},\widehat{\tilde \Lb}), \\
 \tilde\eta :=& \tilde \zeta + \tilde\di \log \tilde\Om, & \tilde\om :=& \tilde L \log \tilde \Om, & \tilde{\omb} :=& \tilde{\Lb} \log \tilde\Om,\\
\tilde{\alpha}_{AB} :=& \Rbf(\pr_{\tilde A},\widehat{\Lt}, \pr_{\tilde B}, \widehat{\Lt}), & \tilde\beta_A :=& \frac{1}{2} \Rbf(\pr_{\tilde A}, \widehat{\Lt},\widehat{\Lbt},\widehat{\Lt}), & \tilde{\rh} :=& \frac{1}{4} \Rbf(\widehat{\Lbt}, \widehat{\Lt}, \widehat{\Lbt}, \widehat{\Lt}), \\
 \tilde{\sigma} \tilde{\iin}_{AB} :=& \frac{1}{2} \Rbf(\pr_{\tilde A},\pr_{\tilde B},\widehat{\Lbt}, \widehat{\Lt}), & 
\tilde\beb_A :=& \frac{1}{2} \Rbf(\pr_{\tilde A}, \widehat{\Lbt},\widehat{\Lbt},\widehat{\Lt}), & \tilde\ab_{AB} :=& \Rbf(\pr_{\tilde A},\widehat{\Lbt}, \pr_{\tilde B}, \widehat{\Lbt}).
\end{aligned}
\end{align*}
where $\tilde{\di}$ and $\tilde{\iin}_{AB}$ denote the exterior derivative on and area element of spheres $\tilde{S}_{u,v}$, respectively.

Now, we make the transformation $\ut\rightarrow u$ on the level set $\tilde\HHb_{2}:=\{\vt=2\}$. This is done through the introduction of a scalar function $f(u,\th^{1},\th^{2})$ which will become our transversal perturbation function. Define $(u,v,\th^{1},\th^{2})$ on $\tilde\HHb_{2}$ by
\begin{align} 
\begin{aligned} 
\tilde{u}=u+f(u,\th^1,\th^2), \,\, \vt=v, \,\, \tht^1=\th^1, \,\, \tht^2=\th^2.
\end{aligned} \label{eq:defucoord}
\end{align}
For $f$ sufficiently small $(u,v=2,\th^{1},\th^{2})$ is a coordinate system on $\tilde\HHb_{2}$. We have the following relations, proved in \cite{ACR2}.
\begin{align}\begin{aligned}
\pr_{u} =& (1+\pr_u f) \pr_{\ut}, &
\pr_{\th^A} =& \pr_{\tht^A} + (\pr_{\th^A} f) \pr_{\tilde u},\\ 
\Om^2 =& \tilde{\Om}^2 (1+\pr_u f), &
\gd^{AB} =& \tilde{\gd}^{AB}.
\end{aligned}\label{eq:relationsforu}
\end{align}
Define $\Lb:=\pr_{u}$ and let $L$ be the unique vectorfield such that
\begin{align*} 
\begin{aligned} 
\g(L,\Lb) = -2 \Om^2, \,\, \g(L,\pr_{\th^1})=\g(L,\pr_{\th^2})=0.
\end{aligned}
\end{align*}
Finally, define $\widehat{L} := \Om^{-1} L$ and $\widehat{\Lb} := \Om^{-1} \Lb$.

\subsection{Ricci coefficients and null curvature components on $\tilde\HHb_{2}$} \label{sec:perturbationrelations} Using the vectorfields $(\widehat{L},\widehat{\Lb})$, we have the Ricci coefficients and null curvature components as defined in Section \ref{sec:doublenullcoords}, $(\chi, \chib, \ze, \eta, \om, D\om, \omb, \Db\omb,\\ D^{2}\om,\Db^{2}\omb)$. Here, we explore how they transform under the coordinate transform \eqref{eq:defucoord}. Transformation formulae for the quantities $(\chi, \chib, \ze, \eta, \om, D\om, \omb, \Db\omb)$ can be found in Appendix A of \cite{ACR2}. Here, we derive the formulae for $D^{2}\om$ and $\Db^{2}\omb$. \\

{\bf Gauge transformation of $\Db^{2}\omb$.} Recall by definition that $\omb=\Lb \log \Om$, then
\begin{align}\begin{aligned}
\Db^{2}\omb=&\pr_{u}^{2}(\pr_{u}\log\Om)\\
=&\tilde\Db^{2}\tilde\omb(1+\pr_{u}f)^{3}+3\tilde\Db\tilde\omb(1+\pr_{u}f)\pr_{u}^{2}f+\tilde\omb\pr_{u}^{3}f \\
&-\frac{3(\pr_{u}^{3}f)(\pr_{u}^{2}f)}{2(1+\pr_{u}f)^{2}} + \frac{(\pr_{u}^{2}f)^{3}}{(1+\pr_{u}f)^{3}}+\frac{\pr_{u}^{4}f}{2(1+\pr_{u}f)}.
\end{aligned}\label{eq:Du2ombcoordu}\end{align}

{\bf Gauge transformation of $D^{2}\om$.} Firstly, recall that the vector $L$ under the coordinate transform \eqref{eq:defucoord} is given by
\begin{equation*}
L = \left( \tilde{\Om}^2\vert \Nabs f \vert_\gd^2 \right) \Lbt+ \Lt + \left(2\tilde{\Om}^2 \tilde{\gd}^{AC}\pr_C f \right) \pr_{\tht^A},
\end{equation*}
the coordinate $v=\vt$ satisfies the equation $L(\vt)=1$ or equivalently, $L'(\vt)=\Om^{-2}$. Then
\begin{align}\begin{aligned}
D^{2}\om = L(L(L\log\Om)) =&\Om^{2}L'\left(\Om^{2}L' \left(-\frac{1}{2}\Om^{4}L'\left(L'\left(\vt\right)\right)\right)\right)\\
=&\Om^{2}L'\left(\left(\Om^{4}L'\left(L'\left(\vt\right)\right)\right)^{2}-\frac{\Om^{6}}{2}L'\left(L'\left( L'\left(\vt\right)\right)\right)\right)\\
=& 14\om D\om - 56\om^{3}-16\Om^{4}\om^{2}-\frac{1}{2}\D_{L}\D_{L}\D_{L}\D_{L}\vt,
\end{aligned}\label{eq:D2ombcoordu}\end{align}
where we have used the geodesic equation $\D_{L'}L'=0$.
We now analyse the null curvature components $\{\a,\widehat{D}\a,\ab,\widehat{\Db}\ab\}$.\\

{\bf Gauge transformation of $\ab$.} By definition,
\begin{align}\begin{aligned}
\ab_{AB}:=&\Rbf(\pr_{ A},\widehat{\Lb}, \pr_{B}, \widehat{\Lb})\\
=&\Rbf( \pr_{\tht^A} + (\pr_{\th^A} f) \pr_{\tilde u},\Om^{-1}\Lb, \pr_{\tht^B} + (\pr_{\th^B} f) \pr_{\tilde u},\Om^{-1}\Lb)\\
=&(1+\pr_u f)\Rbf(\pr_{\tht^A},\widehat{\Lbt},\pr_{\tht^B},\widehat{\Lbt})\\
=&(1+\pr_u f)\tilde\ab_{AB},
\end{aligned}\label{eq:abcoordu}\end{align}
where we have used that $\Rbf$ is antisymmetric in its first two arguments.\\

{\bf Gauge transformation of $\Dbh\ab$.} Using the identity $\Dbh\ab = \Db\ab-\Om(\chibh,\ab)\gd$. Apply the transformation formulae for $\ab$ and $\chibh$, we have that
\begin{align}\begin{aligned}
\Dbh\ab =& \pr_{u}\ab - \Om(\chibh,\ab)\gd\\
=&(1+\pr_{u}f)^{2}\pr_{\ut}\tilde\ab+(\pr_{u}^{2}f)\tilde\ab-\tilde\Om(1+\pr_{u}f)^{2}(\tilde\chibh,\tilde\ab)\tilde\gd\\
=&(1+\pr_{u}f)^{2}\tilde\Dbh\tilde\ab+(\pr_{u}^{2}f)\tilde\ab.
\end{aligned}\label{eq:Duabcoordu}\end{align}
Where we have used \eqref{eq:relationsforu} and that $\chibh$ transforms as
\begin{equation*}
\chibh_{AB} = \frac{\tilde{\Om}}{\Om} (1+\pr_u f) \tilde{\chibh}_{AB}.
\end{equation*}
 For transversal perturbations, we have that the linearisations are $\Db^{2}\ombd = \partial_{u}^{3}(\frac{1}{2}\partial_{u}\dot{f})$ taking the $u$-derivative of the expression for $\Db\ombd$. Since $\abd=0$, $\Db\abd= 0$. 
 
 \begin{remark}
 The two remaining curvature components to consider are $\a,\Dh\a$.  At most, they depend on $\pr_{u}^{2}f$ and both of their linearisations vanish.
 \end{remark}

\section{Fourier analysis on the round unit sphere}\label{sec:sphharm}

In this section, we define Hodge operators on spheres as defined in \cite{ChKl14} and spherical harmonic decompositions. We also analyse the operators appearing in Section \ref{sec:conservationlaws}.

\subsection{Hodge operators and spherical harmonics}\label{sec:defhodge}

We first define the following Hodge operators on sphere $S_{u,v}$.
\begin{definition}[Hodge operators]\label{def:Hodgeoperators} Let $(S, \gd)$ be a $2$-sphere. Let $V$ be a $1$-form and $W$ a $2$-tensor then define the operators
\begin{align*}
\DDd_1(V) := (\Divd V, \Curld V), \qquad 
\DDd_2(W) := \Divd W.
\end{align*}
 Their adjoints are given by
\begin{align}
\begin{aligned}
\DDd_1^{*}(f_1,f_2):= -\di f_1 + {}^*\di f_2, \qquad 
\DDd_2^{*}(V)_{AB} := -\frac{1}{2} \left( \Nabs_A V_B + \Nabs_B V_A - (\Divd V) \gd_{AB} \right).
\end{aligned}\label{eq:hodgeops}
\end{align}
\end{definition}

Using the Hodge operators, we define spherical harmonics on the round unit sphere as follows, see e.g. \cite{Cz16}.

\begin{definition}[Spherical harmonics] The following are defined with respect to the coordinates  $(\th^1,\th^2)$ on the round unit sphere.
\begin{itemize}
\item Denote by $Y^{(lm)}$ be the standard (real-valued) spherical harmonics on the round unit sphere for integers $l\geq0$, $-l \leq m \leq l$,
\item For $l\geq1$, $-l \leq m \leq l$, define the \emph{electric} and \emph{magnetic} vectorfields by
\begin{align*} 
E^{(lm)} := \frac{1}{\sqrt{l(l+1)}} \DDd_1^* (Y^{(lm)},0), \,\, H^{(lm)} :=\frac{1}{\sqrt{l(l+1)}} \DDd_1^* (0, Y^{(lm)}).
\end{align*}
\item For $l\geq2$, $-l \leq m \leq l$, define the  \emph{electric} and \emph{magnetic} tracefree symmetric $2$-tensors by
\begin{align*} 
\psi^{(lm)}:= \frac{1}{\sqrt{\frac{1}{2} l(l+1)-1}} \DDd_2^* \left(E^{(lm)}\right), \,\, \phi^{(lm)} := \frac{1}{\sqrt{\frac{1}{2} l(l+1)-1}} \DDd_2^* \left(H^{(lm)}\right).
\end{align*}
\end{itemize}
\end{definition}

\ni \textbf{ Notation.} Throughout this paper, we make use of the following notation.
 On the round unit sphere, the decomposition of $L^2$-integrable functions $f$, vectorfields $V$ and tracefree symmetric $2$-tensors $W$ is expressed as follows.
\begin{align*}
f =& \sum\limits_{l\geq0} \sum\limits_{-l\leq m \leq l} f^{lm} Y^{(lm)}, \\
V=&  \sum\limits_{l\geq1} \sum\limits_{-l\leq m \leq l} V^{lm}_E E^{(lm)} + V^{lm}_H H^{(lm)},\\
W =&  \sum\limits_{l\geq2} \sum\limits_{-l\leq m \leq l} W^{lm}_\psi \psi^{(lm)} +W^{lm}_\phi \phi^{(lm)},
\end{align*}
where 
\begin{align*} 
\begin{aligned} 
f^{(lm)} :=& \int_{\SSS^2} f Y^{(lm)} d\mu_{\gac}, &&\\
V_E^{(lm)} :=& \int_{\SSS^2} V \cdot E^{(lm)} d\mu_{\gac}, &
V_H^{(lm)} :=& \int_{\SSS^2} V \cdot H^{(lm)} d\mu_{\gac}, \\
W_\psi^{(lm)} :=& \int_{\SSS^2} W \cdot \psi^{(lm)} d\mu_{\gac}, &
W_\phi^{(lm)} :=& \int_{\SSS^2} W \cdot \phi^{(lm)} d\mu_{\gac}.
\end{aligned} 
\end{align*}
$d\mu_\gac$ is the volume element of the round unit metric on $\SSS^2$ and $\cdot$ denotes the inner product with respect to $\gac$. Moreover, for a scalar function $f$ we denote, for integers $l' \geq0$,
\begin{align*} 
\begin{aligned} 
f^{[l']} = \sum\limits_{l=l'} \sum\limits_{-l\leq m \leq l} f^{lm} Y^{(lm)}, \,\, f^{[\geq l']} = \sum\limits_{l\geq l'} \sum\limits_{-l\leq m \leq l} f^{lm} Y^{(lm)}, \,\,f^{[\leq l']} = \sum\limits_{0\leq l\leq l'} \sum\limits_{-l\leq m \leq l} f^{lm} Y^{(lm)}
\end{aligned} 
\end{align*}
and similarly for $V$ and $W$.\\

\begin{lemma}[Decomposition of differential operators] We use spherical harmonics to  decompose the following operators acting on scalar fields $f$, vectorfields $V$ and tracefree symmetric $2$-tensors $W$ .
For $l\geq1$, 
\begin{align}
 \begin{aligned}
(\di f)_E^{(lm)}  =& - \sqrt{l(l+1)} f^{(lm)}, &(\di f)_H^{(lm)}  =& 0,\\
(\DDd_1^*(0,f))_E^{(lm)}  =& 0, &
(\DDd_1^*(0,f))_H^{(lm)}  =& \sqrt{l(l+1)} f^{(lm)},\\
(\Divdo V)^{(lm)} =& \sqrt{l(l+1)} V_E^{(lm)}, &&
\end{aligned}
\end{align}
and for $l\geq2$,
\begin{align} 
\begin{aligned}
\DDd_2^*(V)_\psi^{(lm)} =& \sqrt{\frac{1}{2} l(l+1)-1} \, V_E^{(lm)}, & \DDd_2^*(V)_\phi^{(lm)} =& \sqrt{\frac{1}{2} l(l+1)-1} \, V_H^{(lm)}, \\
(\Divdo W)_E^{(lm)}=& \sqrt{\frac{1}{2} l(l+1)-1} \, W_\psi^{(lm)}, &
(\Divdo W)_H^{(lm)}=& \sqrt{\frac{1}{2} l(l+1)-1} \, W_\phi^{(lm)}.\\
\end{aligned}
\end{align}
\end{lemma}

\subsection{Analysis of elliptic operators}\label{sec:sphharmanalysis}

In this section, we use the spherical harmonic decomposition outlined in Section \ref{sec:defhodge} in order to analyse the kernel of specific elliptic operators that appear in Section \ref{sec:conservationlaws}.
Consider first the operator on the right-hand side of the representation formula \eqref{eq:repformulaforDbab}
\begin{equation*}
 \left(2-\Divdo\DD_2^*\right).
\end{equation*}
Let $V$ be a vectorfield with 
$$V=  \sum\limits_{l\geq1} \sum\limits_{-l\leq m \leq l} V^{lm}_E E^{(lm)} + V^{lm}_H H^{(lm)},\\$$

 Then the we can decompose the operator $2-\Divdo\DD_2^*$ acting on $X$ as
 \begin{align*}
 \left(2-\Divdo\DD_2^*\right)V =&  \sum\limits_{l\geq1} \sum\limits_{-l\leq m \leq l}\left(2-\left(\frac{1}{2} l(l+1)-1\right)\right)\left( V^{lm}_E E^{(lm)} + X^{lm}_H H^{(lm)}\right)\\
  =& \sum\limits_{l\geq1} \sum\limits_{-l\leq m \leq l}\left(3-\frac{1}{2} l(l+1)\right)\left( V^{lm}_E E^{(lm)} + V^{lm}_H H^{(lm)}\right).
 \end{align*}
Hence, when $l=2$ the righthand side vanishes and we have that the kernel of the operator, $2-\Divdo\DD_2^*$ consists of the vectorfields 
\begin{equation*}
\{V:V=V^{[2]}\}.
\end{equation*}
Now consider the operator on the right-hand side of the representation formula \eqref{eq:repDu2ombD}, 
\begin{equation*}
(8-\Divdo\DDd_{2}^{*})
\end{equation*}
acting on a vectorfield $V$. Then
 \begin{align*}
 \left(8-\Divdo\DD_2^*\right)V =&  \sum\limits_{l\geq1} \sum\limits_{-l\leq m \leq l}\left(8-\left(\frac{1}{2} l(l+1)-1\right)\right)\left( V^{lm}_E E^{(lm)} + X^{lm}_H H^{(lm)}\right)\\
  =& \sum\limits_{l\geq1} \sum\limits_{-l\leq m \leq l}\left(9-\frac{1}{2} l(l+1)\right)\left( V^{lm}_E E^{(lm)} + V^{lm}_H H^{(lm)}\right).
 \end{align*}
 Since $(9-\frac{1}{2} l(l+1))\neq0$ for integers $l$, the operator has no kernel.

\bibliographystyle{abbrv}

\end{document}